\documentclass[10pt]{article}
\usepackage{geometry}    

\usepackage{authblk}
\usepackage{amsmath, amsthm, amssymb, latexsym, mathtools}
\usepackage{hyperref}
\theoremstyle{definition}
\newtheorem{definition}{Definition}
\newtheorem{theorem}{Theorem}
\newtheorem{lemma}{Lemma}
\newtheorem{corollary}{Corollary} 
\newtheorem{proposition}{Proposition} 
\newtheorem{example}{Example}
\theoremstyle{remark}
\newtheorem*{remark}{Remark}
\DeclareMathOperator{\tr}{Tr}

\DeclareMathOperator{\res}{Res}
\DeclareMathOperator{\rescl}{Res^{cl}}

\DeclareMathOperator{\id}{id}
\DeclareMathOperator{\spn}{span}
\numberwithin{equation}{section}

\usepackage{pdfpages}
\usepackage[makeroom]{cancel}
\usepackage{tikz, float}
\usepackage{lscape}
\usepackage{array}
\usetikzlibrary{decorations.markings}
\usetikzlibrary{shapes.geometric}
\usetikzlibrary{arrows}
\usepackage[bottom]{footmisc}

\begin{document}

\author[1]{Yana Staneva}
\affil[1]{Mathematical Institute \protect\\ University of Cologne \protect\\ 50931 Cologne, Germany \protect\\ \texttt{ystaneva@math.uni-koeln.de, \texttt{yanastaneva8@gmail.com} \protect\\ \texttt{mi.uni-koeln.de/$\sim$ystaneva/, github.com/yanastaneva8}}}
\title{Explicit Computations for the Classical and Quantum Integrability of the 3-Dimensional Rational Calogero-Moser System}
\maketitle

\begin{abstract}
{\it The integrability of the classical and quantum rational Calogero-Moser systems is verified explicitly via the Lax pair method for the case $n=3$. We provide an extensive survey of reflection groups and root systems. The Olshanetsky-Perelomov operators are constructed for a general root system via Dunkl operators, associated to root systems. The integrability of the quantum rational Calogero-Moser system is discussed via the Olshanetsky-Perelomov operators, which provide a set of commuting integrals of motion. The classical analogues of both the Dunkl and the Olshanetsky-Perelomov operators are also presented.}
\end{abstract}

\par {\bf Keywords:} integrable systems, finite-dimensional Hamiltonian systems, Calogero-Moser system, classical integrability, quantum integrability, Lax pair, Dunkl operator\\
\par {\bf AMS Subject Classification:} 37J35, 70H06 \\

\section*{Introduction}
\indent \par Conservation of energy is one of the fundamental scientific concepts, being well known by both the academic community and the general public. Natural questions to ask are: what other properties of a given dynamical system remain constant over time and how are they related to each other? The answers to these questions are subject to the study of integrable systems, which currently lies at the crossroad of many ares of mathematics and within the scope of interest of experts within various research fields, including, but not limited to Mathematical Physics, Dynamical Systems, Representation Theory, Symplectic Geometry.
\par Various integrable systems were discovered in the $19^{\text{th}}$ century, the Euler, Lagrange and Kovalevskaya tops, and the geodesic flow on an ellipsoid, to name a few \cite{tops}. The subject, however, remained somewhat dormant until the numerical results on the solitary waves phenomenon obtained by Zabusky \& Kruskal in 1965 \cite{zabusky}. Many equations have been studied since, the Korteweg-de Vries being the most prominent example, which yielded phenomenal results about the conserved quantities of the system at hand \cite{zabusky}. A considerable number of new integrable systems was discovered in the $70$s, among which are the Toda systems \cite{todasystem} and the main subject of this project -- the Calogero-Moser systems. These are both $n$-particle systems whose energy is given by a Hamiltonian of the form $H = \sum_{j=1}^{n} p_{j}^{2} + U(q)$, where in the Calogero-Moser case the potential is given by $U(q)=\sum_{1 \leq j < k \leq n} V(q_{j}-q_{k})$, describing the pairwise interaction between the $n$ particles \cite{Calogero}.
\par Calogero-Moser systems have been studied in both the classical and the quantum settings, and are generally classified as rational, trigonometric (hyperbolic) or elliptic, based on the form of the potential $U(q)$ \cite{OlshanetskyPerelomov}. There exist relativistic and non-relativistic versions as well \cite{Rujsenaars}. The main purpose of this project is to provide a somewhat complete and explicit overview of the rational Calogero-Moser system in both the classical and quantum sense. We investigate the integrability of the model via two separate techniques -- the Lax pair formalism and by constructing the integrals of motion via Dunkl operators. Further, we utilize a simple \emph{Mathematica} code in order to verify the explicit computations symbolically.
\par This paper is organized as follows. We begin by introducing the Hamiltonian formalism in the setting of classical mechanics in Section \ref{sec:classical}. Then we provide a working definition of integrability and introduce the Lax method for extracting integrals of motion. Then we present the classical rational Calogero-Moser system and illustrate its integrability by explicitly verifying the $n=3$ case via the \emph{Mathematica} code listed in Appendix \ref{appRational}.
\par Section \ref{sec:quantum} begins with a quick survey of quantum mechanics. Then we provide an analogous to the classical case definition of integrability and the Lax method. Hence we define the Hamiltonian of the quantum rational Calogero-Moser and explicitly show it is integrable in the $n=3$ case via the \emph{Mathematica} code listed in Appendix \ref{appQuantumRational}.
\par In Section \ref{sec:coxeter} we study reflection groups and root systems. We provide a detailed treatment of Coxeter graphs and Dynkin diagrams for root systems of different types.
\par Section \ref{sec:poly} is a quick survey of polynomial invariant theory, in turn needed to prove the main results in Section \ref{sec:Dunkl}, which outlines the construction of the integrals of motion of the quantum Calogero-Moser system via the Dunkl and Olshanetsky-Perelomov operators. Finally, we present the classical analogue of Dunkl operators in Section \ref{sec:dunklClassical}, which is then used in order to construct the classical Olshanetsky-Perelomov operator and thus, show the integrability of the Calogero-Moser system in the classical sense.
\par A working \emph{Wolfram Mathematica} code for the computations with the Lax pairs is provided in Appendices \ref{app}. 
\par I would like to take this opportunity to express my gratitude and dedicate this project to Dr Misha Feigin for his invaluable support both in academia and in life. This work is based on the author's master thesis submitted to the School of Mathematics and Statistics, University of Glasgow, UK back in 2017. The extended results exist due to the generous support provided by the Collaborative Research Centre/Transregio (CRC/TRR 191) on ``Symplectic Structures in Geometry, Algebra and Dynamics''.
%

\section{Classical Mechanics. Classical Rational Calogero-Moser.}\label{sec:classical}
\subsection{The Classical Frame and Poisson Brackets.}
\indent
\par  The \emph{configuration} (or \emph{state}) of an $n$-dimensional classical mechanical system is described via a set of $2n$ canonical variables (\emph{coordinates}) $\{q, p\}$, where $q = (q_{1}, \dots, q_{n})$ and $p = (p_{1}, \dots, p_{n})$ are called \emph{positions} and \emph{momenta}, respectively, or by a point $P = (q, p) \in \Gamma$, where $\Gamma$ is the $2n$-dimensional manifold describing the range of variation of the coordinates, called the phase space. Assume $\Gamma$ is compact, so that the system is confined on a bounded region of space and the energy of the system is bounded as well. The physical quantities (\emph{observables}), which describe the system are smooth functions $f(q,p) \in C^{\infty} ( \Gamma)$. Every state $P$ determines the values of the observables on that state and conversely, any state $P \in \Gamma$ is uniquely determined by the values of all the observables on that state. Thus, there exists a duality relation between states and observables~\cite{Strocchi}.
\par The change of an observable $f$ in time is measured by the time evolution of the canonical variables 
\begin{equation*}
q \to q_{t} = q(t), \quad p \to p_{t} = p(t), \quad f_{t} (q,p) = f(q_{t}, p_{t}).
\end{equation*}
 The time evolution of the variables is then given by Hamilton's equations of motion,
	\begin{equation}\label{hamilton}
	\dot{q_{j}} = \dfrac{\partial H }{\partial p_{j}} \qquad \text{ and } \qquad \dot{p_{j}} = - \dfrac{\partial H}{\partial q_{j}},
	\end{equation} 
where $H = H(q,p)$ is the Hamiltonian \cite{Strocchi}.
\par The observables generate an abelian algebra $\mathcal{A}$ of continuous functions on the compact phase space with a unit given by the constant function $1$ and with multiplication given by the point-wise product of functions. Further, $\mathcal{A}$ has identity $f=1$ due to the natural \emph{involution} (or $*$\emph{-operation}) given by the ordinary complex conjugation, which makes $\mathcal{A}$ a $*$-algebra. The time evolution of every observable $f_{t}(q,p) \in \mathcal{A}$ can be alternatively described by the differential equation
\begin{equation}\label{observablesChange}
\dfrac{d f_{t}(q,p)}{dt } = \sum_{j=1}^{n} \left( \dfrac{\partial f}{\partial q_{j}} \dfrac{d q_{j}}{d t} + \dfrac{\partial f}{\partial p_{j}} \dfrac{d p_{j}}{d t} \right).
\end{equation}
Substituting Hamilton's equations of motion \eqref{hamilton} in Equations \eqref{observablesChange}, we obtain
\begin{equation*}
\dfrac{d f_{t}(q,p)}{dt } = \sum_{j=1}^{n} \left( \dfrac{\partial f}{\partial q_{j}}  \dfrac{\partial H }{\partial p_{j}} + \dfrac{\partial f}{\partial p_{j}} \left( -  \dfrac{\partial H }{\partial q_{j}} \right) \right).
\end{equation*}
Then we define the following expression
\begin{equation}\label{poissonMap}
\{ f, g \} = \displaystyle{ \sum_{j=1}^{n} }\left( \dfrac{\partial f }{\partial p_{j}} \dfrac{\partial g }{\partial q_{j}}  - \dfrac{\partial f }{\partial q_{j}}  \dfrac{\partial g}{\partial p_{j}} \right) ,
\end{equation}
which allows us to concisely rewrite the time evolution of any observable $f$ as 
\begin{equation}\label{poissonBracket}
\dfrac{d f}{ dt} = \{ H, f \}.
\end{equation}
The bilinear map $\{ \cdot, \cdot \}:  C^{\infty} (\Gamma) \times C^{\infty} (\Gamma) \to C^{\infty} (\Gamma)$ given by Equation \eqref{poissonMap} is called the \emph{Poisson bracket} and it turns the algebra of classical observables $\mathcal{A}$ into a Lie algebra with a Lie bracket given by the Poisson bracket, due to the following properties, shown and proved in Proposition \ref{poissonBrackets} below.

\begin{proposition}\label{poissonBrackets}
The Poisson brackets satisfy the following properties, for $f,g, k \in C^{\infty}( \Gamma)$,
\begin{enumerate}
\item{Linearity:}
\begin{equation*} 
\{ \lambda f + \mu g, k\} = \lambda \{f,k\} + \mu \{g,k\}, \quad \forall \lambda, \mu \in \mathbb{R}.
\end{equation*}
\item{Anticommutativity:}
\begin{equation*} 
\{ f, g \} + \{ g, f \} = 0.
\end{equation*}
\item{Leibniz rule:}
\begin{equation*} 
\{  f, gk\} = \{f,g\}k + g\{f,k\}, \qquad \text{and} \qquad \{fg, k\} = \{f,k\}g + f \{g,k\}.
\end{equation*}
\item{Jacobi identity:}
\begin{equation*} 
\{f, \{g, k\} \} + \{ g, \{k, f\} \} + \{k, \{f, g \} \} = 0.
\end{equation*}
\end{enumerate}
\end{proposition}

\begin{proof}
\begin{enumerate}
\item{Linearity.} Using Equation \eqref{poissonMap} to write out the expression explicitly, we get
\begin{align*}
\{ \lambda f + \mu g, k \} &= \sum_{j} \left( \dfrac{\partial \left( \lambda f + \mu g \right)}{\partial p_{j}} \dfrac{\partial k}{ \partial q_{j}} - \dfrac{\partial \left( \lambda f + \mu g \right)}{\partial q_{j}} \dfrac{\partial k}{ \partial p_{j}} \right) \\
&= \sum_{j} \left( \lambda \dfrac{\partial f}{\partial p_{j}} + \mu \dfrac{\partial g}{ \partial p_{j}} \right) \dfrac{\partial k }{\partial q_{j}} - \sum_{j} \left( \lambda \dfrac{\partial f}{ \partial q_{j}} + \mu \dfrac{\partial g }{\partial q_{j}} \right) \dfrac{\partial k}{\partial p_{j}},
\end{align*}
where we use linearity of the ordinary partial derivatives. Next we bring the constant coefficients in front of the summations to get
\begin{align*}
 \{ \lambda f + \mu g, k \}      &= \lambda \sum_{j} \left( \dfrac{\partial f}{\partial p_{j}} \dfrac{\partial k}{\partial q_{j}} \right) + \mu \sum_{j} \left( \dfrac{\partial g}{\partial p_{j}} \dfrac{\partial k}{\partial q_{j}} \right) - \lambda \sum_{j} \left( \dfrac{\partial f}{\partial q_{j}} \dfrac{\partial k}{ \partial p_{j}} \right) - \mu \sum_{j} \left( \dfrac{\partial g}{\partial q_{j}} \dfrac{\partial k}{\partial p_{j}} \right).
\end{align*}
Grouping the summations with the same coefficients we obtain
\begin{align*}
\{ \lambda f + \mu g, k \} &= \lambda \sum_{j} \left( \dfrac{ \partial f}{\partial p_{j}} \dfrac{ \partial k}{\partial q_{j}}  - \dfrac{ \partial f}{\partial q_{j}} \dfrac{ \partial k}{\partial p_{j}} \right) + \mu \sum_{j} \left( \dfrac{ \partial g}{\partial p_{j}} \dfrac{ \partial k}{\partial q_{j}}  - \dfrac{ \partial g}{\partial q_{j}} \dfrac{ \partial k}{\partial p_{j}} \right) \\
&= \lambda \{f,k \} + \mu \{ g,k\},
\end{align*} 
using the definition of the Poisson bracket as per Equation \eqref{poissonMap} and we are done.
\item{Antisymmetry.} Again, writing out the expression using Equation \eqref{poissonMap} we get
\begin{align*}
\{ f,g \} + \{ g, f\} &= \sum_{j} \left( \dfrac{ \partial f}{\partial p_{j}} \dfrac{ \partial g}{\partial q_{j}}  - \dfrac{ \partial f}{\partial q_{j}} \dfrac{ \partial g}{\partial p_{j}} \right) + \left( \dfrac{ \partial g}{\partial p_{j}} \dfrac{ \partial f}{\partial q_{j}}  - \dfrac{ \partial g}{\partial q_{j}} \dfrac{ \partial f}{\partial p_{j}} \right) \\
&= \cancel{\sum_{j} \dfrac{ \partial f}{\partial p_{j}} \dfrac{ \partial g}{\partial q_{j}} } - \bcancel{\sum_{j} \dfrac{ \partial f}{\partial q_{j}} \dfrac{ \partial g}{\partial p_{j}}} + \bcancel{\sum_{j} \dfrac{ \partial g}{\partial p_{j}} \dfrac{ \partial f}{\partial q_{j}} } - \cancel{\sum_{j} \dfrac{\partial g}{\partial q_{j}} \dfrac{ \partial f}{\partial p_{j}} } \\
&= 0,
\end{align*}
since all terms with opposite signs cancel out as illustrated above.
\item{Leibniz rule.} We only prove the first of the two versions of the Leibniz rule. The second one follows the same logic. By Equation \eqref{poissonMap}, we have
\begin{align*}
\{f, gk \} &= \sum_{j} \left( \dfrac{\partial f}{\partial p_{j}} \dfrac{\partial  \left(gk\right)}{ \partial q_{j}} - \dfrac{\partial f}{\partial q_{j}} \dfrac{\partial \left( gk \right)}{\partial p_{j}}    \right) = \sum_{j} \left( \dfrac{\partial f}{\partial p_{j}} \left( \dfrac{\partial g}{\partial q_{j}}k + g\dfrac{\partial k}{\partial q_{j}}\right)  - \dfrac{\partial f}{\partial q_{j}} \left( \dfrac{\partial g}{\partial p_{j}} k + g \dfrac{\partial k}{\partial p_{j}} \right) \right),
\end{align*}
where we have applied the product rule. Expanding the brackets, we get
\begin{align*}
\{f, gk \} &= \sum_{j} \left( \dfrac{\partial f}{\partial p_{j}} \dfrac{\partial g}{\partial q_{j}}\right) k + \sum_{j} g \left( \dfrac{\partial f}{\partial p_{j}}  \dfrac{\partial k}{\partial q_{j}} \right)  - \sum_{j} \left( \dfrac{\partial f}{\partial q_{j}}  \dfrac{\partial g}{\partial p_{j}} \right) k - \sum_{j} g \left( \dfrac{\partial f}{\partial q_{j}}  \dfrac{\partial k}{\partial p_{j}} \right).
\end{align*}
Now we group terms based on the non-differentiated functions $k$ and $g$ to obtain
\begin{align*}
\{f, gk \} &= \sum_{j}  \left( \dfrac{ \partial f}{\partial p_{j}} \dfrac{ \partial g}{\partial q_{j}}  - \dfrac{ \partial f}{\partial q_{j}} \dfrac{ \partial g}{\partial p_{j}} \right) k + \sum_{j} g  \left( \dfrac{ \partial f}{\partial p_{j}} \dfrac{ \partial k}{\partial q_{j}}  - \dfrac{ \partial f}{\partial q_{j}} \dfrac{ \partial k}{\partial p_{j}} \right)  \\
&= \{ f,g\}k + g\{f,k\},
\end{align*}
as required.
\item{Jacobi identity.} In order to show that the Poisson brackets satisfy the Jacobi identity, we consider each term separately. Using the definition of the Poisson bracket as in Equation \eqref{poissonMap}, we obtain the following expression for the first term, namely
\begin{align*}
\{ f, \{g,k\}\} &= \left\{ f, \sum_{j} \left( \dfrac{ \partial f}{\partial p_{j}} \dfrac{ \partial g}{\partial q_{j}}  - \dfrac{ \partial f}{\partial q_{j}} \dfrac{ \partial g}{\partial p_{j}} \right) \right\}.
\end{align*}
Now after one more application of Equation \eqref{poissonMap} we get
\begin{align*}
\{ f, \{g,k\}\} &= \sum_{j} \left( \dfrac{ \partial f}{\partial p_{j}} \dfrac{ \partial }{\partial q_{j}} \left( \dfrac{\partial g}{\partial p_{j}} \dfrac{\partial k}{\partial q_{j}} - \dfrac{\partial g}{\partial q_{j}} \dfrac{\partial k}{\partial p_{j}}    \right) - \dfrac{ \partial f}{\partial q_{j}} \dfrac{ \partial }{\partial p_{j}} \left( \dfrac{\partial g}{\partial p_{j}} \dfrac{\partial k}{\partial q_{j}} - \dfrac{\partial g}{\partial q_{j}} \dfrac{\partial k}{\partial p_{j}}  \right) \right) \\
&= \sum_{j} \left(  \underbrace{\dfrac{ \partial f}{\partial p_{j}}  \dfrac{ \partial^{2} g}{\partial q_{j} \partial p_{j}} \dfrac{\partial k}{\partial q_{j}}}_{a} + \underbrace{\dfrac{ \partial f}{\partial p_{j}} \dfrac{ \partial g}{\partial p_{j}}  \dfrac{ \partial^{2} k}{\partial q_{j}^{2}}}_{b} -  \underbrace{\dfrac{ \partial f}{\partial p_{j}} \dfrac{ \partial^{2} g}{\partial q_{j}^{2}} \dfrac{ \partial k}{\partial p_{j}}}_{c} - \underbrace{\dfrac{ \partial f}{\partial p_{j}} \dfrac{ \partial g}{\partial q_{j}} \dfrac{ \partial^{2} k}{\partial p_{j} \partial q_{j}}}_{d}  \right. \\ 
& \left. \qquad - \underbrace{\dfrac{ \partial f}{\partial q_{j}}  \dfrac{ \partial^{2} g}{\partial p_{j}} \dfrac{ \partial k}{\partial q_{j}}}_{e}
 - \underbrace{\dfrac{ \partial f}{\partial q_{j}}  \dfrac{ \partial g}{\partial p_{j}}  \dfrac{ \partial^{2} k}{\partial p_{j} \partial q_{j}}}_{f} + \underbrace{\dfrac{ \partial f}{\partial q_{j}}  \dfrac{ \partial^{2} g}{\partial p_{j} \partial q_{j}}  \dfrac{ \partial k}{\partial p_{j}}}_{g}  + \underbrace{\dfrac{ \partial f}{\partial q_{j}} \dfrac{ \partial g}{\partial q_{j}} \dfrac{ \partial^{2} k}{\partial p_{j}^{2}}}_{h} 
\right),
\end{align*}
where the letter under each term is a label for the cancelation to be carried out later. Similarly for the second term using Equation \eqref{poissonMap} twice and underbracing the cancelation terms by letters yields
\begin{align*}
\{ g, \{k,f\}\} &= \left\{ g, \sum_{j} \left( \dfrac{ \partial k}{\partial p_{j}} \dfrac{ \partial f}{\partial q_{j}}  - \dfrac{ \partial k}{\partial q_{j}} \dfrac{ \partial f}{\partial p_{j}} \right) \right\} \\
&= \sum_{j} \left( \dfrac{ \partial g}{\partial p_{j}} \dfrac{ \partial }{\partial q_{j}} \left( \dfrac{\partial k}{\partial p_{j}} \dfrac{\partial f}{\partial q_{j}} - \dfrac{\partial k}{\partial q_{j}} \dfrac{\partial f}{\partial p_{j}}    \right) - \dfrac{ \partial g}{\partial q_{j}} \dfrac{ \partial }{\partial p_{j}} \left( \dfrac{\partial k}{\partial p_{j}} \dfrac{\partial f}{\partial q_{j}} - \dfrac{\partial k}{\partial q_{j}} \dfrac{\partial f}{\partial p_{j}}  \right) \right) \\
&= \sum_{j} \left(  \underbrace{\dfrac{ \partial f}{\partial p_{j}}  \dfrac{ \partial^{2} k}{\partial q_{j} \partial p_{j}} \dfrac{\partial f}{\partial q_{j}}}_{f} + \underbrace{\dfrac{ \partial g}{\partial p_{j}} \dfrac{ \partial k}{\partial p_{j}}  \dfrac{ \partial^{2} f}{\partial q_{j}^{2}}}_{i} -  \underbrace{\dfrac{ \partial g}{\partial p_{j}} \dfrac{ \partial^{2} k}{\partial q_{j}^{2}} \dfrac{ \partial f}{\partial p_{j}}}_{b} - \underbrace{\dfrac{ \partial g}{\partial p_{j}} \dfrac{ \partial k}{\partial q_{j}} \dfrac{ \partial^{2} f}{\partial p_{j} \partial q_{j}}}_{j}  \right. \\ 
& \left. \qquad - \underbrace{\dfrac{ \partial g}{\partial q_{j}}  \dfrac{ \partial^{2} k}{\partial p_{j}} \dfrac{ \partial f}{\partial q_{j}}}_{h}
 - \underbrace{\dfrac{ \partial g}{\partial q_{j}}  \dfrac{ \partial k}{\partial p_{j}}  \dfrac{ \partial^{2} f}{\partial p_{j} \partial q_{j}}}_{k} + \underbrace{\dfrac{ \partial g}{\partial q_{j}}  \dfrac{ \partial^{2} k}{\partial p_{j} \partial q_{j}}  \dfrac{ \partial f}{\partial p_{j}}}_{d}  + \underbrace{\dfrac{ \partial g}{\partial q_{j}} \dfrac{ \partial k}{\partial q_{j}} \dfrac{ \partial^{2} f}{\partial p_{j}^{2}}}_{l}
\right).
\end{align*}
In the same way we use Equation \eqref{poissonMap} twice and indicate via letters the pairing cancelations to get
\begin{align*}
\{ k, \{f,g\}\} &= \left\{ k, \sum_{j} \left( \dfrac{ \partial f}{\partial p_{j}} \dfrac{ \partial g}{\partial q_{j}}  - \dfrac{ \partial f}{\partial q_{j}} \dfrac{ \partial g}{\partial p_{j}} \right) \right\} \\
&= \sum_{j} \left( \dfrac{ \partial k}{\partial p_{j}} \dfrac{ \partial }{\partial q_{j}} \left( \dfrac{\partial f}{\partial p_{j}} \dfrac{\partial g}{\partial q_{j}} - \dfrac{\partial f}{\partial q_{j}} \dfrac{\partial g}{\partial p_{j}}    \right) - \dfrac{ \partial k}{\partial q_{j}} \dfrac{ \partial }{\partial p_{j}} \left( \dfrac{\partial f}{\partial p_{j}} \dfrac{\partial g}{\partial q_{j}} - \dfrac{\partial f}{\partial q_{j}} \dfrac{\partial g}{\partial p_{j}}  \right) \right) \\
&= \sum_{j} \left(  \underbrace{\dfrac{ \partial k}{\partial p_{j}}  \dfrac{ \partial^{2} f}{\partial q_{j} \partial p_{j}} \dfrac{\partial g}{\partial q_{j}}}_{k} + \underbrace{\dfrac{ \partial k}{\partial p_{j}} \dfrac{ \partial f}{\partial p_{j}}  \dfrac{ \partial^{2} g}{\partial q_{j}^{2}}}_{c} -  \underbrace{\dfrac{ \partial k}{\partial p_{j}} \dfrac{ \partial^{2} f}{\partial q_{j}^{2}} \dfrac{ \partial g}{\partial p_{j}}}_{i} - \underbrace{\dfrac{ \partial k}{\partial p_{j}} \dfrac{ \partial f}{\partial q_{j}} \dfrac{ \partial^{2} g}{\partial p_{j} \partial q_{j}}}_{g}  \right. \\ 
& \left. \qquad - \underbrace{\dfrac{ \partial k}{\partial q_{j}}  \dfrac{ \partial^{2} f}{\partial p_{j}} \dfrac{ \partial g}{\partial q_{j}}}_{l}
 - \underbrace{\dfrac{ \partial k}{\partial q_{j}}  \dfrac{ \partial f}{\partial p_{j}}  \dfrac{ \partial^{2} g}{\partial p_{j} \partial q_{j}}}_{a} + \underbrace{\dfrac{ \partial k}{\partial q_{j}}  \dfrac{ \partial^{2} f}{\partial p_{j} \partial q_{j}}  \dfrac{ \partial g}{\partial p_{j}}}_{j}  + \underbrace{\dfrac{ \partial k}{\partial q_{j}} \dfrac{ \partial f}{\partial q_{j}} \dfrac{ \partial^{2} g}{\partial p_{j}^{2}}}_{e}
\right).
\end{align*}
All letters from $a$ to $l$ appear twice among the three expressions above, thus, all terms cancel and prove that the Jacobi identity holds as required.
\end{enumerate}
\end{proof}
\subsection{Integrals of Motion and Classical Integrability.}
\par An observable $I \in C^{\infty}(\Gamma)$ is called an \emph{integral of motion} if it is constant in time. Using the Poisson bracket notation, $I$ is an integral of motion if $\dot{I} = \{ H, I \} = 0$. If this holds, it is said that the observables $H$ and $I$ are \emph{in involution} (or \emph{Poisson commute}).
\begin{definition}\label{liouvilleintegrabilitiy}
If an $n$-dimensional dynamical system has $n$ functionally independent smooth functions $I_{1}=H, \dots, I_{n}$ such that they are in pairwise involution, i.e. $\{ I_{i}, I_{j} \}=0$ for $i \neq j$, then the system is said to be \emph{integrable}.
\end{definition} 
\par Liouville showed that given an $n$-dimensional mechanical system with $n$ independent integrals of motion in involution, the equations of motion can be integrated by normal means of integration \cite{Arnold}, \cite{liouvilleint}. One way of identifying the integrals of motion is via the Lax Pair method, which we discuss next.
\subsection{Classical Integrability via the Lax Formalism.}
 \indent \par The Lax Pair method is a powerful tool, which was first formulated by Peter Lax in his remarkable paper \cite{Lax}. Given an integrable system, we wish to find a pair of matrices $L$ and $M$, such that the Lax equation 
	\begin{equation}\label{laxEqn}
	\dot{L} = [ L, M ]
	\end{equation} 
 is equivalent to Hamilton's equations of motion \eqref{hamilton}. The following proposition summarizes the key property of the Lax matrix $L$, which hints at the usefulness in determining the integrals of motion. 
 \begin{proposition}
 The Lax matrix $L = L(t)$ satisfies 
\begin{equation}\label{laxPairIff}
L(t) = U^{-1}(t)L(0)U(t), 
\end{equation}
where $U$ is such that $\dot{U} = UM$ and $U(0) = \mathbb{I}$, where $\mathbb{I}$ is the identity matrix. In particular, the eigenvalues of $L$ are independent of time. 
\end{proposition}

\begin{proof}
By definition of the inverse, $UU^{-1}=\mathbb{I}$. Differentiating and rearranging we get an expression for the derivative of $U^{-1}$, namely
\begin{equation}\label{inverseDerivative}
\dfrac{dU^{-1}}{dt} = -U^{-1}\dfrac{dU}{dt}U^{-1}.
\end{equation} 
Now differentiating Equation \eqref{laxPairIff} and substituting with Equation \eqref{inverseDerivative}, we get
\begin{align*}
\dfrac{d L}{dt} &= \dfrac{dU^{-1}}{dt}L(0)U + U^{-1}L(0)\dfrac{dU}{dt} = -U^{-1}\dfrac{dU}{dt}U^{-1} L(0)U + U^{-1}L(0)\dfrac{dU}{dt}.
\end{align*}
Using the condition that $\dfrac{dU}{dt} = UM$ gives
\begin{align*}
\dfrac{d L}{dt} &= -\underbrace{U^{-1}U}_{\mathbb{I}} M U^{-1}L(0)U + U^{-1}L(0)UM  = -ML + LM = [L,M],
\end{align*}
hence, the Lax matrix $L$ as in Equation \eqref{laxPairIff} indeed provides the solution of the Lax equation \eqref{laxEqn} as required. Note that $L(t)$ and $L(0)$ are conjugate and therefore, have the same eigenvalues. Since $L(0)$ is constant and thus, has constant eigenvalues, it follows that the eigenvalues of $L(t)$ are also independent of time.
\end{proof}
Using the isospectral property of the Lax matrix $L$ we can identify the integrals of motion associated with the mechanical system.

\begin{theorem}
The integrals of motion of system \eqref{hamilton} with Lax pair \eqref{laxEqn} are given by the power traces of the Lax matrix $L$, that is
	\begin{equation}\label{integralsLax}
	I_{k} = \dfrac{1}{k} \tr (L^{k}).
	\end{equation}
\end{theorem}
\begin{remark}
The classical integrals of motion are obtained by the characteristic polynomial of the matrix $L$. Hence, they are uniquely determined by a homogeneous polynomial expression with a fixed highest term in the momenta $p_{j}$. Therefore, they are unique.
\end{remark}

The involutive property, however, is handled in a slightly more complicated manner. One way of showing that the integrals of motion given by Equation \eqref{integralsLax} Poisson commute with each other is utilizing a special object, called the \emph{$r$-matrix}, which satisfies the \emph{Yang-Baxter equation (YBE)} \cite{Babelon}. A thorough treatment and detailed discussion of the Lie theoretic perspective of the classical $r$-matrix is given in \cite{Semenov}. The general idea of the proof is that the existence of the $r$-matrix is equivalent to the eigenvalues of the matrix $L$ being in involution, which in turn guarantees the involution of the integrals of motion. We refer the reader to consult \cite{intlectures} for more details on the matter.

\par On a separate note, verifying explicitly the number of integrals of motion is usually straightforward. Hence, we can now present the system of interest.

\subsection{Classical Rational Calogero-Moser System.}
\indent \par
\begin{definition}
The Hamiltonian of the classical rational Calogero-Moser System \cite{Moser} is
	\begin{align}\label{CM}
	{H} =  \dfrac{1}{2} \sum_{j=1}^{n} {p}_{j}^{2} + \sum_{\substack{ i = 1 \\ i \neq j }}^{n} \dfrac{k^{2}}{(q_{j}-q_{i})^{2}}.
	\end{align}
\end{definition}
	It describes the motion of a system of $n$ particles on the line with parameter of interaction $k \in \mathbb{R}$, as illustrated below in Figure \ref{classical_calogero_fig}.

\begin{figure}
\centering
    	\begin{tikzpicture}
 	 \filldraw 
	         (0, -5) circle (0pt) node {} --
		(2,-5) circle (1pt) node[align=left,   below] {$q_{1}$} --
		(4,-5) circle (1pt) node[align=center, below] {$q_{2}$}     -- 
		(6,-5) circle (1pt) node[align=right,  below] {\dots} --
		(8,-5) circle (1pt) node[align=center, below] {$q_{i}$} --
		(10, -5) circle (1pt) node[align=right, below] {\dots} --
		(12,-5) circle (1pt) node[below] {$q_{n-1}$} --
		(14, -5) circle (1pt) node[align=right, below] {$q_{n}$} --
		(16, -5) circle (0pt) node {} 
		(3, -5.5) node {$\longleftrightarrow$}
		(3, -6) node {$\frac{k^{2}}{(q_{1}-q_{2})^{2}}$}
		(13, -5.5) node {$\longleftrightarrow$}
		(13, -6) node {$\frac{k^{2}}{(q_{n-1}-q_{n})^{2}}$}		;
	\end{tikzpicture} 
	\caption{The configuration of the classical rational Calogero-Moser system.}
	\label{classical_calogero_fig}
\end{figure}
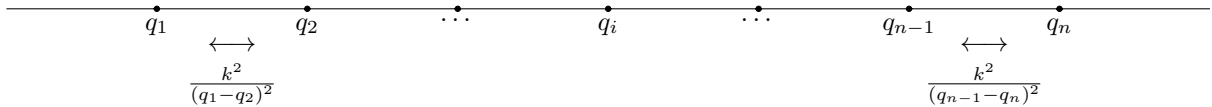
	Whenever we denote a generalized coordinate with a double subscript, i.e. $q_{ij}$ we shall mean the difference between the $i$-th and $j$-th coordinate, that is, $(q_{i}-q_{j})$. Similarly, $q_{ij}^{(k)}$ denotes $(q_{i}-q_{j})^{-k}$. Using Hamilton's equations of motion \eqref{hamilton} we get
\begin{align}
\begin{aligned}\label{hamiltonClassical}
&\dot{q_{i}} = \dfrac{\partial H}{\partial p_{i}} \implies \dot{q_{i}}= p_{i} \implies \ddot{q_{i}}=\dot{p_{i}},\\
& \dot{p_{i}} = -\dfrac{\partial H}{\partial q_{i}} \implies \dot{p_{i}} = \sum_{j=1}^{n} 2k^{2} q_{ij}^{(3)},
\end{aligned}
\end{align}
which are the equations of motion governing the dynamics of the rational Calogero-Moser system. Furthermore, it is an integrable system and we shall verify this statement via the Lax method. 

\subsection{Integrability of the Classical Calogero-Moser via Lax Formalism.}
\begin{theorem}
A Lax pair for the system \eqref{CM} is given by
	\begin{align}\label{laxPairMoserClassical}
	L_{rs} = p_{r} \delta_{rs} + (1-\delta_{rs}) \dfrac{ik}{q_{rs}}, \qquad M_{rs} = -\dfrac{ik (1-\delta_{rs})}{q_{rs}^{2}} + \delta_{rs} \sum_{\substack{t=1 \\ t \neq r}}^{n} \dfrac{ik}{q_{rt}^{2}},
	\end{align}
where $\delta$ is the Kronecker delta, $i = \sqrt{-1}$ and $r,s = 1, \dots, n$ \cite{MoserPaper} . 
\end{theorem}

\begin{proof}
We explicitly compute the case $n=3$. Using the proposed formula, the Lax matrices $L$ and $M$ are given by
\begin{align*}
L = \begin{bmatrix}
  p_{1} & ik q_{12}^{(1)} & ik q_{13}^{(1)} \\
  ik q_{21}^{(1)} &  p_{2} & ik q_{23}^{(1)} \\
  ik q_{31}^{(1)} & ik q_{32}^{(1)} &  p_{3}
 \end{bmatrix}, \qquad \text{ and } \qquad 
M= \begin{bmatrix}
 ik( q_{12}^{(2)}+ q_{13}^{(2)}) & - ik q_{12}^{(2)} & -ik q_{13}^{(2)} \\
  -ik q_{21}^{(2)} &   ik( q_{21}^{(2)}+ q_{23}^{(2)}) & -ik q_{23}^{(2)} \\
  -ik q_{31}^{(2)} & -ik q_{32}^{(2)} &   ik( q_{31}^{(2)}+ q_{32}^{(2)})
 \end{bmatrix}.
\end{align*}
Next we compute the total time derivative of the matrix $L$ to obtain the LHS of the Lax equation \eqref{laxEqn} to get
\begin{align}\label{classicalLaxDot}
\dot{L} =  \begin{bmatrix}
  \dot{p}_{1} & -ik(\dot{q}_{1}-\dot{q}_{2})  q_{12}^{(2)} & -ik(\dot{q}_{1}-\dot{q}_{3})  q_{13}^{(2)} \\
-ik(\dot{q}_{2}-\dot{q}_{1})  q_{21}^{(2)} & \dot{p}_{2} & -ik(\dot{q}_{2}-\dot{q}_{3})  q_{23}^{(2)} \\
-ik(\dot{q}_{3}-\dot{q}_{1})  q_{31}^{(2)}  & -ik(\dot{q}_{3}-\dot{q}_{2})  q_{32}^{(2)} &  \dot{p}_{3}
 \end{bmatrix}.
\end{align}
Now we need to compute the commutator $[L,M] = LM - ML$. To ease notation a bit, we shall denote terms of the form $(q_{rs}^{-w}+q_{rt}^{-w})$ by $(q_{rsrt}^{-w})$, where $r \to s \to t$ denote permutations of the indices in the respective order. Thus, we  calculate the product $LM$ to be
\begin{align*}
&LM = \\
&k^{2}\begin{bmatrix*}[l]
\dfrac{i}{k} p_{1}( q_{1213}^{(2)}) +  q_{1213}^{(3)} & -\dfrac{i}{k} p_{1}  q_{12}^{(2)} -  q_{12}^{(1)} ( q_{2123}^{(2)}) +  q_{13}^{(1)}  q_{32}^{(2)} &  -\dfrac{i}{k}p_{1} q_{13}^{(2)} +  q_{12}^{(1)} q_{23}^{(2)} -  q_{13}^{(1)}( q_{3132}^{(2)}) \\ && \\
- q_{21}^{(1)}( q_{1213}^{(2)})-\dfrac{i}{k}p_{2} q_{21}^{(2)}+ q_{23}^{(1)} q_{31}^{(2)} & q_{21}^{(3)}+\dfrac{i}{k}p_{2}( q_{2123}^{(2)})+ q_{23}^{(3)} &  q_{21}^{(1)} q_{13}^{(2)} - \dfrac{i}{k}p_{2} q_{23}^{(2)}- q_{23}^{(1)} ( q_{3132}^{(2)}) \\ && \\
- q_{31}^{(1)} ( q_{1213}^{(2)}) +  q_{32}^{(1)} q_{21}^{(2)}-\dfrac{i}{k}p_{3} q_{31}^{(2)} &  q_{31}^{(1)} q_{12}^{(2)}- q_{32}^{(1)} ( q_{2123}^{(2)}) -\dfrac{i}{k}p_{3} q_{32}^{(2)} &  q_{3132}^{(3)} + \dfrac{i}{k}p_{3}( q_{3132}^{(2)})
\end{bmatrix*},
\end{align*}
where we have used the fact that $i^{2}=-1$ and $ q_{rs}^{(1)} q_{sr}^{(2)} =  q_{rs}^{(3)}$  to slightly simplify the expressions. Now we calculate the second term of the commutator, that is, $ML$ as follows:
\begin{align*}
&ML= \\
& k^{2}\begin{bmatrix*}[l]
\dfrac{i}{k} ( q_{1213}^{(2)})p_{1}+q_{2131}^{-3} & -( q_{1213}^{(2)}) q_{12}^{(1)}-\dfrac{i}{k} q_{12}^{(2)}p_{2}+ q_{13}^{(2)} q_{32}^{(1)} & - ( q_{1213}^{(2)}) q_{13}^{(1)}+ q_{12}^{(2)} q_{23}^{(1)}-\dfrac{i}{k} q_{13}^{(2)}p_{3} \\ & & \\
-\dfrac{i}{k} q_{21}^{(2)}p_{1}-( q_{2123}^{(2)})q_{21}+ q_{23}^{(2)} q_{31}^{(1)} &  q_{12}^{(3)}+\dfrac{i}{k}( q_{2123}^{(2)})p_{2}+ q_{32}^{(3)} &  q_{21}^{(2)} q_{13}^{(1)}-( q_{2123}^{(2)}) q_{23}^{(1)}-\dfrac{i}{k} q_{23}^{(2)}p_{3} \\ &&\\
-\dfrac{i}{k} q_{31}^{(2)}p_{1}+ q_{32}^{(2)} q_{21}^{(1)}-( q_{3132}^{(2)})  q_{31}^{(1)} &  q_{31}^{(2)} q_{12}^{(1)}-\dfrac{i}{k} q_{32}^{(2)}p_{2}-( q_{3132}^{(2)}) q_{32}^{(1)} & q_{1323}^{-3}+\dfrac{i}{k}( q_{3132}^{(2)})p_{3}
\end{bmatrix*},
\end{align*}
where again we have used $i^{2}=-1$ and $ q_{rs}^{(1)} q_{sr}^{(2)} =  q_{rs}^{(3)}$. Next we subtract the two expressions to obtain the commutator of $L$ and $M$, that is
\begin{align}\label{classicalLaxCommutator}
LM-ML = \begin{bmatrix*}[l]
2k^{2} q_{1213}^{(3)} & -ik(p_{1}-p_{2}) q_{12}^{(2)} & -ik(p_{1}-p_{3}) q_{13}^{(2)} \\ && \\
 -ik(p_{2}-p_{1}) q_{21}^{(2)} & 2k^{2} q_{2123}^{(3)} &  -ik(p_{2}-p_{3}) q_{23}^{(2)} \\ && \\
 -ik(p_{3}-p_{1}) q_{31}^{(2)} &  -ik(p_{3}-p_{2}) q_{32}^{(2)} & 2k^{2} q_{3132}^{(3)}
\end{bmatrix*},
\end{align}
where in the $r^{\text{th}}$ diagonal entry we use the following identities
\begin{align*}
 q_{rs}^{(3)} -  q_{sr}^{(3)} =  q_{rs}^{(3)} +  q_{rs}^{(3)} \qquad \text{ and } \qquad ip_{r} ( q_{rsrt}^{(2)}) - i ( q_{rsrt}^{(2)}) p_{r}= 0, \quad r \neq s,t.
\end{align*}
Similarly, in the $rs^{\text{th}}$ off-diagonal position we utilized the relations 
\begin{align*}
- q_{rs}^{(1)}( q_{srst}^{(2)}) - (-( q_{srst}^{(2)}) q_{rs}^{(1)}) = 0 \qquad \text{ and } \qquad  q_{rs}^{(1)}q_{st}^{-1} -  q_{rs}^{(1)}q_{st}^{-1} = 0.
\end{align*}
Comparing Equations \eqref{classicalLaxDot} and \eqref{classicalLaxCommutator}, we get
\begin{align}
\begin{aligned}\label{classicalLaxResult}
 \dot{L}_{rr} &= \dot{p}_{r}= 2 q_{rs}^{(2)}+2 q_{rt}^{(2)}  = \sum_{\substack{u=1 \\ u \neq r }}^{3} 2k^{2}(q_{ru}^{(3)}) = [L,M]_{rr}, \\ 
\dot{L}_{rs} &= -ik (\dot{q}_{r}-\dot{q}_{s}) q_{rs}^{(2)} = -ik(p_{r}-p_{s}) q_{rs}^{(2)} = [L,M]_{rs} \quad \forall r \neq s, 
\end{aligned}
\end{align}
for the diagonal and off-diagonal elements respectively. Next we observe that Equations \eqref{classicalLaxResult} recover Hamilton's equations for the rational Calogero-Moser system \eqref{hamiltonClassical}, which confirms that the proposed Lax pair \eqref{laxPairMoserClassical} indeed satisfies the Lax equation \eqref{laxEqn}. This also holds for the general $n>3$ case, see \cite{MoserPaper} for the original claim. 
\end{proof}
We can now verify via direct computations the integrability of the rational Calogero-Moser system.
\begin{corollary}
The $n=3$ rational Calogero-Moser system is integrable in the sense of Definition \ref{liouvilleintegrabilitiy}.
\end{corollary}

\begin{proof}
First we need to identify the three integrals of motion. We use Equation \eqref{integralsLax} to read them from the power traces of the Lax matrix $L$. The first integral of motion is given by
\begin{align*}
I_{1} = \sum_{r=1}^{3} L = p_{1}+p_{2}+p_{3} = \sum_{r=1}^{3} p_{r},
\end{align*}
which corresponds to conservation of momentum. In order to compute the second one, we need to evaluate $L^{2}$, yielding
\begin{align*}
&L^{2} = LL =  \begin{bmatrix}
  p_{1} & ik q_{12}^{(1)} & ik q_{13}^{(1)} \\
  ik q_{21}^{(1)} &  p_{2} & ik q_{23}^{(1)} \\
  ik q_{31}^{(1)} & ik q_{32}^{(1)} &  p_{3}
 \end{bmatrix}  \begin{bmatrix}
  p_{1} & ik q_{12}^{(1)} & ik q_{13}^{(1)} \\
  ik q_{21}^{(1)} &  p_{2} & ik q_{23}^{(1)} \\
  ik q_{31}^{(1)} & ik q_{32}^{(1)} &  p_{3}
 \end{bmatrix} \\ & \\
&= 
\begin{bmatrix*}
p_{1}^{2}+k^{2} q_{1213}^{(2)} & ikp_{1} q_{12}^{(1)}+ik q_{12}^{(1)}p_{2}-k^{2} q_{13}^{(1)} q_{32}^{(1)} & ikp_{1} q_{13}^{(1)}-k^{2} q_{12}^{(1)} q_{23}^{(1)}+ik q_{13}^{(1)}p_{3} \\ && \\
ik q_{21}^{(1)}p_{1}+ikp_{2} q_{21}^{(1)} - k^{2} q_{23}^{(1)} q_{31}^{(1)} & k^{2} q_{21}^{(2)}+p_{2}^{2}+k^{2} q_{23}^{(2)} & -k^{2} q_{21}^{(1)} q_{13}^{(1)}+ikp_{2} q_{23}^{(1)}+ik q_{23}^{(1)}p_{3} \\ && \\
ik q_{31}^{(1)}p_{1}-k^{2} q_{32}^{(1)} q_{21}^{(1)}+ikp_{3} q_{31}^{(1)} & -k^{2} q_{31}^{(1)} q_{12}^{(1)}+ik q_{32}^{(1)}p_{2}+ikp_{3} q_{32}^{(1)} & k^{2} q_{3132}^{(2)}+p_{3}^{2}
\end{bmatrix*},
\end{align*}
where we have used that
\begin{align*}
 q_{rs}^{(1)} q_{sr}^{(1)} = -  q_{rs}^{(2)}
\end{align*} 
to simplify the expressions in the diagonal entries. Thus, using Equation \eqref{integralsLax} we obtain
\begin{align*}
I_{2} &= \dfrac{1}{2} \sum_{r=1}^{3} L^{2} = \dfrac{1}{2} \left(   p_{1}^{2}+ k^{2} q_{12}^{(2)}+ k^{2} { q_{13}^{(2)}} + k^{2}{  q_{21}^{(2)}}+p_{2}^{2}+ k^{2}{ q_{23}^{(2)}} +   k^{2} { q_{31}^{(2)}}+ k^{2}{ q_{32}^{(2)}}+p_{3}^{2} \right) \\
&= \dfrac{1}{2} (p_{1}^{2}+p_{2}^{2}+p_{3}^{2}) + \dfrac{1}{2} ( 2 k^{2}  q_{12}^{(2)} + 2 k^{2}  q_{23}^{(2)} + 2 k^{2}  q_{13}^{(2)} ) \\
&= \dfrac{1}{2} \sum_{r=1}^{3} p_{r}^{2} + \sum_{\substack{s =1 \\ s \neq r }}^{3} k^{2}  q_{rs}^{(2)} = H,
\end{align*}
which coincides with the Hamiltonian \eqref{CM} of the rational Calogero-Moser when $n=3$. Thus, what we have shown is that energy is a conserved quantity. Finally, in order to compute the last integral of motion $I_{3}$ we need an expression for $L^{3}$, which is calculated via Wolfram Mathematica using the code in Appendix \ref{appRational}. The entries of $L^{3}$ are then
\begin{align*}
L_{11}^{3} &=  p_1^3+2 k^2 \left(  q_{31}^{(2)}+ q_{21}^{(2)}\right) p_1+k^2 \left(p_2  q_{21}^{(2)}+p_3  q_{31}^{(2)}\right), \\
L_{12}^{3} &=  \frac{k \left(k I_{1} q_{31} q_{32} q_{21}^3-i \left(\left(k^2+\left(p_1^2+p_2 p_1+p_2^2\right)
   q_{32}^2\right) q_{31}^2+k^2 q_{32}^2\right) q_{21}^2-i k^2 q_{31}^2 q_{32}^2\right)}{q_{21}^3 q_{31}^2 q_{32}^2} ,\\ 
L_{13}^{3} &=  \frac{k \left(-k I_{1}q_{21} q_{32} q_{31}^3-i k^2 q_{32}^2 q_{31}^2-i q_{21}^2 \left(\left(k^2+\left(p_1^2+p_3
   p_1+p_3^2\right) q_{32}^2\right) q_{31}^2+k^2 q_{32}^2\right)\right)}{q_{21}^2 q_{31}^3 q_{32}^2} ,\\
L_{21}^{3} &= 
 \frac{k \left(k I_{1}q_{31} q_{32} q_{21}^3+i \left(\left(k^2+\left(p_1^2+p_2 p_1+p_2^2\right) q_{32}^2\right) q_{31}^2+k^2 q_{32}^2\right) q_{21}^2+i k^2 q_{31}^2 q_{32}^2\right)}{q_{21}^3 q_{31}^2 q_{32}^2}, \\
L_{22}^{3} &=  p_2^3+2 k^2
   \left( q_{32}^{(2)}+  q_{21}^{(2)}\right) p_2+k^2 \left( p_1  q_{21}^{(2)}+ p_3  q_{32}^{(2)}\right), \\
L_{23}^{3} &= \frac{k \left(k I_{1}q_{21} q_{31} q_{32}^3-i k^2 q_{31}^2 q_{32}^2-i q_{21}^2 \left(\left(k^2+\left(p_2^2+p_3
   p_2+p_3^2\right) q_{32}^2\right) q_{31}^2+k^2 q_{32}^2\right)\right)}{q_{21}^2 q_{31}^2 q_{32}^3}, \\
L_{31}^{3} &=  \frac{k \left(-k I_{1}q_{21} q_{32} q_{31}^3+i k^2 q_{32}^2 q_{31}^2+i q_{21}^2 \left(\left(k^2+\left(p_1^2+p_3 p_1+p_3^2\right) q_{32}^2\right) q_{31}^2+k^2 q_{32}^2\right)\right)}{q_{21}^2 q_{31}^3 q_{32}^2}, \\
L_{32}^{3} &=  \frac{k \left(k
   I_{1}q_{21} q_{31} q_{32}^3+i k^2 q_{31}^2 q_{32}^2+i q_{21}^2 \left(\left(k^2+\left(p_2^2+p_3 p_2+p_3^2\right) q_{32}^2\right) q_{31}^2+k^2 q_{32}^2\right)\right)}{q_{21}^2 q_{31}^2 q_{32}^3}, \\
L_{33}^{3} &=  p_3^3+2 k^2
   \left( q_{32}^{(2)}+ q_{31}^{(2)}\right) p_3+k^2 \left( p_1  q_{31}^{(2)}+ p_2  q_{32}^{(2)}\right) ,\\
\end{align*}
where $I_{1} = p_{1}+p_{2}+p_{3}$. Then using Equation \eqref{integralsLax} we get
\begin{align*}
I_{3} = \dfrac{1}{3} \tr L^{3} &=\dfrac{1}{3} \left(p_1^3+2 k^2 \left( q_{31}^{(2)}+ q_{21}^{(2)}\right) p_1+k^2 \left(p_2q_{21}^{(2)}+ p_3 q_{31}^{(2)}\right) +   p_2^3+2 k^2 \left( q_{32}^{(2)}+ q_{21}^{(2)}\right) p_2+ \right. \\
& \left. \qquad k^2 \left( p_1 q_{21}^{(2)}+p_3 q_{32}^{(2)}\right) +  p_3^3+2 k^2  \left( q_{32}^{(2)}+ q_{31}^{(2)}\right) p_3+k^2 \left( p_1 q_{31}^{(2)}+ p_2 q_{32}^{(2)}\right) \right) \\
&= \dfrac{1}{3} \sum_{r=1}^{3} p_{r}^{3} + \dfrac{k^{2}}{3} \sum_{r=1}^{3} \left( 2 p_{r} \sum_{\substack{s=1 \\ s \neq r}}  q_{rs}^{(2)} + p_{r} \sum_{\substack{s=1 \\ s \neq r}}  q_{rs}^{(2)} \right) \\
&=  \dfrac{1}{3} \sum_{r=1}^{3} p_{r}^{3} + k^{2}  \sum_{r=1}^{3} p_{r}\sum_{\substack{s=1 \\ s \neq r}}  q_{rs}^{(2)} ,
\end{align*}
which is the third and last integral of motion of the $n=3$ rational Calogero-Moser. Next we need to verify the involutivity of the integrals of motion. Consider the Poisson bracket of $I_{1}$ and $I_{2} = H$. We get
\begin{align*}
\{ I_{1},H \} &= \sum_{j=1}^{3} \left(    \dfrac{\partial I_{1} }{\partial p_{j}} \dfrac{\partial H }{\partial q_{j}}  - \dfrac{\partial I_{1} }{\partial q_{j}}  \dfrac{\partial H}{\partial p_{j}}  \right) \\
&= \dfrac{\partial I_{1} }{\partial p_{1}} \dfrac{\partial H }{\partial q_{1}} -\cancel{ \dfrac{\partial I_{1} }{\partial q_{1}}  \dfrac{\partial H}{\partial p_{1}}} + \dfrac{\partial I_{1} }{\partial p_{2}} \dfrac{\partial H }{\partial q_{2}} - \cancel{ \dfrac{\partial I_{1} }{\partial q_{2}}  \dfrac{\partial H}{\partial p_{2}}} + \dfrac{\partial I_{1} }{\partial p_{3}} \dfrac{\partial H }{\partial q_{3}} - \cancel{ \dfrac{\partial I_{1} }{\partial q_{3}}  \dfrac{\partial H}{\partial p_{3}}}, 
\end{align*}
where the cancellation is due to the independence of $I_{1}$ on $q_{j}$, $\forall j=1,2,3$. Using the fact that $\dfrac{\partial I_{1}}{\partial p_{j}} = 1 $ for $j=1,2,3$, we obtain
\begin{align}\label{poissonI1H}
\{ I_{1},H \} &= \dfrac{\partial H }{\partial q_{1}} +  \dfrac{\partial H }{\partial q_{2}} + \dfrac{\partial H }{\partial q_{3}} =\cancel{ -2k^{2} q_{12}^{(3)}} - \bcancel{2k^{2} q_{13}^{(3)}} - \xcancel{2k^{2} q_{23}^{(3)}} +\cancel{ 2k^{2} q_{12}^{(3)}} + \bcancel{2k^{2} q_{13}^{(3)}} + \xcancel{2k^{2} q_{23}^{(3)}} = 0.
\end{align}
Next we verify the involution of $I_{1}$ and $I_{3}$ and again, using the independence of $I_{1}$ on any of the $q_{j}$, we simply write
\begin{align}\label{poissonI1I3}
\begin{aligned}
\{ I_{1}, I_{3} \} &=\dfrac{\partial I_{3}}{\partial q_{1}} +  \dfrac{\partial I_{3} }{\partial q_{2}} + \dfrac{\partial I_{3} }{\partial q_{3}} \\
&= 2k^{2} \left( p_{1} \left(  q_{31}^{(3)} + q_{21}^{(3)} \right) + p_{2} q_{21}^{(3)} + p_{3}  q_{31}^{(3)} - p_{2} (q_{21}^{(3)} +  q_{23}^{(3)} ) + p_{1}  q_{12}^{(3)} + \right. \\
& \left. \qquad \quad + p_{3}  q_{32}^{(3)} - p_{3}( q_{32}^{(3)} +  q_{31}^{(3)}) -p_{1}  q_{31}^{(3)} -p_{2} q_{32}^{(3)}  \right) \\
&= 0,
\end{aligned}
\end{align}
where we have used that $p_{r}  q_{rs}^{(3)} + p_{r}  q_{sr}^{(3)}=0$ and we have cancelled the terms with same indices and opposite signs. Finally, we compute the Poisson bracket of the Hamiltonian $H$ and the third integral of motion $I_{3}$, which yields
\begin{align*}
\{ H, I_{3} \} =   \sum_{j=1}^{3} \dfrac{\partial H }{\partial p_{j}} \dfrac{ \partial I_{3}}{\partial q_{j}} - \dfrac{\partial I_{3} }{\partial p_{j}} \dfrac{\partial H }{\partial q_{j}},
 \end{align*}
where we note that $\dfrac{\partial H}{\partial p_{j}} = p_{j}$ and $\dfrac{\partial I_{3}}{\partial p_{j}} = p_{j}^{2}$. Hence, we get
\begin{align*}
\{ H, I_{3} \} = \sum_{j=1}^{3} p_{j} \dfrac{\partial I_{3}}{\partial q_{j}} - p_{j}^{2} \dfrac{\partial H }{\partial q_{j}} = 0
\end{align*}
using the results in \eqref{poissonI1H} and \eqref{poissonI1I3}. Hence, all three integrals of motion are in involution, and the rational Calogero-Moser system for $n=3$  is integrable in the sense of Definition \ref{liouvilleintegrabilitiy}.
\end{proof}

\section{Quantum Mechanics. Quantum Rational Calogero-Moser.}\label{sec:quantum}

\indent \par \emph{Canonical quantization} is a procedure initially introduced in Paul Dirac's PhD thesis \cite{dirac1} and further developed in \cite{dirac2}. The aim of the technique is to quantize a classical theory, while preserving formal structures and properties, such as symmetries and integrability. The Poisson brackets formulation of the classical theory \ref{poissonBracket} are generally only partially preserved in the canonical quantization \cite{quantizationproblem}. However, this is not the case in the Calogero-Moser system as we shall now see.
\par The quantization of the classical mechanical system \eqref{hamilton} is obtained by replacing the canonical momentum $p_{j}$ in the Hamiltonian $H$ with the \emph{quantum momentum operator $\widehat{p}_{j}$} by
\begin{equation}\label{quantumMomentum}
p_{j} \mapsto -i \hbar \dfrac{\partial}{\partial q_{j}} \equiv \widehat{p_{j}}, \quad j = 1, \dots, n,
\end{equation}
where $i = \sqrt{-1}$ and $\hbar$ is Planck's constant (set $\hbar=1$). In this way we obtain a new quantum Hamiltonian, which we shall denote by $\widehat{H}$. Hence, the state of a quantum mechanical system is described via the set $\{ \widehat{q}, \widehat{p} \}$, where $\widehat{q}=(\widehat{q}_{1}, \dots, \widehat{q}_{n})$ and $\widehat{p}=(\widehat{p}_{1}, \dots, \widehat{p}_{n})$ are the quantum coordinates and momenta operators, respectively. The quantum variables satisfy the following relations,
\begin{align}\label{heisenbergCommutation}
[\widehat{q}_{j},\widehat{p}_{k} ] = i \delta_{jk}, \qquad [\widehat{q}_{j}, \widehat{q}_{k}] = [\hat{p}_{j} , \hat{p}_{k} ] =0, \qquad \text{ where } j,k=1, \dots, n,
\end{align}
which are known as the \emph{Heisenberg commutation relations}. The vector space over $\mathbb{C}$ generated by $\widehat{q}$ and $\widehat{p}$ with Lie brackets given by Equation \eqref{heisenbergCommutation} is the \emph{Heisenberg Lie algebra} \cite{Strocchi}. The time evolution equations for $\widehat{q}$ and $\widehat{p}$ are given by the \emph{Heisenberg equations of motion} 
	\begin{equation}\label{heisenbergEqns}
	\dot{\widehat{q}_{j}} = i [\widehat{H}, \widehat{q}_{j}] \qquad \text{ and } \qquad \dot{\widehat{p}}_{j} = i[\widehat{H},\widehat{p}_{j}],
	\end{equation} 
	in analogy with Hamilton's Equations \eqref{hamilton}.
\par The quantum analogues of the classical observables are operators. An operator $F$ satisfies Heisenberg's equation 
 	\begin{equation*}
	\dot{F} = i [ \widehat{H}, F],
	\end{equation*} 
in analogy with Equations \eqref{poissonBracket} \cite{Strocchi}.
\subsection{Quantum Integrals of Motion and Quantum Integrability.}
Given an operator $J $, such that $[ \widehat{H} , J ] =0$, $J$ is called a \emph{quantum integral of motion}.
\begin{definition}
If an $n$-dimensional quantum dynamical system has $n$ algebraically independent operators $J_{1} = \widehat{H}, \dots, J_{n}$ such that they are pairwisely commutative, i.e. $[ J_{i}, J_{j} ]=0$ for $i \neq j$, then the system is said to be \emph{quantum integrable} \cite{Rujsenaars} .

\end{definition}
\par It is of great importance to note that upon quantization, we obtain operators which \emph{do not} necessarily commute. Thus, we need to carefully consider the order of multiplication, when doing computations in the quantum setting. The following identity turns out to be quite useful and clearly illustrates the lack of commutativity.
\begin{lemma}\label{keyLemma}
The commutator of the quantum momentum operator $\hat{p}_{i}$ and any function $f(\widehat{q}_{i})$ is the total derivative of the function with respect to the $\widehat{q}_{i}$ up to a scale, i.e.
\begin{align}\label{eqkeyLemma}
[ \widehat{p}_{i}, f(\widehat{q}_{i}) ] = -if'(\widehat{q}_{i}).
\end{align}
\end{lemma}

\begin{proof}
Writing out the definition of the quantum momentum operator explicitly and substituting into the commutator relation, we get
\begin{align*}
[ \widehat{p}_{i}, f(\widehat{q}_{i}) ]  &= -i \partial_{\widehat{q}_{i}} f(\widehat{q}_{i}) - f(\widehat{q}_{i})\left( -i \partial_{\widehat{q}_{i}} \right) =\cancel{i f(\widehat{q}_{i}) \partial_{\widehat{q}_{i}}} - \cancel{if(\widehat{q}_{i}) \partial_{\widehat{q}_{i}}} -i f'(\widehat{q}_{i})    =- i f'(\widehat{q}_{i}),
\end{align*}
where prime denotes derivative with respect to $\widehat{q}_{i}$.
\end{proof}
\par Thus, when conjugating by a first order partial differential operator, we think of this operation as computing the ordinary derivative of the function at hand. Next we present the quantum analogue of the Lax method, which is used in a similar way to extract the integrals of motion.

\subsection{Integrability via the Lax Formalism.}\label{quantumLaxSection}
\indent
\par The quantum analogue of the Lax pair method consists of finding a pair of operator-valued matrices $\widehat{L}$ and $\widehat{M}$ such that the quantum Lax equation
	\begin{equation}\label{quantumLax1}
	\dot{\widehat{L}}= i \left[\widehat{L} , \widehat{M}  \right],
	\end{equation} 
is equivalent to the Heisenberg equation of motion
	\begin{equation}\label{quantumLax2}
	\dot{\widehat{L}} =i \left[\tilde{H} ,\widehat{L}  \right],
	\end{equation} 
	where $\tilde{H} = \widehat{H} \mathbb{I}$, $\mathbb{I}$ is the identity matrix \cite{Jurco}.
	Then under some extra assumptions, which we shall discuss later on, 
the set of integrals of motion is given by summing the elements of the powers of $\widehat{L}$, i.e.
	\begin{equation}\label{integralsQuantumLax}
	J_{m} = \dfrac{1}{m} \sum_{j,k \neq 1}^{n} \left( \widehat{L}^{m} \right)_{jk}.
	\end{equation}
\par In contrast with the classical setting, where the Poisson bracket of any two integrals of motion is well defined, in the quantum setting this is not always the case when we compute the commutator of two operators \cite{quantizationproblem}. The proof we provide here is outlined in great detail in \cite{OlshanetskyPerelomov}.
\begin{proposition}
The quantum integrals of motion given by \eqref{integralsQuantumLax} commute with the Hamiltonian of the system, i.e. $[\widehat{J}_{k} , \widehat{H} ] = 0$, for all $k=1,\dots, n$.
\end{proposition}
\begin{proof}
The classical integrals $I_{k}$, which can be defined as the coefficients of the characteristic polynomial $\det [L - \lambda \mathbb{I} ]$, where $L$ is the Lax matrix and $\lambda$ the corresponding eigenvalue as in Equation \eqref{laxEqn}, correspond in the quantum case to the operators ${J}_{k}$, which can in turn be considered as the determinant $\det [ \widehat{L} - \widehat{\lambda} \mathbb{I} ]$ for the corresponding quantum analogues as illustrated in great detail in \cite{CalogeroMarcioro}. However, the vanishing of the Poisson bracket of any two classical integrals of motion $\{ I_{j}, I_{k} \}$ does not necessarily guarantee the vanishing of the commutator of two quantum integrals $[{J}_{j}, {J}_{k} ]$. It is sufficient to show that the $n$-th quantum integral of motion $J_{n}$ commutes with $\widehat{H}$, since we have the following recursive relation 
\begin{align}\label{integralsrecursive}
i {J}_{k-1} = \dfrac{1}{n-k+1} \left[ \left( \sum_{l=1}^{n} \widehat{q}_{l} \right) , {J}_{k} \right].
\end{align}
Using the Jacobi identity for the operators $\displaystyle{\sum_{l=1}^{n} \widehat{q}_{l}}$, $\widehat{H}$ and ${J}_{k}$ we have
\begin{align}\label{op41}
\left[ \left[ \sum_{l=1}^{n} \widehat{q}_{l} , \widehat{H} \right] , {J}_{k} \right] = \left[ \sum_{l=1}^{n} \widehat{q}_{l} , \left[ \widehat{H}, {J}_{k} \right] \right] - \left[ \widehat{H}, \left[ \sum_{l=1}^{n} \widehat{q}_{l}, {J}_{k} \right] \right],
\end{align}
where the last term clearly indicates that if ${J}_{k}$ is an integral of motion, then so is ${J}_{k-1}$. As shown in \cite{CalogeroMarcioro}, the $n$-th operator ${J}_{n}$ can be expressed as
\begin{align*}
{J}_{n} = A_{12} \left( \widehat{p}_{1} \widehat{p}_{2} - f_{12}^{2} \right) + B_{1} \widehat{p}_{1} + B_{2} \widehat{p}_{2} + B_{12} f_{12} + C_{12},
\end{align*}
where $f_{12} = f(\widehat{q}_{1},\widehat{q}_{2})$ accounts for the potential of the Hamiltonian, and the coefficients $A_{12}$, $B_{1}$, $B_{2}$ and $C_{12}$ are independent of $f_{12}$, $\widehat{p}_{1}$ and $\widehat{p}_{2}$, \cite{OP22}. Next we take the commutator of $\widehat{H}$ and ${J}_{k}$ yielding
\begin{align*}
\left[\widehat{H}, {J}_{n} \right] &= \left[ \left( \dfrac{1}{2} \widehat{p}_{1} + \dfrac{1}{2} \widehat{p}_{2} + f_{12} \right) , \left(A_{12} \left( \widehat{p}_{1} \widehat{p}_{2} - f_{12}^{2} \right) + B_{1} \widehat{p}_{1} + B_{2} \widehat{p}_{2} + B_{12} f_{12} + C_{12} \right) \right] \\
&= B_{1} f_{12}^{2} \widehat{p}_{1} + B_{2} f_{12}^{2}\widehat{p}_{2} + \dfrac{B_{12}}{2} \widehat{p}_{1}^{2}f_{12} +  \dfrac{B_{12}}{2} \widehat{p}_{2}^{2}f_{12} - B_{1} \widehat{p}_{1}  f_{12}^{2} -B_{2} \widehat{p}_{2} f_{12}^{2} - \dfrac{B_{12}}{2} f_{12}\widehat{p}_{1}^{2} -\dfrac{B_{12}}{2} f_{12}\widehat{p}_{2}^{2} + \dots \\
&= 2(B_{2}-B_{1}) f_{12} f'_{12} + \dfrac{B_{12}}{2} \left( (\widehat{p}_{1}-\widehat{p}_{2}) f'_{12} + f'_{12} (\widehat{p}_{1} - \widehat{p}_{2})\right) + \dots ,
\end{align*}
where we have discarded all terms with coefficients $A_{12}$ and $C_{12}$ as they are straightforward to handle and have applied Lemma \ref{keyLemma} subsequently.
The first term of the expression above is a well-defined operator, but the second is somewhat ambiguous. We need to ensure that they vanish as well. We use a different expression for ${J}_{n}$, initially introduced in \cite{Wojciechowski}, namely
\begin{align*}
\widehat{J}_{n} = \exp \left( - \dfrac{g^{2}}{2} \sum_{k,l} f^{2}(q_{k}, q_{l} ) \dfrac{\partial }{\partial \widehat{p}_{k}} \dfrac{\partial }{\partial \widehat{p}_{l}} \right) \widehat{p}_{1} \cdots \widehat{p}_{n},
\end{align*}
which clearly illustrates that the only quadratic terms in the operator ${J}_{n}$ are of the form $f_{kl}$, and in our particular instance, $f_{12}$. This allows us to rewrite Equation \eqref{op41} as
\begin{align}\label{commutej}
{J}_{n} = a_{12} \left( \widehat{p}_{1} \widehat{p}_{2} - f^{2}_{12} \right) + b_{1} \widehat{p}_{1} + b_{2} \widehat{p}_{2} + c'_{12},
\end{align}
where once again $a_{12}$, $b_{1}$, $b_{2}$ and $c'_{12}$ are independent of $f_{12}$, $\widehat{p}_{1}$ and $\widehat{p}_{2}$. Now we commute ${J}_{n}$ as in Equation \eqref{commutej} with $\widehat{H}$ to get
\begin{align}
[{J}_{n}, \widehat{H} ] &= \left[ a_{12} \left( \widehat{p}_{1} \widehat{p}_{2} - f^{2}_{12} \right) + b_{1} \widehat{p}_{1} + b_{2} \widehat{p}_{2} + c'_{12} , \widehat{H} \right] \nonumber \\
&= \left[ a_{12} \widehat{p}_{1} \widehat{p}_{2} ,  \widehat{H} \right] - \left[  a_{12} f^{2}_{12} , \widehat{H} \right] + \left[ b_{1}, \widehat{p}_{1} \right] + \left[ b_{2} \widehat{p}_{2}, \widehat{H} \right] + \left[ c'_{12}, \widehat{H} \right], \label{cancelquantumclassical}
\end{align}
which all yield well-defined operators. Assuming the vanishing of the Poisson bracket of the corresponding classical integral $I_{n}$ with the Hamiltonian, i.e. $\{ H, I_{n}\}=0$, all the commutative relations in Equation \eqref{cancelquantumclassical} cancel, which shows that the operators ${J}_{k}$ are indeed integrals of motion. 
\par Now the operators ${J}_{k}$ are functionally independent as polynomials in $\widehat{p}_{j}$ of degree $k$. Using the recursive relation \eqref{integralsrecursive}, we can present any integral in a form similar to the one in Equation \eqref{cancelquantumclassical}, which only contains well-defined operators. Therefore, the commutator of any two integrals of motion vanishes as long as their classical analogues Poisson commute.
\end{proof}
\par So far we have provided a consistent methodology for extracting the integrals of motion from a given Lax pair and we have shown that they indeed commute, under some assumptions. We are now ready to present the quantum version of the rational Calogero-Moser system and explore its integrability via the Lax pair method.

\subsection{Quantum Rational Calogero-Moser System.}
\indent \par
\begin{definition}
The Hamiltonian of the quantum Calogero-Moser System \cite{Calogero} is
	\begin{align}\label{quantumCM}
	\widehat{H} =  \dfrac{1}{2} \sum_{j=1}^{n} \widehat{p}_{j}^{2} +  \sum_{\substack{i = 1 \\ i \neq j}}^{n} g{\widehat{q}_{ji}^{(2)}}.
	\end{align}
\end{definition}

	It describes the motion of a system of $n$ particles on the line with parameter of interaction $g \in \mathbb{R}$. As before, $\widehat{q}_{ji}^{(2)}$ denotes $(\widehat{q}_{j}-\widehat{q}_{i})^{-2}$. We remark that $g=k(k+\hbar)$, where $\hbar$ is Planck's constant and $k \in \mathbb{R}$. However, we set $\hbar=1$. The interaction within the system of particles is illustrated below in Figure \ref{quantum_calogero_fig}.

\begin{figure}
\centering
    	\begin{tikzpicture}
 	 \filldraw 
	         (0, -5) circle (0pt) node {} --
		(2,-5) circle (1pt) node[align=left,   below] {$\widehat{q}_{1}$} --
		(4,-5) circle (1pt) node[align=center, below] {$\widehat{q}_{2}$}     -- 
		(6,-5) circle (1pt) node[align=right,  below] {\dots} --
		(8,-5) circle (1pt) node[align=center, below] {$\widehat{q}_{i}$} --
		(10, -5) circle (1pt) node[align=right, below] {\dots} --
		(12,-5) circle (1pt) node[below] {$\widehat{q}_{n-1}$} --
		(14, -5) circle (1pt) node[align=right, below] {$\widehat{q}_{n}$} --
		(16, -5) circle (0pt) node {} 
		(3, -5.5) node {$\longleftrightarrow$}
		(3, -6) node {$\frac{k(k+1)}{(\widehat{q}_{1}-\widehat{q}_{2})^{2}}$}
		(13, -5.5) node {$\longleftrightarrow$}
		(13, -6) node {$\frac{k(k+1)}{(\widehat{q}_{n-1}-\widehat{q}_{n})^{2}}$}		;
	\end{tikzpicture} 
	\caption{The configuration of the quantum rational Calogero-Moser system.}
	\label{quantum_calogero_fig}
\end{figure}
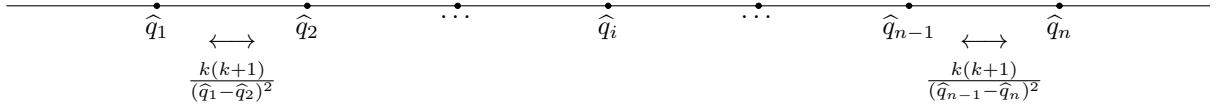

\par Using Heisenberg's equations of motion as in \eqref{heisenbergEqns}, we obtain
\begin{align*}
 \dot{\widehat{q}}_{r} &= i [ \widehat{H}, \widehat{q}_{r} ] = \dfrac{i}{2} \left(   \sum_{j=1}^{n} \widehat{p}_{j}^{2} +  \sum_{\substack{i = 1 \\ i \neq j}}^{n} g \widehat{q}_{ji}^{(2)} \right)\widehat{q}_{r} - \dfrac{i}{2} \widehat{q}_{r} \left(  \sum_{j=1}^{n} \widehat{p}_{j}^{2} + \sum_{\substack{i = 1 \\ i \neq j}}^{n} g \widehat{q}_{ji}^{(2)} \right),
\end{align*}
where we note that all terms with coefficients $g$ cancel since multiplication by the canonical coordinates $\widehat{q}_{j}$, $\forall j$ is commutative and results in opposite signs of the terms. Further, we observe that applying the quantum momentum operator $\widehat{p}_{s}$ to a coordinate $\widehat{q}_{r}$ where $r \neq s$ is automatically zero. Therefore, we only consider the case when $r=s$ and after applying Lemma \ref{keyLemma} we obtain
\begin{align}\label{quantumHeisenbergQ}
 \dot{\widehat{q}}_{r} &= \dfrac{i}{2} \left( \widehat{p}_{r}^{2} \widehat{q}_{r} - \widehat{q}_{r} \widehat{p}_{r}^{2} \right) = - \dfrac{i}{\cancel{2}} \cancel{2} i \widehat{p}_{r} = \widehat{p}_{r}.
\end{align}
Next, using the second Heisenberg equation of motion as in \eqref{heisenbergEqns} we get
\begin{align*}
 \dot{\widehat{p}}_{r} &= i [\widehat{H}, \widehat{p}_{r} ] = -i \left( \widehat{p}_{r} \left(  \sum_{j=1}^{n} \widehat{p}_{j}^{2} + \sum_{\substack{i = 1 \\ i \neq j}}^{n} g \widehat{q}_{ji}^{(2)}\right) -  \left(  \sum_{j=1}^{n} \widehat{p}_{j}^{2} + \sum_{\substack{i = 1 \\ i \neq j}}^{n} g \widehat{q}_{ji}^{(2)}\right) \widehat{p}_{r}  \right),
\end{align*}
where we perform the following cancellations:
\begin{align*}
\widehat{p}_{r} \widehat{p}_{r}^{2} - \widehat{p}_{r}^{2} \widehat{p}_{r} = \widehat{p}_{r}^{3} - \widehat{p}_{r}^{3} = 0 
\end{align*}
and for $r \neq s$,
\begin{align*}
 \widehat{p}_{r} \widehat{p}_{s}^{2} - \widehat{p}_{s}^{2} \widehat{p}_{r} = -i \partial_{\widehat{q}_{r}} \partial^{2}_{\widehat{q}_{s}} - \partial^{2}_{\widehat{q}_{s}}(-i\partial_{\widehat{q}_{r}}) = i ( \partial^{2}_{\widehat{q}_{s}} \partial_{\widehat{q}_{r}} - \partial_{\widehat{q}_{r}} \partial^{2}_{\widehat{q}_{s}} ) = 0
\end{align*}
due to the symmetry of ordinary first order partial derivatives and the fact that $i \neq 0$. Thus, we are left with the terms with coefficient $g$, which we handle as before. Applying the quantum momentum operator $\widehat{p}_{s}$ to a term of the form $\widehat{q}_{rt}$, for $s \neq r, t$ gives a zero, therefore we are left with the commutator
\begin{align}\label{quantumHeisenbergEqns}
 \dot{\widehat{p}}_{r} =  i \left[ \widehat{p}_{r}, g \sum_{\substack{ s = 1 \\ s \neq r}}^{n}  \widehat{q}_{rs}^{(2)}  \right] = i \left[ -i \partial_{\widehat{q}_{r}}, g\sum_{\substack{ s = 1 \\ s \neq r}}^{n}  \widehat{q}_{rs}^{(2)}  \right] =  \sum_{\substack{ s = 1 \\ s \neq r}}^{n} 2g \widehat{q}_{rs}^{(3)},
\end{align}
after applying Lemma \ref{keyLemma}.

The quantum rational Calogero-Moser system is also integrable and we proceed by verifying this statement via the Lax method.
\subsection{Integrability of the Quantum Calogero-Moser System.}
We begin by verifying that the Lax pair given in \cite{Ujino} indeed satisfies the Lax equation.
\begin{theorem}
	A Lax pair for the quantum rational Calogero-Moser model is given by
	\begin{align*}
	\widehat{L}_{rs} =\widehat{p}_{r} \delta_{rs} + (1-\delta_{rs}) ik \widehat{q}_{rs}^{(1)}, \qquad \widehat{M}_{rs} = -k(1-\delta_{rs}) \widehat{q}_{rs}^{(2)} + k \delta_{rs} \sum_{\substack{t=1 \\ t \neq r}}^{n} \widehat{q}_{rt}^{(2)},
	\end{align*}
for $r,s=1,\dots, n$ \cite{Ujino}.
\end{theorem}

\begin{proof}
We explicitly verify the case $n=3$ as it is a clear illustration on the procedure used to handle the lack of commutativity of the momenta operators. Using the proposed expressions for the Lax matrices $\widehat{L}$ and $\widehat{M}$ we get
\begin{align*}
\widehat{L} = \begin{bmatrix}
  \widehat{p}_{1} & ik \widehat{q}_{12}^{(1)} & ik \widehat{q}_{13}^{(1)} \\
  ik \widehat{q}_{21}^{(1)} &   \widehat{p}_{2} & ik \widehat{q}_{23}^{(1)} \\
  ik \widehat{q}_{31}^{(1)} & ik \widehat{q}_{32}^{(1)} &   \widehat{p}_{3}
 \end{bmatrix}, \qquad \text{ and } \qquad 
\widehat{M}= \begin{bmatrix}
 k( \widehat{q}_{12}^{(2)}+ \widehat{q}_{13}^{(2)}) & - k \widehat{q}_{12}^{(2)} & -k \widehat{q}_{13}^{(2)} \\
  -k \widehat{q}_{21}^{(2)} &   k( \widehat{q}_{21}^{(2)}+ \widehat{q}_{23}^{(2)}) & -k \widehat{q}_{23}^{(2)} \\
  -k \widehat{q}_{31}^{(2)} & -k \widehat{q}_{32}^{(2)} &   k( \widehat{q}_{31}^{(2)}+ \widehat{q}_{32}^{(2)})
 \end{bmatrix}.
\end{align*}
In order to obtain $\dot{\widehat{L}}$ as in Equation \eqref{quantumLax2} we first multiply the Hamiltonian operator by the $3 \times 3$ identity matrix to get
\begin{align*}
\tilde{H} = \widehat{H} \mathbb{I}_{3} = 
\begin{bmatrix*}
\dfrac{1}{2}\displaystyle{\sum_{r=1}^{3}} p_{r}^{2} +g\displaystyle{\sum_{\substack{s=1 \\ s \neq r}}^{3}} \widehat{q}_{rs}^{(2)} & 0 & 0 \\
 0 & \dfrac{1}{2}\displaystyle{\sum_{r=1}^{3}} p_{r}^{2} +g\displaystyle{\sum_{\substack{s=1 \\ s \neq r}}^{3}}  \widehat{q}_{rs}^{(2)}  & 0 \\
 0 & 0 &\dfrac{1}{2}\displaystyle{\sum_{r=1}^{3}} p_{r}^{2} +g\displaystyle{\sum_{\substack{s=1 \\ s \neq r}}^{3}}  \widehat{q}_{rs}^{(2)}  \\
\end{bmatrix*},
\end{align*} 
which we then use to compute the commutator $\tilde{H}$ and $\widehat{L}$. We have
\begin{align*}
\dot{\widehat{L}} = i \left[ \tilde{H} ,\widehat{L} \right] = i
\begin{bmatrix*}
\widehat{H} \widehat{p}_{1} - \widehat{p}_{1} \widehat{H} & \widehat{H} \left(i k  \widehat{q}_{12}^{(1)}\right) -\left(i k  \widehat{q}_{12}^{(1)}\right) \widehat{H} & H \left(i k
    \widehat{q}_{13}^{(1)}\right)-\left(i k  \widehat{q}_{13}^{(1)}\right) H \\
\widehat{H} \left(i k  \widehat{q}_{21}^{(1)}\right)-\left(i k  \widehat{q}_{21}^{(1)}\right)\widehat{H} & \widehat{H} \widehat{p}_{2} - \widehat{p}_{2} \widehat{H} &\widehat{H} \left(i k  \widehat{q}_{23}^{(1)}\right)-\left(i k  \widehat{q}_{23}^{(1)}\right)\widehat{H} \\
 \widehat{H} \left(i k  \widehat{q}_{31}^{(1)}\right)-\left(i k  \widehat{q}_{31}^{(1)}\right)\widehat{H}&\widehat{H} \left(i k  \widehat{q}_{32}^{(1)}\right)-\left(i k  \widehat{q}_{32}^{(1)}\right)\widehat{H} &\widehat{H} \widehat{p}_{3} - \widehat{p}_{3} \widehat{H} \\
\end{bmatrix*}.
\end{align*}
For simplicity, we consider the diagonal and off-diagonal entries separately. Taking a factor of $-1$ in front of the resulting expressions, we get
\begin{align*}
\dot{\widehat{L}}_{rr} = -i [\widehat{p}_{r}, \widehat{H} ] \qquad \text{ and } \qquad \dot{\widehat{L}}_{rs} = -i [k\widehat{q}_{rs}^{(1)}, \widehat{H} ].
\end{align*}
Hence, we apply Lemma \ref{keyLemma} to the diagonal entries and taking into account the commutativity of the momentum-free terms, we obtain
\begin{align*}
\dot{\widehat{L}}_{rr} = \sum_{\substack{ s = 1 \\ s \neq r}}^{n} 2g \widehat{q}_{rs}^{(3)} = \dot{\widehat{p}}_{r} \quad \text{ and } \quad \dot{\widehat{L}}_{rs} =- i \left[ k \widehat{q}_{rs}^{(1)} , \sum_{t=1}^{3} \widehat{p}_{t}^{2} \right]  = -ik ( \dot{\widehat{q}}_{r}-\dot{\widehat{q}}_{s}) \widehat{q}_{rs}^{(2)} -2ik  \widehat{q}_{rs}^{(3)}, \quad r \neq s,
\end{align*}
by Equation \eqref{quantumHeisenbergEqns}. Thus, the LHS of the Lax equation \eqref{quantumLax1} is
\begin{align}\label{quantumLaxDot}
\dot{\widehat{L}} =  \begin{bmatrix}
  \dot{\widehat{p}}_{1} & -ik(\dot{\widehat{q}}_{1}-\dot{\widehat{q}}_{2})  \widehat{q}_{12}^{(2)} -2ik  \widehat{q}_{12}^{(3)} & -ik(\dot{\widehat{q}}_{1}-\dot{\widehat{q}}_{3})  \widehat{q}_{13}^{(2)}- 2ik  \widehat{q}_{13}^{(3)} \\
-ik(\dot{\widehat{q}}_{2}-\dot{\widehat{q}}_{1})  \widehat{q}_{21}^{(2)} - 2ik \widehat{q}_{21}^{(3)} & \dot{\widehat{p}}_{2} & -ik(\dot{\widehat{q}}_{2}-\dot{\widehat{q}}_{3})  \widehat{q}_{23}^{(2)} - 2ik  \widehat{q}_{23}^{(3)}\\
-ik(\dot{\widehat{q}}_{3}-\dot{\widehat{q}}_{1})  \widehat{q}_{31}^{(2)} - 2ik  \widehat{q}_{31}^{(3)} & -ik(\dot{\widehat{q}}_{3}-\dot{\widehat{q}}_{2})  \widehat{q}_{32}^{(2)} - 2ik  \widehat{q}_{32}^{(3)}&  \dot{\widehat{p}}_{3}
 \end{bmatrix}.
\end{align}
For the RHS of the Lax equation \eqref{quantumLax1} we split the calculations in two bits. Denoting terms of the form $( \widehat{q}_{rs}^{(2)}+ \widehat{q}_{rt}^{(2)})$ by $( \widehat{q}_{rsrt}^{(2)})$, we first compute $\widehat{L}\widehat{M}$, yielding
\begin{align*}
&\widehat{L}\widehat{M} =\\
&ik^{2}\begin{bmatrix*}[l]
- \dfrac{i}{k}   \widehat{p}_{1}    \widehat{q}_{1213}^{(2)}-  \widehat{q}_{12}^{(1)}   \widehat{q}_{21}^{(2)}-  \widehat{q}_{13}^{(1)}  \widehat{q}_{31}^{(2)} &
 \dfrac{i}{k}    \widehat{p}_{1}    \widehat{q}_{12}^{(2)}+  \widehat{q}_{12}^{(1)}    \widehat{q}_{2123}^{(2)}-  \widehat{q}_{13}^{(1)}  \widehat{q}_{32}^{(2)} & 
\dfrac{i}{k} \widehat{p}_{1}   \widehat{q}_{13}^{(2)}-  \widehat{q}_{12}^{(1)}  \widehat{q}_{23}^{(2)}+  \widehat{q}_{13}^{(1)}  \widehat{q}_{3132}^{(2)} \\ && \\
  \dfrac{i}{k}  \widehat{p}_{2}   \widehat{q}_{21}^{(2)}+  \widehat{q}_{21}^{(1)} \widehat{q}_{1213}^{(2)} +  \widehat{q}_{23}^{(1)} \widehat{q}_{31}^{(2)} &
- \dfrac{i}{k}   \widehat{p}_{2}   \widehat{q}_{2123}^{(2)}-   \widehat{q}_{21}^{(1)} \widehat{q}_{12}^{(2)}-  \widehat{q}_{23}^{(1)}  \widehat{q}_{32}^{(2)}&
\dfrac{i}{k} \widehat{p}_{2}    \widehat{q}_{23}^{(2)}-  \widehat{q}_{21}^{(1)}  \widehat{q}_{13}^{(2)}+  \widehat{q}_{23}^{(1)}  \widehat{q}_{3132}^{(2)}\\ && \\
  \dfrac{i}{k}   \widehat{p}_{3}  \widehat{q}_{31}^{(2)}+ \widehat{q}_{31}^{(1)}  \widehat{q}_{1213}^{(2)}+  \widehat{q}_{32}^{(1)}  \widehat{q}_{21}^{(2)} &
\dfrac{i}{k}  \widehat{p}_{3}   \widehat{q}_{32}^{(2)}-    \widehat{q}_{31}^{(1)}  \widehat{q}_{12}^{(2)}+ \widehat{q}_{32}^{(1)}  \widehat{q}_{2123}^{(2)}&
- \dfrac{i}{k} \widehat{p}_{3}   \widehat{q}_{3132}^{(2)}-  \widehat{q}_{31}^{(1)}  \widehat{q}_{13}^{(2)}-  \widehat{q}_{32}^{(1)}  \widehat{q}_{23}^{(2)} \\
\end{bmatrix*},
\end{align*}
and similarly for $\widehat{M}\widehat{L}$ we get
\begin{align*}
&\widehat{M}\widehat{L}= \\
&i k^{2}\begin{bmatrix*}[l]
 - \dfrac{i}{k}  \widehat{q}_{1213}^{(2)}  \widehat{p}_{1}- \widehat{q}_{12}^{(2)}  \widehat{q}_{21}^{(1)}- \widehat{q}_{13}^{(2)}  \widehat{q}_{31}^{(1)} &
\dfrac{i}{k}  \widehat{q}_{12}^{(2)}  \widehat{p}_{2}-  \widehat{q}_{13}^{(2)} \widehat{q}_{32}^{(1)}+ \widehat{q}_{1213}^{(2)} \widehat{q}_{12}^{(1)} &
\dfrac{i}{k} \widehat{q}_{13}^{(2)} \widehat{p}_{3}-  \widehat{q}_{12}^{(2)}  \widehat{q}_{23}^{(1)}+ \widehat{q}_{1213}^{(2)} \widehat{q}_{13}^{(1)} \\ && \\
\dfrac{i}{k}  \widehat{q}_{21}^{(2)} \widehat{p}_{1}-  \widehat{q}_{23}^{(2)}  \widehat{q}_{31}^{(1)}+  \widehat{q}_{2123}^{(2)}  \widehat{q}_{21}^{(1)} & 
-\dfrac{i}{k} \widehat{q}_{2123}^{(2)} \widehat{p}_{2} -   \widehat{q}_{21}^{(2)}  \widehat{q}_{12}^{(1)}-  \widehat{q}_{23}^{(2)} \widehat{q}_{32}^{(1)} & 
\dfrac{i}{k} \widehat{q}_{23}^{(2)} \widehat{p}_{3}-  \widehat{q}_{21}^{(2)}  \widehat{q}_{13}^{(1)}+  \widehat{q}_{2123}^{(2)} \widehat{q}_{23}^{(1)}\\ && \\
\dfrac{i}{k} \widehat{q}_{31}^{(2)}  \widehat{p}_{1}-  \widehat{q}_{32}^{(2)}  \widehat{q}_{21}^{(1)}+  \widehat{q}_{3132}^{(2)}  \widehat{q}_{31}^{(1)} & 
\dfrac{i}{k}  \widehat{q}_{32}^{(2)}\widehat{p}_{2}-  \widehat{q}_{31}^{(2)} \widehat{q}_{12}^{(1)}+ \widehat{q}_{3132}^{(2)} \widehat{q}_{32}^{(1)} & 
- \dfrac{i}{k} \widehat{q}_{3132}^{(2)} \widehat{p}_{3}-  \widehat{q}_{31}^{(2)}  \widehat{q}_{13}^{(1)}-  \widehat{q}_{32}^{(2)}  \widehat{q}_{23}^{(1)}\\
\end{bmatrix*}.
\end{align*}
The reader is now invited to consider the two products above and observe the significance of the order of occurrence of the quantum momentum operators $\widehat{p}_{r}$. There are several patterns we shall point out. First we focus on the diagonal elements of the products. They are of the form
\begin{align*}
 (\widehat{L}\widehat{M})_{rr} = k \widehat{p}_{r}  \widehat{q}_{rsrt}^{(2)} -ik^{2} \widehat{q}_{rs}^{(1)} \widehat{q}_{sr}^{(2)} -i k^{2}  \widehat{q}_{rt}^{(1)}\widehat{q}_{tr}^{-2} \quad \text{and} \quad (\widehat{M}\widehat{L})_{rr} = k  \widehat{q}_{rsrt}^{(2)} \widehat{p}_{r} -ik^{2} \widehat{q}_{rs}^{(2)} \widehat{q}_{sr}^{(1)} - ik^{2}  \widehat{q}_{rt}^{(2)}\widehat{q}_{tr}^{-1},
\end{align*}
 where we can write terms the form $- \widehat{q}_{rs}^{(1)} \widehat{q}_{sr}^{(2)}$ as $ \widehat{q}_{sr}^{(3)}$, $r \neq s,t $ due to the permutation of the indices in the odd power term caused by the negative sign. Computing the difference of $\widehat{L}\widehat{M}$ and $\widehat{M}\widehat{L}$ and multiplying by $i$ we then get
\begin{align*}
i\left(\widehat{L}\widehat{M} - \widehat{M}\widehat{L} \right)_{rr} = ik \left[ \widehat{p}_{r},   \widehat{q}_{rsrt}^{(2)}\right] + 2k^{2} \widehat{q}_{rs}^{(3)} + 2k^{2} \widehat{q}_{rt}^{(3)},
\end{align*}
since $i\left(ik^{2} \widehat{q}_{rs}^{(3)}-(-ik^{2} \widehat{q}_{rs}^{(3)})\right)=2k^{2} \widehat{q}_{rs}^{(3)}$. After applying Lemma \ref{keyLemma} we get
\begin{align}\label{quantumLaxDiagonal}
i\left(\widehat{L}\widehat{M} - \widehat{M}\widehat{L} \right)_{rr} =2k \widehat{q}_{rsrt}^{(3)} + 2k^{2} \widehat{q}_{rsrt}^{(3)}=2g  \widehat{q}_{rsrt}^{(3)},
\end{align}
for $g=k(k+1)$. Now consider the off-diagonal entries. For $r\neq s,t$ we have
\begin{align*}
 (\widehat{L}\widehat{M})_{rs} = -k\widehat{p}_{r} \widehat{q}_{rs}^{(2)}+ik^{2} \widehat{q}_{rs}^{(1)} \widehat{q}_{srst}^{(2)}-ik^{2} \widehat{q}_{rt}^{(1)} \widehat{q}_{ts}^{(2)} \quad \text{and} \quad  
(\widehat{M}\widehat{L})_{rs} =  -k \widehat{q}_{rs}^{(2)}\widehat{p}_{s}-ik^{2} \widehat{q}_{rt}^{(2)} \widehat{q}_{ts}^{(1)}+ik^{2} \widehat{q}_{rsrt}^{(2)} \widehat{q}_{rs}^{(1)},
\end{align*}
which we group by coefficients powers of $k$ yielding
\begin{align*}
 (\widehat{L}\widehat{M})_{rs} &= -k\widehat{p}_{r} \widehat{q}_{rs}^{(2)}+ik^{2}\left(\underbrace{ \widehat{q}_{rs}^{(1)} \widehat{q}_{sr}^{(2)}}_{\alpha}+ \widehat{q}_{rs}^{(1)} \widehat{q}_{st}^{(2)}- \widehat{q}_{rt}^{(1)} \widehat{q}_{ts}^{(2)}\right) \quad \text{and} \\
(\widehat{M}\widehat{L})_{rs} &=  -k \widehat{q}_{rs}^{(2)}\widehat{p}_{s}+ik^{2}\left(- \widehat{q}_{rt}^{(2)} \widehat{q}_{ts}^{(1)}+\underbrace{ \widehat{q}_{rs}^{(2)} \widehat{q}_{rs}^{(1)}}_{\alpha}+ \widehat{q}_{rt}^{(2)} \widehat{q}_{rs}^{(1)}\right) ,
\end{align*}
where the underbraced terms indicate opposite signs, which will result in cancellation immediately. Now consider for a moment the rest of the momentum-free terms and ignore the coefficient $k^{2}$. Computing the commutator and multiplying by $i$ we get
\begin{align*}
- \widehat{q}_{rs}^{(1)} \widehat{q}_{st}^{(2)}+ \widehat{q}_{rt}^{(1)} \widehat{q}_{ts}^{(2)}+ \widehat{q}_{rt}^{(2)} \widehat{q}_{ts}^{(1)}- \widehat{q}_{rt}^{(2)} \widehat{q}_{rs}^{(1)} &=  \widehat{q}_{st}^{(2)}( \widehat{q}_{rs}^{(1)}- \widehat{q}_{rt}^{(1)}) + \widehat{q}_{rt}^{(2)}( \widehat{q}_{rs}^{(1)}- \widehat{q}_{ts}^{(1)})  \\
&=  \widehat{q}_{st}^{(2)} \left( \dfrac{\cancel{\widehat{q}_{r}}-\widehat{q}_{t}-\cancel{\widehat{q}_{r}}+\widehat{q}_{s}}{\widehat{q}_{rs}\widehat{q}_{rt}} \right) + \widehat{q}_{rt}^{(2)} \left( \dfrac{\widehat{q}_{t}-\cancel{\widehat{q}_{s}}-\widehat{q}_{r}+\cancel{\widehat{q}_{s}}}{\widehat{q}_{rs}\widehat{q}_{ts}} \right) \\
&=  \widehat{q}_{st}^{(2)} \dfrac{\widehat{q}_{st}}{\widehat{q}_{rs}\widehat{q}_{rt}} +  \widehat{q}_{rt}^{(2)} \dfrac{\widehat{q}_{tr}}{\widehat{q}_{rs}\widehat{q}_{ts}},
\end{align*}
which enables us to cancel the quadratic terms. Hence, we get
\begin{align*}
- \widehat{q}_{rs}^{(1)} \widehat{q}_{st}^{(2)}+ \widehat{q}_{rt}^{(1)} \widehat{q}_{ts}^{(2)}+ \widehat{q}_{rt}^{(2)} \widehat{q}_{ts}^{(1)}- \widehat{q}_{rt}^{(2)} \widehat{q}_{rs}^{(1)} &= \widehat{q}_{st}^{-1} \widehat{q}_{rs}^{(1)} \widehat{q}_{rt}^{(1)} +  \widehat{q}_{rt}^{(1)} \widehat{q}_{rs}^{(1)} \widehat{q}_{ts}^{(1)} = 0,
\end{align*}
since $\widehat{q}_{st}+ \widehat{q}_{ts}^{(1)}=0$. Thus, computing the difference of the two products and multiplying by the imaginary unit yields
\begin{align}
i\left(\widehat{L}\widehat{M} - \widehat{M}\widehat{L} \right)_{rs} &= -ik\left( \widehat{p}_{r} \widehat{q}_{rs}^{(2)}- \widehat{q}_{rs}^{(2)}\widehat{p}_{s} \right)   = -ik(-i(\partial_{\widehat{q}_{r}}  \widehat{q}_{rs}^{(2)} - (-i \partial_{\widehat{q}_{s}}  \widehat{q}_{rs}^{(2)} +2i  \widehat{q}_{rs}^{(3)} )) \nonumber \\ 
&= -ik (-i( \partial_{\widehat{q}_{r}}-\partial_{\widehat{q}_{s}})) -2ik  \widehat{q}_{rs}^{(3)},
 \label{quantumLaxOffdiagonal}
\end{align}
where we have used Lemma \ref{keyLemma} via $ \widehat{q}_{rs}^{(2)}\widehat{p}_{s} = \widehat{p}_{s}  \widehat{q}_{rs}^{(2)} + 2i  \widehat{q}_{rs}^{(3)}$, after applying the least common denominator in the momentum-free term.
Comparing the expression in \eqref{quantumLaxDiagonal} with the first Equation of \eqref{quantumLaxDot} we get
\begin{align*}
\dot{\widehat{L}}_{rr} = \dot{\widehat{p}}_{r}  = 2g  \widehat{q}_{rsrt}^{(3)} = \sum_{\substack{ s = 1 \\ s \neq r}}^{3} 2g \widehat{q}_{rs}^{(3)} = i\left[ \widehat{L}, \widehat{M} \right]_{rr},
\end{align*}
which indeed recovers Heisenberg's equation of motion for the momentum operator as in \eqref{quantumHeisenbergEqns}. Hence, comparing the expression derived in Equation \eqref{quantumLaxOffdiagonal} with the second statement of \eqref{quantumLaxDot} we see that
\begin{align*}
\dot{\widehat{L}}_{rs} = -ik ( \dot{\widehat{q}}_{r}-\dot{\widehat{q}}_{s}) \widehat{q}_{rs}^{(2)} -2ik  \widehat{q}_{rs}^{(3)} = -ik (-i( \partial_{\widehat{q}_{r}}-\partial_{\widehat{q}_{s}})) -2ik  \widehat{q}_{rs}^{(3)} = i\left(\widehat{L}\widehat{M} - \widehat{M}\widehat{L} \right)_{rs},
\end{align*}
implying that
\begin{align*}
\dot{\widehat{q}}_{r} = -i \partial_{\widehat{q}_{r}} = \widehat{p}_{r}
\end{align*}
by Equation \eqref{quantumMomentum} and we recover Heisenberg's equation for the coordinate as in \eqref{quantumHeisenbergQ}. Thus, we conclude that the proposed Lax pair indeed satisfies the Lax equation \eqref{quantumLax1}. 
\end{proof}
\begin{remark}
Notice that the Lax matrix $\widehat{M}$ exhibits the following property
\begin{equation}\label{sumToZero}
\sum_{j=1}^{n} \widehat{M}_{ij} = \sum_{j=1}^{n} \widehat{M}_{ji} = 0 \qquad \text{ for any }  i,
\end{equation}
which we refer to as \emph{sum-to-zero condition} and is the extra assumption we mentioned in Subsection \ref{quantumLaxSection} \cite{Wadati}. Next we observe that
\begin{align*}
\left[ \tilde{H}, \sum_{j,k=1}^{n}  \widehat{L}^{m}_{jk}\right] = \sum_{j,k=1}^{n} \left[ \widehat{L}^{m}, \widehat{M} \right]_{jk} = 0,
\end{align*}
which shows that a set of quantum integrals of motion is given by
\begin{align*}
J_{m} = \dfrac{1}{m} \sum_{j,k=1}^{n}  \widehat{L}^{m}_{jk},
\end{align*}
that is, by summing the powers of elements of the Lax matrix $\widehat{L}$ \cite{Ujino}. Note that when first introduced, the factor of $m^{-1}$ in front of the summation was introduced ``merely for convenience'' \cite{Ujino}. But as we shall see below, it is crucial in order to identify the Hamiltonian of the system with one of the integrals of motion. 
\end{remark}
Hence, we proceed by identifying the integrals of motion and showing that they commute.
\begin{corollary}
The $n=3$ quantum rational Calogero-Moser system is integrable.
\end{corollary}
\begin{proof}
We present the first three integrals of motion of the quantum rational Calogero-Moser system, namely
\begin{align*}
J_{1} &= \sum_{r,s=1}^{n} \widehat{L}_{rs} =     \widehat{p}_{1} + ik \widehat{q}_{12}^{(1)} + ik \widehat{q}_{13}^{(1)}+  ik \widehat{q}_{21}^{(1)} +   \widehat{p}_{2} + ik \widehat{q}_{23}^{(1)} +   ik \widehat{q}_{31}^{(1)} + ik \widehat{q}_{32}^{(1)} +  \widehat{p}_{3} = \sum_{r=1}^{3} \widehat{p}_{r},
\end{align*}
where we have used $ \widehat{q}_{rs}^{(1)} =  \widehat{q}_{sr}^{(1)}$ to cancel the terms with coefficients $ik$. In order to compute the second integral of motion $J_{2}$ we derive the explicit form of $\widehat{L}^{2}$ as
\begin{align*}
&\widehat{L}^{2} =  \begin{bmatrix}
  \widehat{p}_{1} & ik \widehat{q}_{12}^{(1)} & ik \widehat{q}_{13}^{(1)} \\
  ik \widehat{q}_{21}^{(1)} &   \widehat{p}_{2} & ik \widehat{q}_{23}^{(1)} \\
  ik \widehat{q}_{31}^{(1)} & ik \widehat{q}_{32}^{(1)} &   \widehat{p}_{3}
 \end{bmatrix} \begin{bmatrix}
  \widehat{p}_{1} & ik \widehat{q}_{12}^{(1)} & ik \widehat{q}_{13}^{(1)} \\
  ik \widehat{q}_{21}^{(1)} &   \widehat{p}_{2} & ik \widehat{q}_{23}^{(1)} \\
  ik \widehat{q}_{31}^{(1)} & ik \widehat{q}_{32}^{(1)} &   \widehat{p}_{3}
 \end{bmatrix}\\ & \\
&=\begin{bmatrix*}
\widehat{p}_{1}^{2}+k^{2} \widehat{q}_{1213}^{(2)} & ik\widehat{p}_{1} \widehat{q}_{12}^{(1)}+ik \widehat{q}_{12}^{(1)}\widehat{p}_{2}-k^{2} \widehat{q}_{1332}^{(1)} & ik\widehat{p}_{1} \widehat{q}_{13}^{(1)}-k^{2} \widehat{q}_{1223}^{(1)}+ik \widehat{q}_{13}^{(1)}\widehat{p}_{3} \\ && \\
ik \widehat{q}_{21}^{(1)}\widehat{p}_{1}+ik\widehat{p}_{2} \widehat{q}_{21}^{(1)} - k^{2} \widehat{q}_{2331}^{(1)} & k^{2} \widehat{q}_{21}^{(2)}+\widehat{p}_{2}^{2}+k^{2} \widehat{q}_{23}^{(2)} & -k^{2} \widehat{q}_{2113}^{(1)}+ik\widehat{p}_{2} \widehat{q}_{23}^{(1)}+ik \widehat{q}_{23}^{(1)}\widehat{p}_{3} \\ && \\
ik \widehat{q}_{31}^{(1)}\widehat{p}_{1}-k^{2} \widehat{q}_{3221}^{(1)}+ik\widehat{p}_{3} \widehat{q}_{31}^{(1)} & -k^{2} \widehat{q}_{3112}^{(1)}+ik \widehat{q}_{32}^{(1)}\widehat{p}_{2}+ik\widehat{p}_{3} \widehat{q}_{32}^{(1)} & k^{2} \widehat{q}_{3132}^{(2)}+\widehat{p}_{3}^{2}
\end{bmatrix*}.
\end{align*}
To simplify the calculations, we take advantage of decomposing the sum via
\begin{align*}
J_{m} = \dfrac{1}{m} \sum_{r=1}^{n} \widehat{L}_{rr}^{m} + \dfrac{1}{m} \sum_{\substack{r,s=1 \\ s \neq r}}^{n} \widehat{L}^{m}_{rs},
\end{align*}
as in \cite{Ujino}. Summing the diagonal elements of the resulting matrix using the above we get
\begin{align}\label{quantumLaxIntegral21}
\dfrac{1}{2} \sum_{r=1}^{3} \widehat{L}^{2}_{rr} = \dfrac{1}{2}( \widehat{p}_{1}^{2} + \widehat{p}_{2}^{2} + \widehat{p}_{3}^{2} )  +\dfrac{1}{2} ( 2 k^{2}  \widehat{q}_{12}^{(2)} + 2 k^{2}  \widehat{q}_{23}^{(2)} + 2 k^{2}  \widehat{q}_{13}^{(2)} ) = \dfrac{1}{2} \sum_{r=1}^{3} p_{r}^{2} + \sum_{\substack{s =1 \\ s \neq r }}^{3} k^{2}  \widehat{q}_{rs}^{(2)} ,
\end{align}
where we have used $ \widehat{q}_{rs}^{(2)} = -  \widehat{q}_{sr}^{(2)}$. Now summing the off-diagonal elements we get
\begin{align*}
 \dfrac{1}{2} \sum_{\substack{r,s=1 \\ s \neq r}}^{3} \widehat{L}^{2}_{rs}= & ik \left( [\widehat{p}_{1},  \widehat{q}_{12}^{(1)} ] +[ \widehat{p}_{2},  \widehat{q}_{21}^{(1)} ] + [ \widehat{p}_{1},  \widehat{q}_{13}^{(1)} ] + [\widehat{p}_{3},  \widehat{q}_{31}^{(1)}] + [\widehat{p}_{2},  \widehat{q}_{23}^{(1)} ] + [\widehat{p}_{3},  \widehat{q}_{32}^{(1)}] \right)  \\
& \quad - k^{2} \left(  \widehat{q}_{13}^{(1)}+  \widehat{q}_{32}^{(1)} +  \widehat{q}_{12}^{(1)} +  \widehat{q}_{23}^{(1)} +  \widehat{q}_{23}^{(1)} +  \widehat{q}_{31}^{(1)} +  \widehat{q}_{13}^{(1)} +  \widehat{q}_{32}^{(1)} +  \widehat{q}_{21}^{(1)} +  \widehat{q}_{31}^{(1)} +  \widehat{q}_{12}^{(1)}  \right),
\end{align*}
since we can rewrite elements involving multiplication by the quantum momenta operators $\widehat{p}_r$ as
\begin{align*}
ik \widehat{p}_{r}  \widehat{q}_{rs}^{(1)} + ik  \widehat{q}_{sr}^{(1)} \widehat{p}_{r} =ik\widehat{p}_{r}  \widehat{q}_{rs}^{(1)} -ik  \widehat{q}_{rs}^{(1)}  \widehat{p}_{r}  = ik [\widehat{p}_{r}, \widehat{q}_{rs}^{(1)}].
\end{align*}
Using the fact that $ \widehat{q}_{rs}^{(1)} =  \widehat{q}_{sr}^{(1)}$ we cancel all terms with coefficient $k^{2}$ to obtain
\begin{align*}
 \dfrac{1}{2} \sum_{\substack{r,s=1 \\ s \neq r}}^{3} \widehat{L}^{2}_{rs}= & ik \left( [\widehat{p}_{1},  \widehat{q}_{12}^{(1)} ] +[ \widehat{p}_{2},  \widehat{q}_{21}^{(1)} ] + [ \widehat{p}_{1},  \widehat{q}_{13}^{(1)} ] + [\widehat{p}_{3},  \widehat{q}_{31}^{(1)}] + [\widehat{p}_{2},  \widehat{q}_{23}^{(1)} ] + [\widehat{p}_{3},  \widehat{q}_{32}^{(1)}] \right),
\end{align*}
and upon applying Lemma \ref{keyLemma} we have
\begin{align}\label{quantumLaxIntegral22}
 \dfrac{1}{2} \sum_{\substack{r,s=1 \\ s \neq r}}^{3} \widehat{L}^{2}_{rs}= ik \sum_{\substack{r,s=1 \\ r \neq s}} ( -i  \widehat{q}_{rs}^{(2)} ) .
\end{align}
Then adding Equations \eqref{quantumLaxIntegral21} and \eqref{quantumLaxIntegral22} we get

\begin{align*}
J_{2} &= \sum_{r,s=1}^{3}( \widehat{L}^{2} )_{rs} = \dfrac{1}{2} \sum_{r=1}^{3} \widehat{p}_{r}^{2} + \sum_{\substack{s = 1 \\ s \neq r}}^{3} k(k+1) \widehat{q}_{rs}^{(2)} = \widehat{H}. \\
\end{align*}
In order to compute $\widehat{L}^{3}$ and $J_{3}$ we use the Wolfram Mathematica code in Appendix \ref{appQuantumRational}. The third integral $J_{3}$ is given by
\begin{align*}
J_{3} = \dfrac{1}{3} \sum_{r,s=1}^{3} \left( \widehat{L}^{3} \right)_{rs} = \dfrac{1}{3} \sum_{r=1}^{3}\widehat{p}_{r}^{3} + \dfrac{1}{3}g \sum_{\substack{s=1 \\ s \neq r}}^{3} \left( \widehat{p}_{r}  \widehat{q}_{rs}^{(2)} +  \widehat{q}_{rs}^{(1)}\widehat{p}_{r} \widehat{q}_{rs}^{(1)} +  \widehat{q}_{rs}^{(2)}\widehat{p}_{r} \right),
\end{align*}
where once again we point out the importance of the order of appearance of the quantum momentum operator $\widehat{p}_{r}$. We now check that the first two integrals commute. 
\begin{align*}
[ J_{1}, J_{2} ] &= \left[ (\widehat{p}_{1} + \widehat{p}_{2} + \widehat{p}_{3} ), \widehat{H} \right]= [\widehat{p}_{1}, \widehat{H}] + [\widehat{p}_{2}, \widehat{H} ] + [\widehat{p}_{3}, \widehat{H} ].
\end{align*}
We shall illustrate what happens to each commutator separately. By Lemma \ref{keyLemma} we have
\begin{align*}
[\widehat{p}_{1}, \widehat{H} ] = 2ig \left( - \widehat{q}_{12}^{(3)}- \widehat{q}_{13}^{(3)} \right), \quad [\widehat{p}_{2}, \widehat{H} ] = 2ig (-\widehat{q}_{21}^{(3)} -  \widehat{q}_{23}^{(3)} ), \quad \text{ and } \quad  [\widehat{p}_{3}, \widehat{H} ] = 2ig (- \widehat{q}_{31}^{(3)} -  \widehat{q}_{32}^{(3)} ).
\end{align*}
Adding all three commutators we get
\begin{align*}
[ J_{1}, J_{2} ] &=2ig  \left( - \widehat{q}_{12}^{(3)}- \widehat{q}_{13}^{(3)}-\widehat{q}_{21}^{(3)} -  \widehat{q}_{23}^{(3)} - \widehat{q}_{31}^{(3)} -  \widehat{q}_{32}^{(3)} \right) = 0,
\end{align*}
which all cancel since $\widehat{q}_{rs}^{-(2k+1)} - \widehat{q}_{sr}^{-(2k+1)} = 0$, $\forall k \in \mathbb{Z}$. Next we check the commutativity of $J_{1}$ and $J_{3}$.
 Writing out the explicit form of $J_{3}$ we have
\begin{align*}
J_{3} = \widehat{p}_{1}^{3} + \widehat{p}_{2}^{3} + \widehat{p}_{3}^{3} +  \dfrac{1}{3} g \begin{cases}\begin{drcases}
\widehat{p}_{1} \widehat{q}_{1213}^{(2)} +  \widehat{q}_{12}^{(1)} \widehat{p}_{2}\widehat{q}_{12}^{(1)} +   \widehat{q}_{13}^{(1)} \widehat{p}_{3}\widehat{q}_{13}^{(1)}  +  \widehat{q}_{12}^{(2)} \widehat{p}_{2} + \widehat{q}_{13}^{(2)} \widehat{p}_{3}  \\ \\
+ \widehat{p}_{2}  \widehat{q}_{2123}^{(2)} +  \widehat{q}_{21}^{(1)} \widehat{p}_{1}\widehat{q}_{21}^{(1)} +   \widehat{q}_{23}^{(1)} \widehat{p}_{3}\widehat{q}_{23}^{(1)}  +  \widehat{q}_{21}^{(2)} \widehat{p}_{1} + \widehat{q}_{23}^{(2)} \widehat{p}_{3} \\ \\
 +\widehat{p}_{3}  \widehat{q}_{3132}^{(2)} +   \widehat{q}_{31}^{(1)} \widehat{p}_{1}\widehat{q}_{31}^{(1)} +   \widehat{q}_{32}^{(1)} \widehat{p}_{2}\widehat{q}_{32}^{(1)}  +  \widehat{q}_{31}^{(2)} \widehat{p}_{1} + \widehat{q}_{32}^{(2)} \widehat{p}_{2}
\end{drcases}
\end{cases}.
\end{align*}
Now applying Lemma \ref{keyLemma} to the middle terms, we get
\begin{align*}
J_{3} = \widehat{p}_{1}^{3} + \widehat{p}_{2}^{3} + \widehat{p}_{3}^{3} +  \dfrac{1}{3} g \begin{cases}\begin{drcases}
\widehat{p}_{1} \widehat{q}_{1213}^{(2)} +  \widehat{q}_{12}^{(1)}\left(\widehat{q}_{12}^{(1)} \widehat{p}_{2}- \cancel{i \widehat{q}_{12}^{(2)}}\right) +  \widehat{q}_{13}^{(1)}\left(\widehat{q}_{13}^{(1)} \widehat{p}_{2}-\bcancel{i \widehat{q}_{13}^{(2)}}\right)  +  \widehat{q}_{12}^{(2)} \widehat{p}_{2} + \widehat{q}_{13}^{(2)} \widehat{p}_{3}   \\ \\
+ \widehat{p}_{2}  \widehat{q}_{2123}^{(2)} +  \widehat{q}_{21}^{(1)}\left(\widehat{q}_{21}^{(1)} \widehat{p}_{1}-\cancel{i \widehat{q}_{21}^{(2)}}\right) +  \widehat{q}_{23}^{(1)}\left(\widehat{q}_{23}^{(1)} \widehat{p}_{3}-\xcancel{i \widehat{q}_{23}^{(2)}}\right) + \widehat{q}_{21}^{(2)} \widehat{p}_{1} + \widehat{q}_{23}^{(2)} \widehat{p}_{3}   \\ \\
 +\widehat{p}_{3}  \widehat{q}_{3132}^{(2)} + \widehat{q}_{31}^{(1)}\left(\widehat{q}_{31}^{(1)} \widehat{p}_{1}-\bcancel{i \widehat{q}_{31}^{(2)}}\right) +  \widehat{q}_{32}^{(1)}\left(\widehat{q}_{32}^{(1)} \widehat{p}_{2}-\xcancel{i \widehat{q}_{32}^{(2)}}\right) + \widehat{q}_{31}^{(2)} \widehat{p}_{1} + \widehat{q}_{32}^{(2)} \widehat{p}_{2}
\end{drcases}
\end{cases},
\end{align*}
where we have canceled the respectively crossed terms via $\widehat{q}_{rs}^{-(2k+1)} - \widehat{q}_{sr}^{-(2k+1)} = 0$, $\forall k \in \mathbb{Z}$. Hence, we rewrite $J_{3}$ as follows
\begin{align*}
J_{3} & = \widehat{p}_{1}^{3} + \widehat{p}_{2}^{3} + \widehat{p}_{3}^{3} +  \dfrac{1}{3} g \begin{cases}\begin{drcases}
\widehat{p}_{1} \widehat{q}_{1213}^{(2)} + 2\widehat{q}_{12}^{(2)} \widehat{p}_{2} + 2 \widehat{q}_{13}^{(2)} \widehat{p}_{3}  \\ \\
+ \widehat{p}_{2}  \widehat{q}_{2123}^{(2)} +  2\widehat{q}_{21}^{(2)} \widehat{p}_{1} + 2\widehat{q}_{23}^{(2)} \widehat{p}_{3}  \\ \\
 +\widehat{p}_{3}  \widehat{q}_{3132}^{(2)} +  2\widehat{q}_{31}^{(2)} \widehat{p}_{1} + 2\widehat{q}_{32}^{(2)} \widehat{p}_{2}
\end{drcases}
\end{cases} \\
&= \sum_{r=1}^{3} \left( \widehat{p}_{r}^{3} + \dfrac{1}{3} g \sum_{\substack{s=1 \\ s \neq r}}^{3} \left( \widehat{p}_{r} \widehat{q}_{rs}^{(2)} + 2\widehat{q}_{rs}^{(2)} \widehat{p}_{s} \right) \right).
\end{align*}
Next we consider the commutator of $J_{1}$ and $J_{3}$. Note that all terms involving cubed momenta $\widehat{p}_{r}^{3}$ for $r=1,2,3$ commute with $\widehat{p}_{s}$ for $s=1,2,3$ due to properties of ordinary partial derivatives. Hence, we proceed by analyzing the commutator of $J_{1}$ and the remaining terms by splitting them in two cases as follows.
\begin{align*}
[ J_{1}, J_{3} ] = \left[ \sum_{r=1}^{3} \widehat{p}_{r}, \widehat{p}_{s} \widehat{q}_{st}^{(2)} \right] +  \left[ \sum_{r=1}^{3} \widehat{p}_{r}, 2\widehat{q}_{st}^{(2)} \widehat{p}_{t} \right] = \left[ (\widehat{p}_{s} + \widehat{p}_{t} ), \widehat{p}_{s} \widehat{q}_{st}^{(2)} \right] + \left[ (\widehat{p}_{s} + \widehat{p}_{t} ),2\widehat{q}_{st}^{(2)} \widehat{p}_{t} \right],
\end{align*}
since the commutator $[ \widehat{p}_{r}, f(\widehat{q}_{st})] =0$ for any $f(\widehat{q}_{st})$, $r \neq s,t$. Dealing with each expression separately we have
\begin{align*}
 \left[ (\widehat{p}_{s} + \widehat{p}_{t} ), \widehat{p}_{s} \widehat{q}_{st}^{(2)} \right]  &= (\widehat{p}_{s} + \widehat{p}_{t}) \widehat{p}_{s} \widehat{q}_{st}^{(2)} - \widehat{p}_{s} \widehat{q}_{st}^{(2)} ( \widehat{p}_{s} + \widehat{p}_{t}) = \widehat{p}_{s} (\widehat{p}_{s} + \widehat{p}_{t}) \widehat{q}_{st}^{(2)} - \widehat{p}_{s} \widehat{q}_{st}^{(2)} ( \widehat{p}_{s} + \widehat{p}_{t})  \\
&= \widehat{p}_{s} \left( (\widehat{p}_{s}+\widehat{p}_{t})\widehat{q}_{st}^{(2)} - \widehat{q}_{st}^{(2)} ( \widehat{p}_{s}+\widehat{p}_{t})  \right) = \widehat{p}_{s} \left( -i \left( -2\widehat{q}_{st}^{(3)} - (-2) \widehat{q}_{st}^{(3)} \right) \right) = 0,
\end{align*}
where we have switched the order of multiplication by $(\widehat{p}_{s} + \widehat{p}_{t} )$ and $\widehat{p}_{s} $ due to properties of ordinary differentials and then applied Lemma \ref{keyLemma}. Similarly, for the second commutator we get
\begin{align*}
 \left[ (\widehat{p}_{s} + \widehat{p}_{t} ),2\widehat{q}_{st}^{(2)} \widehat{p}_{t} \right] &= (\widehat{p}_{s} + \widehat{p}_{t} )2\widehat{q}_{st}^{(2)} \widehat{p}_{t} - 2\widehat{q}_{st}^{(2)} \widehat{p}_{t}  (\widehat{p}_{s} + \widehat{p}_{t} ) =  (\widehat{p}_{s} + \widehat{p}_{t} )2\widehat{q}_{st}^{(2)} \widehat{p}_{t} - 2\widehat{q}_{st}^{(2)} (\widehat{p}_{s} + \widehat{p}_{t} ) \widehat{p}_{t} \\
&= \left(  (\widehat{p}_{s} + \widehat{p}_{t} )2\widehat{q}_{st}^{(2)} - 2\widehat{q}_{st}^{(2)} (\widehat{p}_{s} + \widehat{p}_{t} )\right) \widehat{p}_{t} = \left( 2 (-2) \widehat{q}_{st}^{(3)} - 2 (-2)(-1)\widehat{q}_{st}^{(3)} \right) \widehat{p}_{t} = 0,
\end{align*}
again since we switched the order of multiplication by $(\widehat{p}_{s}+\widehat{p}_{t})$ and $\widehat{p}_{t}$ and applied Lemma \ref{keyLemma}. Thus, we get that
\begin{align*}
[ J_{1}, J_{3} ] = \left[ (\widehat{p}_{s} + \widehat{p}_{t} ), \widehat{p}_{s} \widehat{q}_{st}^{(2)} \right] + \left[ (\widehat{p}_{s} + \widehat{p}_{t} ),2\widehat{q}_{st}^{(2)} \widehat{p}_{t} \right] = 0 + 0 = 0.
\end{align*}
Using the Jacobi identity and the fact that $J_{1}$ commutes with $J_{2}$ and $J_{3}$ we have
\begin{align*}
\left[ J_{1}, \left[ J_{2}, J_{3} \right] \right] = - \left[ J_{2}, \left[ J_{3}, J_{1} \right] \right] - \left[ J_{3}, \left[ J_{1}, J_{2} \right] \right] = 0 \implies \left[ J_{2}, J_{3} \right] =0,
\end{align*}
which in turn shows that all three integrals of motion of the quantum rational Calogero-Moser commute and thus, the system is integrable. This statement can be generalized in the $n$-dimensional case via induction, see \cite{Wadati}.
\end{proof}

\section{Coxeter and Weyl Groups.}\label{sec:coxeter}
\indent
\par
In this section we study finite groups generated by reflections and a special type of vectors, which are orthogonal to the hyperplanes we reflect about. These vectors are known as \emph{roots} and provide an efficient set of generators for the reflection group, which in turn enables us to identify the group as a \emph{Coxeter group}. Here we refer to \cite{Humphreys}, unless stated otherwise.
\par
\subsection{Reflection Groups and Root Systems.}
Let $V$ be an $n$-dimensional real vector space with the usual inner product $\langle \cdot, \cdot \rangle$ and orthonormal basis $e_{1}, \dots, e_{n}$. Let $\alpha\in V$, $\alpha \neq 0$. Let $H^{\alpha}$ denote the hyperplane orthogonal to $\alpha$ as shown below in Figure \ref{fig:orthogonalReflection}.
\begin{figure}[H]
\begin{center}
\begin{tikzpicture}[plane/.style={trapezium,draw,fill=black!20,trapezium left angle=60,trapezium right angle=120,minimum height=2cm},scale=0.7]
\node (p)[plane] at (0,0){.};
\node at (1,0){$H^{\alpha}$};
\node at (-0.5,0.5)  {$\alpha$};
\draw (p.center) edge ++(0,2cm) edge[densely dashed]  (p.south) (p.south) edge ++(0,-1cm);
\end{tikzpicture}
\end{center}\caption{A vector $\alpha \in V$ and its orthogonal hyperplane $H^{\alpha}$.}\label{fig:orthogonalReflection}
\end{figure}
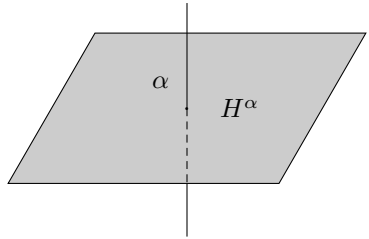
\begin{definition}\label{definitionOrthogonalReflection}
The \emph{orthogonal reflection} $\sigma_{\alpha}$ of a vector $\lambda \in V$ about $H^{\alpha}=\langle \alpha, x \rangle = 0$, where $x = (x_{1}, \dots, x_{n})$, is given by the formula
\begin{align*}
\sigma_{\alpha} \lambda =  \lambda - \dfrac{2 \langle \lambda, \alpha \rangle}{\langle \alpha, \alpha \rangle} \alpha.
\end{align*}
\end{definition}
\begin{figure}[H]
\begin{center}
\begin{tikzpicture}[plane/.style={trapezium,draw,fill=black!20,trapezium left angle=60,trapezium right angle=120,minimum height=3cm},scale=0.7]
\node (p)[plane] at (0,0){.};
\node at (1,0){$H^{\alpha}$};
\node at (-0.5,0.5)  {$\alpha$};
\draw (p.center) edge ++(0,3cm) edge[densely dashed]  (p.south) (p.south) edge ++(0,-1cm);
\draw (p.center) edge ++(2cm,2cm) edge (p.north)  edge[densely dashed] ++(2cm,-2cm);
\node at (2,2) [below] {$\lambda$};
\node at (2.3,-2) [above] {$\sigma_{\alpha} \lambda$};
\end{tikzpicture}
\end{center}\caption{A vector $\lambda \in V$ and its orthogonal reflection $\sigma_{\alpha} \lambda$ about $H^{\alpha}$.}\label{fig:orthogonalReflectionVector}
\end{figure}
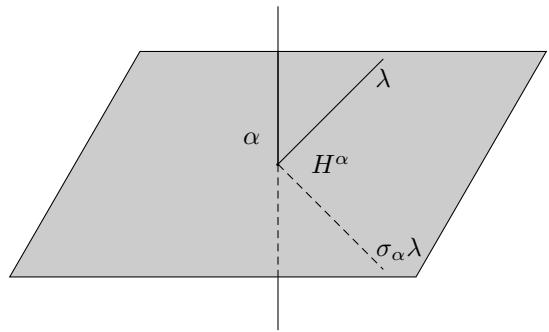

\par The setting of Definition \ref{definitionOrthogonalReflection} is illustrated in Figure \ref{fig:orthogonalReflectionVector}. We clearly have $\sigma_{\alpha} (\alpha) = -\alpha$ and $\sigma_{\alpha} (\lambda) = \lambda$ if $\langle \lambda, \alpha \rangle = 0$. Clearly $\sigma_{\alpha} \in O(V)$, the group of all orthogonal transformations and $\sigma_{\alpha}^{2} = \id$, where $\id$ is the identity transformation in $O(V)$. As stated in the beginning of this section, reflections generate groups. We formalize this statement in the following definition.

\begin{definition}
A \emph{finite reflection group} $W$ is a subgroup of $O(V)$ which is generated by a set of orthogonal reflections.
\end{definition}

\par In order to understand the properties and the structure of a finite reflection group $W$, we examine its actions on the euclidean space $V$. As per Definition \ref{definitionOrthogonalReflection}, each reflection fixes the hyperplane $H^{\alpha}$ orthogonal to the vector $\alpha$. Let $L_{\alpha} = \mathbb{R}\alpha$ denote the line orthogonal to $H^{\alpha}$. We now show that $W$ permutes the collection of all such lines.

\begin{proposition}\label{propgroupaction}
Let $g \in O(V)$, let $\alpha \in V$, $\alpha \neq 0$ and let $\sigma_{\alpha}$ be as in Definition \ref{definitionOrthogonalReflection}. Then 
\begin{equation*}
g \sigma_{\alpha} g^{-1} = \sigma_{g \alpha}.
\end{equation*}
In particular, if $ w \in W$ and $\sigma_{\alpha} \in W$, then $w \sigma_{\alpha} w^{-1} = \sigma_{w \alpha} \in W$ .
\end{proposition}

\begin{proof}
Consider the setting in Figure \ref{fig:orthogonalReflectionVector}. It is clear that $g \sigma_{\alpha} g^{-1}$ maps $g\alpha$ to $-g \alpha$, since $g \sigma_{\alpha} g^{-1} g \alpha = g \sigma_{\alpha} \alpha = g (- \alpha) = - \alpha g$. Hence, we just need to make sure that $\sigma_{g \alpha}$ leaves $H^{g \alpha}$ fixed pointwise. Since $g \in O(V)$, we have
\begin{equation*}
\langle g \lambda, g \alpha \rangle = \langle \lambda, \alpha \rangle, \quad \forall \lambda, \alpha \in V,
\end{equation*}
which means that $\lambda \in H^{\alpha}$ if and only if $g \lambda \in H^{g \alpha}$. Conversely,
\begin{align*}
\sigma_{g \alpha} g \lambda = g \sigma_{\alpha} g^{-1} g \lambda = g \sigma_{\alpha} \lambda = g \lambda.
\end{align*}
when $\lambda \in H^{\alpha}$.
\end{proof}

\par Hence, we see that $W$ acts on the lines $L_{\alpha}$ by $w ( L_{\alpha} ) = L_{w \alpha}$, which indicates that $W$ determines the lines $L_{\alpha}$, but not the vectors $\alpha$. The stability under $W$ is an important property, which leads us to identifying the special type of vectors we mentioned in the beginning.  
\begin{definition}
Let $R$ be a finite set of non-zero vectors $R \subset V \setminus \{ 0 \}$ satisfying
\begin{equation*}
 \forall \alpha \in R, R \cap \mathbb{R}\alpha = \{ \alpha, -\alpha \}  \text{ and }  \sigma_{\alpha} (R) = R,
\end{equation*}
 Let $W = \langle \sigma_{\alpha}, \alpha \in R \rangle$. Then $R$ is called a \emph{Coxeter root system} associated with the finite reflection group $W$. The elements of $R$ are called \emph{roots}.
\end{definition}
\par In other words, we require $R$ to be closed under reflections and we only allow the scalar multiples of a root $\alpha$ to be $\pm \alpha$. Below we introduce a more specific type of root systems with an extra condition.
\begin{definition}\label{normalizedrootsystem}
A root system $R$ is said to be \emph{normalized} if $\forall \alpha \in R$, $\langle \alpha, \alpha \rangle = 2$, such that $\sigma_{\alpha}(\beta) \in R$, $\forall \alpha, \beta \in R$.
\end{definition}
\begin{example}\label{normRootSysEx1}
Let $V=\mathbb{R}^{n}$ with standard basis $e_{1}, \dots, e_{n}$ and standard inner product $\langle e_{i}, e_{j} \rangle = \delta_{ij}$, where $\delta$ is the Kronecker delta. The set $R = \{ \pm (e_{i}-e_{j} ) \lvert 1 \leq i < j \leq n\}$ is a normalized root system in V. The associated group $W$ is the symmetric group $\mathcal{S}_{n}$. The reflection in the root $e_{i}-e_{j}$ interchanges the two basis vectors, thus, yielding $e_{j}-e_{i}$. The remaining roots are left fixed. 
\end{example}
\begin{example}
Define the map $z : \mathbb{C} \to \mathbb{R}^{2}$ by $z=x+iy \mapsto (x,y)$. Then the set $R = \{ \sqrt{2} \exp(\pi i j/m) \lvert j = 0,1,\dots, 2m-1\}$ of renormalized $2m^{th}$ roots of unity has associated group of reflections $W = I_{2}(m)$ the dihedral group of order $2m$ with $m$ reflections and $m$ rotations. 
\end{example}

\par
 Next we fix a root system $R \in V$ and a reflection group $W$ generated by the orthogonal reflections $\sigma_{\alpha}$, where $\alpha \in R$, that is, $W = \langle \sigma_{\alpha} , \alpha \in R \rangle $. Note that we have no restrictions on the size of $R$, thus, it might be quite large in comparison with the dimension of the vector space $V$. We proceed by seeking a linearly independent subset of $R$, from which we can recover $R$. Recall the notion of a total ordering of the vector space $V$.
\begin{definition}\label{partialOrdering}
A \emph{total ordering} of the real vector space $V$ is given by the transitive relation $<$, which satisfies the following:
\begin{enumerate}
\item For $\lambda, \mu \in V$, exactly one of the following holds: $\lambda< \mu$, $\lambda = \mu$ or $\mu < \lambda$.
\item For all $\lambda, \mu, \eta \in V$, if $\lambda < \mu$, then $\lambda + \eta < \mu + \eta$.
\item Let $c \in \mathbb{R}$ and $c \neq 0$. If $c>0$ and $\lambda < \mu$, then $c \lambda < c \mu$. If $c<0$, then $c \mu < c \lambda$.
\end{enumerate}
\end{definition}

\par Using the convention given in Definition \ref{partialOrdering}, we say that $\lambda \in V$ is \emph{positive} if $\lambda$ satisfies $0 < \lambda$. We are now ready to define what we mean by a positive root system.

\begin{definition}
A subset $R_{+}$ of $R$, which consists of all positive roots, relative to some total ordering of $V$, is called a \emph{positive root system} 
\end{definition}

\par Furthermore, we introduce the notion of a \emph{simple system} below.

\begin{definition}
A subset $\Delta$ of $R$ is called a \emph{simple system} and its elements are called \emph{simple roots} respectively, if $\Delta$ is a vector space basis for the $\mathbb{R}$-span of $R$ in $V$ and if each $\alpha \in R$ can be expressed as a linear combination of $\Delta$ with coefficients with the same sign. In other words, a positive root $\alpha \in R_{+}$ is called \emph{simple} in $R_{+}$ if $\alpha$ is not of the form $\alpha = x_{1}\alpha_{1}+x_{2}\alpha_{2}$ with $x_{1}, x_{2} \geq 1$ and $\alpha_{1}, \alpha_{2} \in R_{+}$.
\end{definition}

There also exists the notion of the dual of a root and a root system, respectively. We present them both in the definition below.
\begin{definition}\label{corootDef}
The \emph{coroot} of $\alpha \in R$ is given by
\begin{equation*}
\alpha^{\vee} \coloneqq \dfrac{2 \alpha}{\langle \alpha, \alpha \rangle}.
\end{equation*}
The set $R^{\vee}$ of all coroots $\alpha^{\vee}$, $\alpha \in R$ is also a root system in $V$, known as the \emph{dual root system}. 
\end{definition}

\par We are not going into details about the existence of the systems discussed above, but the reader is invited to consult \cite{Humphreys} for more information. We can now conveniently define a presentation of our finite reflection group $W$ in a concise way.

\begin{definition}\label{coxeterGroup}
Let $R$ be a root system and fix a simple system $\Delta \in R$. Then $W$ is generated by the set $S = \{ \sigma_{\alpha}, \alpha \in \Delta \}$ subject only to the relations
\begin{equation*}
( \sigma_{\alpha} \sigma_{\beta} ) ^{m(\alpha, \beta)} = 1, \quad \alpha, \beta \in \Delta,
\end{equation*}
where $m(\alpha, \beta)$ denotes the order of $\sigma_{\alpha} \sigma_{\beta}$ in $W$ for all $\alpha$, $\beta$. A (finite or infinite) group $W$ with such a presentation relative to a generating set $S$ is called a \emph{Coxeter group}. We refer to the pair $(W,S)$ as a \emph{Coxeter system}.
\end{definition}
We now prove a short result about conjugating reflections by an element of the Coxeter group, which will be used a lot in subsequent proofs.
\begin{lemma}\label{conjugateReflection}
Let $g \in W$, let $\sigma_{\alpha}$ be such that $\langle \alpha, q \rangle =0$, and let $\lambda \in V$. Then 
\begin{align*}
g \sigma_{\alpha} g^{-1} = \sigma_{g \alpha} \lambda.
\end{align*} 
\end{lemma}
\begin{proof}
Writing out the action of $\sigma_{\alpha}$ on the vector $\lambda$ explicitly we get
\begin{align*}
\sigma_{\alpha} \lambda = \lambda - \dfrac{2 \langle \lambda, \alpha \rangle}{\langle \alpha, \alpha \rangle} \alpha.
\end{align*}
Then for $g \in W$, we have
\begin{align*}
g \sigma_{\alpha} g^{-1} \lambda &= g \left( g^{-1} (\lambda) - \dfrac{2 \langle g^{-1} (\lambda), \alpha \rangle}{ \langle \alpha, \alpha \rangle} \alpha \right) \\
&= \lambda - \dfrac{2 \langle \lambda, g(\alpha) \rangle}{ \langle g(\alpha), g(\alpha) \rangle} g(\alpha) \\
&= \sigma_{g \alpha} \lambda.
\end{align*}
\end{proof}
\subsection{Coxeter Graphs and Dynkin Diagrams. Types of Root Systems.}
\par Our next aim is to introduce a way of graphically illustrating finite reflection groups. Then we introduce an extra condition, referred to as the \emph{crystallographic condition}, which enables the study of reflection groups in the frame of Lie theory, where they arise as \emph{Weyl groups}.
\par
According to Definition \ref{coxeterGroup}, the presentation of a group $W$ is determined by the set of integers $m(\alpha, \beta)$ up to an isomorphism for $\alpha, \beta \in \Delta$ . A useful way of illustrating this setting is via a \emph{Coxeter graph}.
\begin{definition}\label{coxeterGraph}
Let $W$ be a Coxeter group with presentation as in Definition \ref{coxeterGroup}. A graph $\Gamma$ with vertices in bijection with $\Delta$ and edges joining vertices $\alpha \neq \beta$ for $m(\alpha,\beta) \geq 3$ is called a \emph{Coxeter graph} of $W$. If a pair of vertices are not joined by an edge, then $m(\alpha, \beta) = 2$. If the Coxeter graph $\Gamma$ is connected, we say that the Coxeter system $(W,S)$ is \emph{irreducible}.
\end{definition}
\begin{example}
Let $\mathcal{I}_{2}(m)$ be the dihedral group of order $2m$, consisting of $m$ rotations via multiples of $2 \pi /m$ and $m$ reflections about the diagonal elements of the polygon. Further, let $\mathcal{S}_{n+1}$ be the symmetric group on $n+1$ elements. Then the Coxeter graphs of $\mathcal{I}_{2}(m)$ and $\mathcal{S}_{n+1}$ are given by
\begin{center}
  \begin{tikzpicture}[scale=.4]
    \draw (-2,0) node  {$\mathcal{I}_{2}(m)$};
	\node[circle,fill=black,draw, minimum size=.1cm] (A) at  (0,0) {};
	\node[circle,fill=black,draw, minimum size=.1cm] (B) at  (3,0) {};
	\draw (A) -- (B);
\draw (1.5, 1) node {$m$};
\end{tikzpicture}\quad \text{ and } \quad
\begin{tikzpicture}[scale=.4]
	\draw (7, 0) node {$\mathcal{S}_{n+1}$};
\node[circle, fill=black, draw, minimum size=.2cm] (C) at (9,0) {};
\node[circle, fill=black,draw, minimum size=.2cm] (D) at (12,0) {};
\node[circle,fill=black, draw, minimum size=.2cm] (E) at (17,0) {};
\node[circle, fill=black,draw, minimum size=.2cm] (F) at (20,0) {};
\draw (C)--(D);
\draw (E) -- (F);
\draw[dotted] (D)--(E);
\draw (10.5,1) node {$3$};
\draw (18.5,1) node {$3$};
  \end{tikzpicture},
\end{center}
respectively, since the Coxeter graph of $\mathcal{S}_{n+1}$ has $n$ vertices.
\end{example}

\par We now wish to extend this idea one step further.
\begin{definition}
 Let $W$ be a general Coxeter group, rather than a finite reflection one. Define a \emph{general Coxeter graph} to be a finite undirected graph, with edges labelled with $z \geq 3$, $z \in \mathbb{Z}$ or with the symbol $\infty$. Let $S$ denote the set of vertices. Let $s, s' \in S$ be two distinct vertices and let $m(s,s')$ denote the label on the edge joining $s$ and $s'$. If $m(s,s')=3$, we shall not write it due to its frequent occurrence. Further, if two vertices $s \neq s'$ are not joined with an edge, then $m(s,s')=2$. Lastly, we set $m(s,s)=1$.
\end{definition}
\par Coxeter graphs can be associated to bilinear forms using symmetric matrices as follows.
\begin{definition}\label{bilinearForm}
A Coxeter graph $\Gamma$ with vertex set $S$ of cardinality $n$ is associated with a symmetric $n \times n$ matrix $A$ by setting
\begin{align}\label{matrixElementGraph}
a(s,s') \coloneqq - \cos \dfrac{\pi}{m(s,s')}.
\end{align}
\end{definition}
\begin{remark}
The formulation in Definition \ref{bilinearForm} allows for the classification of Coxeter graphs in terms of positive/negative (semi-)definiteness. The general idea is that given a basis $\alpha_{1}, \dots, \alpha_{n}$ of simple roots in $R_{+}$, the matrix $M=(m_{ij})$ defined by $(\alpha_{i}, \alpha_{j}) = -2 \cos(\pi/m_{ij})$ which is known as the \emph{Coxeter matrix} defines a bilinear form, which in turn is positive/negative (semi-)definite. In this exposition we only treat positive definite Coxeter graphs. Further details can be found in \cite{Humphreys}. We shall only use Equation \eqref{matrixElementGraph} to obtain restrictions on the values of $m(s,s')$ in the special type of group studied below.
\end{remark}
\par The study of low-dimensional crystal structures heavily relies on symmetry groups \cite{Humphreys}. Hence, the special type of groups we present shortly are referred to as \emph{crystallographic}. Before doing so, we introduce the notion of a \emph{lattice}, which is the subset of $V$ the latter groups act on.

\begin{definition}\label{lattice}
Let $V = \mathbb{R}^{n}$ with orthonormal basis $e_{1}, \dots, e_{n}$. A \emph{lattice} $\Lambda$ is a subset of $V$ given by
\begin{equation*}
\Lambda = \left\{ \sum_{j=1}^{n} z_{j} e_{j} | z_{j} \in \mathbb{Z} \right\}.
\end{equation*}
\end{definition}
A \emph{crystallographic group} is one which stabilizes a lattice $\Lambda$ in $V$, formally defined as follows.
\begin{definition}\label{crystallographicGroup}
A subgroup $G$ of the general linear group $GL(V)$ is called \emph{crystallographic} if $g \Lambda \subset \Lambda$, $\forall g \in G$.
\end{definition}
\par The essential difference between a normalized root system and a crystallographic one is that the roots need not have norm $2$, but instead we demand that for any $\alpha, \beta \in R$, $\langle \beta, \alpha^{\vee} \rangle \in \mathbb{Z}$. Let $\theta$ be the angle between any two roots $\alpha$ and $\beta$ in $R$. Then we have
\begin{align*}
\langle \beta, \alpha^{\vee} \rangle  \langle \alpha, \beta^{\vee} \rangle = \dfrac{2 \langle \alpha, \beta \rangle^{2} }{\langle \alpha, \alpha \rangle \langle \beta, \beta \rangle} = 4 \cos^{2} \theta \in \mathbb{Z},
\end{align*}
which means that $\theta = \left( (\pi /2), (\pi /3 \text{ or } 2 \pi /3 ), (\pi /4 \text{ or } 3 \pi /4), (\pi /6 \text{ or } 5 \pi /6), (0 \text{ or } \pi) \right)$, which correspond to $ 4 \cos \theta = 0,1,2,3,4$ respectively. These values of $\theta$ are known as the \emph{crystallographic angles}. There are just $4$ root systems in rank two. They are displayed below in Figure \ref{fig:rank2rootsystems}.
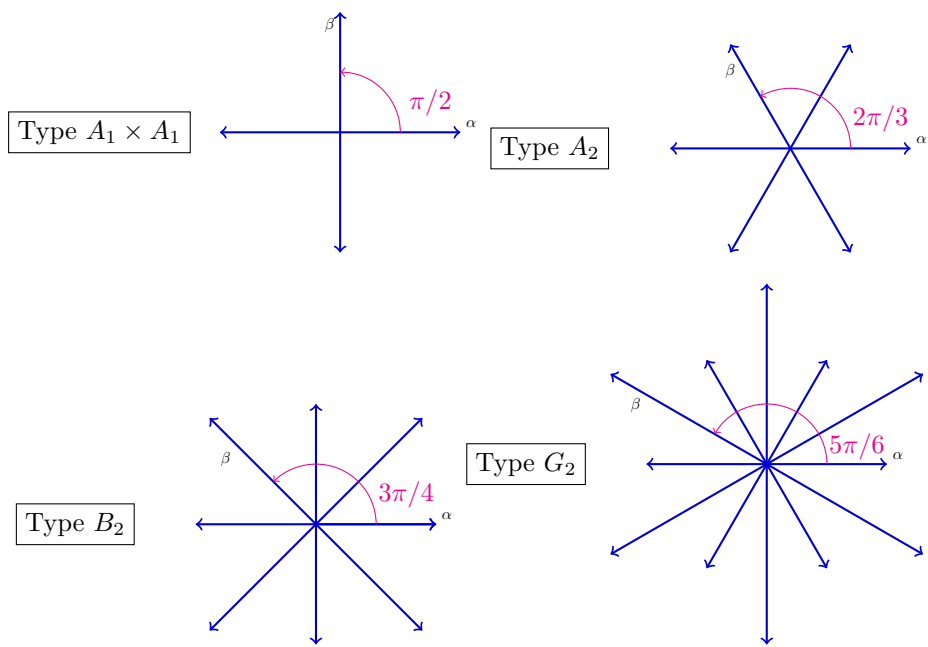
\begin{figure}[H]
\begin{center}
\begin{tikzpicture}[scale=0.8]
    \foreach\ang in {90,180,...,360}{
     \draw[->,blue!80!black,thick] (0,0) -- (\ang:2cm);
    }
    \draw[magenta,->](1,0) arc(0:90:1cm)node at (1.5,0.5){$\pi/2$};
    \node[anchor=south west,scale=0.6] at (2,0) {$\alpha$};
\node[anchor= north east, scale = 0.6] at (0,2) {$\beta$};
	\node[rectangle,draw] at (-4,0) {Type $A_{1} \times A_{1}$};
\end{tikzpicture}
\begin{tikzpicture}[scale=0.8]
    \foreach\ang in {60,120,...,360}{
     \draw[->,blue!80!black,thick] (0,0) -- (\ang:2cm);
    }
    \draw[magenta,->](1,0) arc(0:120:1cm)node at (1.5,0.5){$2\pi/3$};
    \node[anchor=south west,scale=0.6] at (2,0) {$\alpha$};
\node[anchor=north, scale=0.6] at (-1,1.5) {$\beta$};
	\node[rectangle,draw] at (-4,0) {Type $A_2$};
\end{tikzpicture}
\end{center}
\begin{center}
\begin{tikzpicture}[scale=0.8]
    \foreach\ang in {0,90,...,360}{
     \draw[->,blue!80!black,thick] (0,0) -- (\ang:2cm);
    }
  \foreach\ang in {45,135,...,360}{
     \draw[->,blue!80!black,thick] (0,0) -- (\ang:2.5cm);
    }
    \draw[magenta,->](1,0) arc(0:135:1cm)node at (1.5,0.5) {$3\pi/4$};
    \node[anchor=south west,scale=0.6] at (2,0) {$\alpha$};
\node[anchor=north, scale=0.6] at (-1.5,1.3) {$\beta$};
	\node[rectangle,draw] at (-4,0) {Type $ B_{2}$};
\end{tikzpicture}
  \begin{tikzpicture}[scale=0.8]
    \foreach\ang in {60,120,...,360}{
     \draw[->,blue!80!black,thick] (0,0) -- (\ang:2cm);
    }
    \foreach\ang in {30,90,...,330}{
     \draw[->,blue!80!black,thick] (0,0) -- (\ang:3cm);
    }
    \draw[magenta,->](1,0) arc(0:150:1cm)node at (1.5,0.3) {$5\pi/6$};
    \node[anchor=south west,scale=0.6] at (2,0) {$\alpha$};
 \node[anchor= north east, scale = 0.6] at (-2,1.2) {$\beta$};
	\node[rectangle,draw] at (-4,0) {Type $ G_{2}$};
  \end{tikzpicture}
\end{center}
\caption{The four different root systems of rank 2.}\label{fig:rank2rootsystems}
\end{figure}

As we stated earlier, reflection groups associated with crystallographic root systems have a special name, stated in the definition below.
\begin{definition}
Let $R$ be a crystallographic root system. Then the group $W$ generated by all reflections $\sigma_{\alpha}$, $\alpha \in R$ is called the \emph{Weyl group} of $R$.
\end{definition}
\begin{remark}
The extra crystallographic condition we impose on the root system $R$ ensures that all reflections $\sigma_{\alpha} \beta$ are the result of adding an integer multiple of $\alpha$ to $\beta$. Hence, all roots are $\mathbb{Z}$ linear combinations of $\Delta$ and the $\mathbb{Z}$-span of $\Delta$ in $V$ is a lattice, stable under the action of $W$. Hence, the Weyl group $W$ is crystallographic in the sense of Definition \ref{crystallographicGroup}. 
 \begin{proposition}\label{propCryst}
If $G$ is a crystallographic group, then the possible values for $m(\alpha, \beta)$ are $2$, $3$, $4$, or $6$ for $\alpha \neq \beta$ in $\Delta$.
\end{proposition}
\begin{proof}
We know that for $\alpha \neq \beta$, $\sigma_{\alpha} \sigma_{\beta} \neq 1$ acts on the plane spanned by $\alpha$ and $\beta$ as a rotation by $\phi = 2 \pi / m(\alpha, \beta)$. Hence, given a basis for $V$ and $n = \dim V$, the trace of the rotation relative to the basis if given by $(n-2)+2 \cos \phi$. For $\phi \in (0, \pi]$, we have $\cos \phi = \{ -1, -1/2, 0, 1/2 \}$. Using Equation \eqref{matrixElementGraph}, we get that $m(\alpha, \beta) = 2, 3, 4,6$ as required.
\end{proof}
\par We can now go back to the notion of a root system and define the corresponding crystallographic type.

\begin{definition}
A root system $R$ is \emph{crystallographic} if $ \forall \alpha, \beta \in R$, 
\begin{align*}
c = \langle \beta, \alpha^{\vee} \rangle \in \mathbb{Z}.
\end{align*}
\end{definition}

Further, by Proposition \ref{propCryst} we conclude that the Weyl groups are precisely the reflection groups with $m(\alpha, \beta) = \{2,3,4,6\}$, for $\alpha \neq \beta$. Thus, Weyl groups are essentially crystallographic reflection groups. However, it is important to note that two distinct crystallographic root systems might give rise to the same Weyl group.
\end{remark}
\par When the crystallographic root system $R$ (or alternatively the Weyl group $W$) is irreducible in the sense of Definition \ref{coxeterGraph}, the roots can only be of two possible lengths. If there are both `long' and `short' roots, the ratio of the squared lengths can be either $2$ or $3$. In other words, Coxeter graphs do not show the difference between $R$ and $R^{\vee}$. In order to encode this information in a graph, we no longer use Coxeter graphs, but a different type of diagrams.
\begin{definition}
Let $R$ be an irreducible crystallographic root system with both `long' and `short' roots with Coxeter graph $\Gamma$. Then the \emph{Dynkin diagram} of the root system $R$ is the partially directed graph, obtained by directing an arrow towards the short root when adjacent vertices represent a long and a short root. If the label of an edge is $4$ or $6$, then we replace the edge by a double or a triple edge respectively.
\end{definition}
\par We are now ready to survey the construction of crystallographic root systems of all possible types. Our methodology follows the one outlined in \cite{Bourbaki}. We choose a suitable lattice $\Lambda$ in $\mathbb{R}^{n}$ and define $R$ to be the set of all vectors with one or two prescribed lengths. Hence, the reflections with respect to the vectors in $R$ stabilize the lattice $\Lambda$ and in turn permute $R$ as required. Let $e_{1}, \dots, e_{n}$ be the standard basis of $\mathbb{R}^{n}$ throughout.
\begin{example} ($A_{n-1}, n\geq 2$).
Let $V$ be the hyperplane in $\mathbb{R}^{n}$, consisting of vectors whose coordinates add up to $0$. Let $\Lambda $ be as in Definition \ref{lattice}. Define $R$ as the set of all $\alpha \in V \cap \Lambda $ such that $ | \alpha |^{2} = 2$. Then $R$ is given by
\begin{align*}
R = \{ \pm ( e_{i} - e_{j} ),  1 \leq i \leq j \leq n \}.
\end{align*}
Root systems of such a type are referred to as type $A_{n-1}$ for $n \geq 2$. The simple system $\Delta$ is given by
\begin{align*}
\Delta = \{ (e_{1}-e_{2}), (e_{2}-e_{3}), \dots,( e_{n-1}-e_{n-2})\}.
\end{align*}
The Weyl group $W$ is then $\mathcal{S}_{n}$, which permutes the roots $e_{i}$ and $e_{j}$. The Dynkin diagram of $A_{n-1}$ is 
\begin{center}
  \begin{tikzpicture}[scale=.4]
    \draw (-1,0) node[anchor=east]  {$A_{n-1}$};
\node[circle, draw, minimum size=.2cm] (C) at (1,0) {};
\node[circle, draw, minimum size=.2cm] (D) at (4,0) {};
\node[circle, draw, minimum size=.2cm] (E) at (8,0) {};
\node[circle, draw, minimum size=.2cm] (F) at (11,0) {};
\node[circle, draw, minimum size=.2cm] (G) at (14,0) {};
\draw (C)--(D);
\draw (E) -- (F) -- (G);
\draw[dotted] (D)--(E);
\draw (16,0) node[anchor=west] {$n \geq 2$.};
  \end{tikzpicture}
\end{center}
Its Coxeter graph is the same, namely
\begin{center}
  \begin{tikzpicture}[scale=.4]
    \draw (-1,0) node[anchor=east]  {$A_{n-1}$};
\node[circle, draw, fill, minimum size=.2cm] (C) at (1,0) {};
\node[circle, draw,fill, minimum size=.2cm] (D) at (4,0) {};
\node[circle, draw,fill, minimum size=.2cm] (E) at (8,0) {};
\node[circle, draw,fill, minimum size=.2cm] (F) at (11,0) {};
\node[circle, draw,fill, minimum size=.2cm] (G) at (14,0) {};
\draw (C)--(D);
\draw (E) -- (F) -- (G);
\draw[dotted] (D)--(E);
\draw (16,0) node[anchor=west] {$n \geq 2$.};
  \end{tikzpicture}
\end{center}
\end{example}
\begin{example} \label{bnExample} $(B_{n}, n\geq 2)$.
Let $V=\mathbb{R}^{n}$ and define $R$ as the set of all vectors of squared length $1$ or $2$ in the standard lattice $\Lambda$. So $R$ contains the $2n$ short roots $(\pm e_{i})$ and the $2n(n-1)$ long roots $(\pm e_{i} \pm e_{j})$, $i<j$, which gives the total of $2n+2n^{2}-2n=2n^{2}$ roots. That is,
\begin{equation*}
R = \{ \pm e_{i} \pm e_{j}, 1 \leq i < j \leq n \} \cup \{ \pm e_{i} , 1 \leq i \leq n \}.
\end{equation*}
The simple system $\Delta$ is
\begin{equation*}
\Delta = \{ (e_{1} -e_{2}), (e_{2}-e_{3}), \dots, (e_{n-1}-e_{n}), e_{n} \}.
\end{equation*}
The Weyl group $W$ stabilizes $\{ \pm e_{1}, \dots, \pm e_{n} \}$ and is the product of the group of all permutations of $\{ e_{1}, \dots, e_{n}\}$ and the group of all sign changes on $e_{i}$, i.e. $\mathcal{S}_{n} \rtimes (\mathbb{Z} / 2\mathbb{Z})^{n}$. Such a root system is known as type $B_{n}$ for $n \geq 2$ and has Dynkin diagram
\begin{center}
  \begin{tikzpicture}[scale=.4]
    \draw (-1,0) node[anchor=east]  {$B_{n}$};
\node[circle, draw, minimum size=.2cm] (C) at (1,0) {};
\node[circle, draw, minimum size=.2cm] (D) at (4,0) {};
\node[circle, draw, minimum size=.2cm] (E) at (8,0) {};
\node[circle, draw, minimum size=.2cm] (F) at (11,0) {};
\node[circle, draw, minimum size=.2cm] (G) at (14,0) {};
    \draw[thick] (11.5, .2 cm) --  +(2 cm,0);
    \draw[thick] (11.5, -.2 cm) -- +(2 cm,0);
\draw (C)--(D);
\draw (E) -- (F);
\draw[dotted] (D)--(E);
\draw (16,0) node[anchor=west] {$n \geq 2$.};
\draw (12.7,0) --++ (120:.5)
(12.7,0) --++ (-120:.5);
  \end{tikzpicture}
\end{center}
The Coxeter graph of the root system type $B_{n}$ is
\begin{center}
  \begin{tikzpicture}[scale=.4]
    \draw (-1,0) node[anchor=east]  {$B_{n}$};
\node[circle, draw, fill, minimum size=.2cm] (C) at (1,0) {};
\node[circle, draw, fill,minimum size=.2cm] (D) at (4,0) {};
\node[circle, draw, fill,minimum size=.2cm] (E) at (8,0) {};
\node[circle, draw, fill,minimum size=.2cm] (F) at (11,0) {};
\node[circle, draw, fill,minimum size=.2cm] (G) at (14,0) {};
\draw (C)--(D);
\draw (E) -- (F)--(G);
\node at (12.5,0.5) {4};
\draw[dotted] (D)--(E);
\draw (16,0) node[anchor=west] {$n \geq 2$.};
  \end{tikzpicture}
\end{center}
\end{example}
\begin{example} $(C_{n}, n \geq 2)$.
The root system we construct now is known as the inverse of the $B_{n}$ type as in Example \ref{bnExample}. Again, let $V=\mathbb{R}^{n}$ and define $R$ as the set of all vectors of squared length $1$ or $2$ in the standard lattice $\Lambda$. This time, however, $R$ contains the $2n$ long roots $(\pm2 e_{i})$ and the $2n(n-1)$ short roots $(\pm e_{i} \pm e_{j})$, $i<j$, which gives a total of $2n^{2}$ roots as well. That is,
\begin{align*}
R = \{ (\pm e_{i} \pm e_{j}), 1 \leq i <  j \leq n \} \cup \{ \pm 2 e_{i} , 1 \leq i \leq n \}.
\end{align*}
For $\Delta$ take
\begin{equation*}
\Delta = \{ (e_{1} -e_{2}), (e_{2}-e_{3}), \dots, (e_{n-1}-e_{n}), 2e_{n} \}.
\end{equation*}
This is known as the $C_{n}$ type root system for $n \geq 2$. The Dynkin diagram is
\begin{center}
  \begin{tikzpicture}[scale=.4]
    \draw (-1,0) node[anchor=east]  {$C_{n}$};
\node[circle, draw, minimum size=.2cm] (C) at (1,0) {};
\node[circle, draw, minimum size=.2cm] (D) at (4,0) {};
\node[circle, draw, minimum size=.2cm] (E) at (8,0) {};
\node[circle, draw, minimum size=.2cm] (F) at (11,0) {};
\node[circle, draw, minimum size=.2cm] (G) at (14,0) {};
    \draw[thick] (11.5, .2 cm) --  +(2 cm,0);
    \draw[thick] (11.5, -.2 cm) -- +(2 cm,0);
\draw (C)--(D);
\draw (E) -- (F);
\draw[dotted] (D)--(E);
\draw (16,0) node[anchor=west] {$n \geq 2$.};
\draw (12.7,0) --++ (60:.5)
(12.7,0) --++ (-60:.5);
  \end{tikzpicture}
\end{center}
The Coxeter graph of the root system type $C_{n}$ is the same as the Coxeter graph of $B_{n}$. This is a good example of how the Dynkin diagrams convey more information.
\end{example}
\begin{example} $(D_{n}, n \geq 4)$.
Let $V = \mathbb{R}^{n}$ and let $R$ be the set of all vectors of squared length equal to $2$ in the standard lattice $\Lambda$. Thus, $R$ has $2n(n-1)$ roots and is of the form
\begin{equation*}
R = \{ (\pm e_{i} \pm e_{j}), 1 \leq i < j \leq n \}.
\end{equation*}
For the simple system $\Delta$ we have
\begin{equation*}
\Delta = \{ (e_{1} -e_{2}), (e_{2}-e_{3}), \dots, (e_{n-1}-e_{n}), (e_{n-1}+e_{n}) \}.
\end{equation*}
The Weyl group $W$ stabilizes $\{ \pm e_{1}, \dots, \pm e_{n} \}$ and is the semidirect product of the group of all permutations of $\{ e_{1}, \dots, e_{n} \}$, that is, $\mathcal{S}_{n}$, and the group of sign changes, with evenly many signs equal to $-1$, i.e. $(\mathbb{Z}/2 \mathbb{Z})^{n-1}$. The Dynkin diagram of this root system, called the $D_{n}$ type for $n \geq 4$, is
\begin{center}
  \begin{tikzpicture}[scale=.4]
    \draw (-1,0) node[anchor=east]  {$D_{n}$};
\node[circle, draw, minimum size=.2cm] (C) at (1,0) {};
\node[circle, draw, minimum size=.2cm] (D) at (4,0) {};
\node[circle, draw, minimum size=.2cm] (E) at (8,0) {};
\node[circle, draw, minimum size=.2cm] (F) at (11,0) {};
\node[circle, draw, minimum size=.2cm] (G) at (14,1) {};
\node[circle, draw, minimum size=.2cm] (H) at (14,-1) {};
\draw (C)--(D);
\draw (E) -- (F) -- (G);
\draw (F) --(H);
\draw[dotted] (D)--(E);
\draw (16,0) node[anchor=west] {$n \geq 4$.};
  \end{tikzpicture}
\end{center}
The Coxeter graph of $D_{n}$ is 
\begin{center}
  \begin{tikzpicture}[scale=.4]
    \draw (-1,0) node[anchor=east]  {$D_{n}$};
\node[circle, draw, fill, minimum size=.2cm] (C) at (1,0) {};
\node[circle, draw,fill, minimum size=.2cm] (D) at (4,0) {};
\node[circle, draw,fill, minimum size=.2cm] (E) at (8,0) {};
\node[circle, draw,fill, minimum size=.2cm] (F) at (11,0) {};
\node[circle, draw,fill, minimum size=.2cm] (G) at (14,0) {};
\node[circle, draw, fill, minimum size=.2cm] (H) at (17,0) {};
\node[circle, draw, fill, minimum size=.2cm] (I) at (14,3) {};
\draw (C)--(D);
\draw (E) -- (F) -- (G) -- (H);
\draw (G)--(I);
\draw[dotted] (D)--(E);
\draw (19,0) node[anchor=west] {$n \geq 2$.};
  \end{tikzpicture}
\end{center}

\end{example} 
\begin{example}$(G_{2})$.
Let $V$ be the hyperplane in $\mathbb{R}^{3}$ consisting of all vectors with coordinates adding up to $0$. Take $R$ to be the set if vectors of squared length either $2$ or $6$ in $V \cap \Lambda$. Thus, $R$ is defined as
\begin{equation*}
R = \{ \pm e_{1}, \pm e_{2}, \pm (e_{1}+e_{2} ), \pm (e_{1}+2e_{2}), \pm (e_{1}+3e_{2}), \pm(2e_{1}+3e_{2}) \}.
\end{equation*}
Then the simple system is $\Delta = \{e_{1}, e_{2} \}$. This type of system is known as $G_{2}$ and has Dynkin diagram
\begin{center}
  \begin{tikzpicture}[scale=.4]
    \draw (-1,0) node[anchor=east]  {$G_{2}$};
\node[circle, draw, minimum size=.2cm] (F) at (1,0) {};
\node[circle, draw, minimum size=.2cm] (G) at (4,0) {};
    \draw[thick] (1.5, .2 cm) --  +(2 cm,0);
    \draw[thick] (1.5, -.2 cm) -- +(2 cm,0);
   \draw[thick] (1.5, -.0cm) -- + (2cm, 0);
\draw (2.7,0) --++ (120:.5)
(2.7,0) --++ (-120:.5);
  \end{tikzpicture}.
\end{center}
The Coxeter graph of $G_{2}$ is
\begin{center}
  \begin{tikzpicture}[scale=.4]
    \draw (-1,0) node[anchor=east]  {$G_{2}$};
\node[circle, draw, fill, minimum size=.2cm] (F) at (1,0) {};
\node[circle, draw, fill, minimum size=.2cm] (G) at (4,0) {};
\node at (2.5,0.5) {6};
\draw (F)--(G);
  \end{tikzpicture}.
\end{center}
\end{example}
\begin{example}$(F_{4})$.
Let $V = \mathbb{R}^{4}$. Let $\Lambda$ be the standard lattice. Define a new lattice $\tilde{\Lambda} \coloneqq \Lambda + \mathbb{Z} \dfrac{1}{2} \left(e_{1}+e_{2}+e_{3}+e_{4} \right)$. Hence, define $R$ to be the set of all vectors in the new lattice $\tilde{\Lambda}$ of squared length $1$ or $2$. Thus, $R$ contains $24$ long and $24$ short roots, that is,
\begin{equation*}
R = \{ (\pm e_{i} \pm e_{j}), 1 \leq i < j \leq 4\} \cup \{\pm e_{i}, 1\leq i \leq 4 \} \cup \left\{ \dfrac{1}{2} ( \pm e_{1} \pm e_{2} \pm e_{3} \pm e_{4} ) \right\}.
\end{equation*} 
For $\Delta$ take
\begin{equation*}
\Delta = \left\{ (e_{2} - e_{3}), (e_{3}-e_{4}), e_{4}, \dfrac{1}{2} \left( e_{1}- e_{2}-e_{3}-e_{4} \right) \right\}.
\end{equation*}
This type of root system is known as $F_{4}$ and is represented via the Dynkin diagram
\begin{center}
  \begin{tikzpicture}[scale=.4]
    \draw (-1,0) node[anchor=east]  {$F_{4}$};
\node[circle, draw, minimum size=.2cm] (A) at (1,0) {};
\node[circle, draw, minimum size=.2cm] (F) at (4,0) {};
\node[circle, draw, minimum size=.2cm] (G) at (7,0) {};
\node[circle, draw, minimum size=.2cm] (B) at (10,0) {};
    \draw[thick] (4.5, .2 cm) --  +(2 cm,0);
    \draw[thick] (4.5, -.2 cm) -- +(2 cm,0);
\draw (5.7,0) --++ (120:.5)
(5.7,0) --++ (-120:.5);
\draw (A) -- (F);
\draw (G) -- (B);
  \end{tikzpicture}.
\end{center}
Its Coxeter graph is given by
\begin{center}
  \begin{tikzpicture}[scale=.4]
    \draw (-1,0) node[anchor=east]  {$F_{4}$};
\node[circle, draw,fill, minimum size=.2cm] (A) at (1,0) {};
\node[circle, draw,fill, minimum size=.2cm] (F) at (4,0) {};
\node[circle, draw,fill, minimum size=.2cm] (G) at (7,0) {};
\node[circle, draw, fill, minimum size=.2cm] (B) at (10,0) {};
\draw (A) -- (F)--(G) -- (B);
\node at (5.5,0.5) {4};
  \end{tikzpicture}.
\end{center}
\end{example}
\begin{example}\label{e8root} $(E_{8})$.
Let $V = \mathbb{R}^{8}$. Let $\Lambda'' = \left\{ \displaystyle{\sum_{j=1}^{n}} 2z_{j} e_{j} \lvert z_{j} \in \mathbb{Z} \right\}$. Then let $\Lambda' = \Lambda'' + \mathbb{Z} \left( \dfrac{1}{2} \displaystyle{ \sum_{i=1}^{8}} e_{i} \right) $. Define $R$ to be the set of all vectors of squared length $2$ in $\Lambda'$. Thus, $R$ is
\begin{align*}
R = \{ (\pm e_{i} \pm e_{j}), 1 \leq i < j \leq 8 \} \cup \left\{ \dfrac{1}{2} \sum^{8}_{i=1} \pm e_{i} \right\}, \quad \text{ with an even number of $+$ signs} .
\end{align*}
For $\Delta$ take
\begin{align*}
\Delta = \left\{ \alpha_{1}, \alpha_{2}, \alpha_{i} \right\}, \quad \text{ where } \alpha_{1} &= \dfrac{1}{2} \left( e_{1}-e_{2}-e_{3}-e_{4}-e_{5}-e_{6}-e_{7} + e_{8} \right), \\ \alpha_{2} &= (e_{1}+e_{2} ), \\ 
\alpha_{i} &= (e_{i-1}-2_{i-2}) \text{ for } 3 \leq i \leq 8. 
\end{align*}
This root system is known as $E_{8}$ and has $240$ roots. Its Dynkin diagram is
\begin{center}
  \begin{tikzpicture}[scale=.4]
    \draw (-1,0) node[anchor=east]  {$E_{8}$};
\node[circle, draw, minimum size=.2cm] (A) at (1,0) {};
\node[circle, draw, minimum size=.2cm] (B) at (4,0) {};
\node[circle, draw, minimum size=.2cm] (C) at (7,0) {};
\node[circle, draw, minimum size=.2cm] (D) at (10,0) {};
\node[circle, draw, minimum size=.2cm] (E) at (13,0) {};
\node[circle, draw, minimum size=.2cm] (F) at (16,0) {};
\node[circle, draw, minimum size=.2cm] (G) at (19,0) {};
\draw (A) --(B)--(C)--(D)--(E)--(F)--(G);
\node[circle, draw, minimum size=.2cm] (H) at (7,3) {};
\draw (C)--(H);
  \end{tikzpicture}.
\end{center}
The Coxeter graph is exactly the same.
\end{example}
\begin{example}
Consider the root system of type $E_{8}$ in Example \ref{e8root}. Let $V$ be the span of the roots $\alpha_{i}$ for $1 \leq i \leq 7$ in $\mathbb{R}^{8}$. Let $R$ be the set of the roots of $E_{8}$ lying also in $V$, that is,
\begin{align*}
R = \left\{ (\pm e_{i} \pm e_{j}), \pm (e_{7}-e_{8}), \pm \dfrac{1}{2} \left(e_{7}-e_{8}+\sum_{l=1}^{6} \pm e_{l} \right) \right\}, \quad 1 \leq i < j \leq 6,
\end{align*}
and there are an odd number of minus signs in the last summation term. Then a simple system $\Delta$ is formed by the collection of $\alpha_{i}$, for $1 \leq i \leq 7$. This root system is referred to as the $E_{7}$ type with cardinality $126$ roots and has an identical Dynkin diagram to the one of $E_{8}$, just lacking the last vertex, i.e.
\begin{center}
  \begin{tikzpicture}[scale=.4]
    \draw (-1,0) node[anchor=east]  {$E_{7}$};
\node[circle, draw, minimum size=.2cm] (A) at (1,0) {};
\node[circle, draw, minimum size=.2cm] (B) at (4,0) {};
\node[circle, draw, minimum size=.2cm] (C) at (7,0) {};
\node[circle, draw, minimum size=.2cm] (D) at (10,0) {};
\node[circle, draw, minimum size=.2cm] (E) at (13,0) {};
\node[circle, draw, minimum size=.2cm] (F) at (16,0) {};
\draw (A) --(B)--(C)--(D)--(E)--(F);
\node[circle, draw, minimum size=.2cm] (H) at (7,3) {};
\draw (C)--(H);
  \end{tikzpicture}.
\end{center}
The Coxeter graph of $E_{7}$ is the same as the Dynkin diagram.
\end{example}
\begin{example}$(E_{6})$.
Lastly, consider once again the root system of type $E_{8}$ as in Example \ref{e8root}. Let $V$ be the span of the $\alpha_{i}$ for $1 \leq i \leq 6$. Take $R$ to be 
\begin{align*}
R = \left\{ (\pm e_{i} \pm e_{j}), \pm \dfrac{1}{2} \left( e_{8}-e_{7}-e_{6}+ \sum_{l=1}^{5} \pm e_{l} \right) \right\}, \quad 1 \leq i < j \leq 6,
\end{align*}
and there are an odd number of minus signs in the last sum. A simple system is $\Delta = \{ \alpha_{i}, 1 \leq i \leq 6 \}$. This root system is known as type $E_{6}$ and has Dynkin diagram
\begin{center}
  \begin{tikzpicture}[scale=.4]
    \draw (-1,0) node[anchor=east]  {$E_{6}$};
\node[circle, draw, minimum size=.2cm] (A) at (1,0) {};
\node[circle, draw, minimum size=.2cm] (B) at (4,0) {};
\node[circle, draw, minimum size=.2cm] (C) at (7,0) {};
\node[circle, draw, minimum size=.2cm] (D) at (10,0) {};
\node[circle, draw, minimum size=.2cm] (E) at (13,0) {};
\draw (A) --(B)--(C)--(D)--(E);
\node[circle, draw, minimum size=.2cm] (H) at (7,3) {};
\draw (C)--(H);
  \end{tikzpicture}.
\end{center}
The Dynkin diagram is identical to the Coxeter graph of $E_{8}$.
\end{example}
\par Having surveyed the different types of root systems and the groups they give rise to, we now proceed by examining the action of a general finite subgroup $W$ of $GL(V)$ on the ring of polynomial functions in the space $V$.
\section{Polynomial Invariant Theory.}\label{sec:poly}
\indent \par In this section we discuss the structure of the subalgebra of polynomial functions, which remains invariant under the action of a subgroup of the general linear group. We then present a key result, known as \emph{Chevalley's theorem}, which we need to prove the integrability of the rational quantum Calogero-Moser system.
\par
Let $V$ be an $n$-dimensional vector space over a field $K$ of characteristic $0$. Let $\xi_{1}, \dots, \xi_{n}$ be a basis for $V$. For $x_{1}, \dots, x_{n} \in K$, if we write $\xi = x_{1} \xi_{1} + \dots + x_{n} \xi_{n}$ for $\xi \in V$, then $x_{j} (\xi) \in K$ are the linear coordinate functions relative to the given basis of $V$. Let $A(V)$ denote the commutative algebra of all functions on $V$ with values in $K$.
\begin{definition}
The \emph{algebra of polynomial functions on $V$}, denoted by $P(V)$, is the smallest subalgebra of $A(V)$ containing the linear coordinate functions $\xi \mapsto x_{j} (\xi)$. 
\end{definition} 
\begin{remark}
Note that $P(V)$ is independent of the choice of basis of $V$ and any choice of basis in $V$ yields an isomorphism $P(V) \cong K[x_{1}, \dots, x_{n}]$.
\end{remark}
\par Now a basis of $P(V)$ is given by the monomials $x_{1}^{d_{1}} \cdots x_{n}^{d_{n}}$, where $d_{j} \in \mathbb{N}_{0}$. Consider the subspace $P_{d}(V) \subset P(V)$ given by $P_{d}(V) = \{ p \in P(V) \lvert p(t \xi) = t^{d} p(\xi), \forall t \in K, \xi \in V$. The elements of this subspace are called \emph{homogeneous polynomials of degree $d$}, where $d \in \mathbb{N}_{0}$. 
\par The general linear group $GL(V)$ acts on the algebra of polynomial functions $P(V)$ naturally, i.e. for $g \in GL(V)$, $p \in P(V)$ and $\xi \in V$, we have
\begin{equation*}
( g \cdot p ) (\xi) = p(g^{-1} \xi).
\end{equation*}
Also, for $g, h \in GL(V)$ and $p \in P(V)$ we have $g \cdot (h \cdot p ) = (gh ) \cdot p$. Even further, $g \cdot P_{d}(V) \subset P_{d} (V)$ for any $g \in GL(V)$. Next we formalize the idea of invariance under the action of an arbitrary subgroup $G$ of $GL(V)$.
\begin{definition}
Let $G $ be a subgroup of $GL(V)$ and let $p \in P(V)$. A polynomial $p$ is said to be \emph{$G$-invariant} if $g \cdot p = p$, $\forall g \in G$. We denote the set of all $G$-invariant polynomial functions in $P(V)$ by $P(V)^{G}$.
\end{definition}
\par We illustrate all concept introduced so far via the following example.

\begin{example}\label{symmetricPoly}
Let $p \in K[x_{1}, \dots, x_{n}]$ be given by $p = ( c_{d_{1}} x_{1}^{d_{1}} + \dots + c_{d_{n}} x_{n}^{d_{n}})$, where $c_{d_{j}}$ are constants, and let $\sigma \in \mathcal{S}_{n}$. Then $\sigma$ acts on $p$ by permuting the indices of the elements in $p$, that is
\begin{align*}
\sigma \cdot p =  c_{d_{1}} x_{\sigma (1)}^{d_{1}} + \dots + c_{d_{n}} x_{\sigma(n)}^{d_{n}} =  c_{j_{\sigma(1)}} x_{1}^{j_{1}} + \dots + c_{j _{\sigma(n)}} x_{n}^{j_{n}}.
\end{align*}
Polynomials which remain invariant under the action of the symmetric group $\mathcal{S}_{n}$ are called \emph{symmetric polynomials}. The elementary symmetric polynomials are given by
\begin{align*}
p_{d} (x_{1}, \dots, x_{n} )  = \sum_{1 \leq j_{1} < \dots < j_{k} \leq n} x_{j_{1}} \cdots x_{j_{n}}, \quad k=1, \dots, n.
\end{align*}
By a classical result in algebra, each symmetric polynomial is uniquely determined in the elementary symmetric polynomials, thus, $K[x_{1}, \dots, x_{n}]^{\mathcal{S}_{n}} \cong K [p_{1}, \dots, p_{n}]$.
\end{example}
\par Next we state an important theorem due to Hilbert. An extensive proof can be found in \cite{heckmanCoxeter}, Section 3, Theorem 3.2.
\begin{theorem}
Let $G$ be a finite subgroup of $GL(V)$ and let $P^{G}(V)$ be the invariant subalgebra of $P(V)$. Then $P^{G}(V)$ is finitely generated.
\end{theorem}
\par The next result is known as Chevalley's theorem and is the main tool we shall employ to prove the integrability of the quantum rational Calogero-Moser system. The proof of the statement is omitted, but extensive ones can be found in both \cite{heckmanCoxeter} see Section 3.2, Theorem 3.8 and \cite{Humphreys} see Chapter 3, Section 3.5.
	\begin{theorem}\label{chevalley}
Let $W$ be a finite reflection group, $W \subset O(V)$. Then $P^{W}(V) \cong \mathbb{C}[p_{1}, \dots, p_{n}]$, where $p_{j}$ are homogeneous polynomial functions of degree $d_{j}$, algebraically independent and invariant under the action of $W$.
	\end{theorem}
\par Chevalley's theorem suggests that the idea discussed in Example \ref{symmetricPoly} can be generalized to the context of arbitrary Coxeter groups. We are now ready to introduce the main tool, used to obtain the integrals of motion of the quantum rational Calogero-Moser system, namely the \emph{Dunkl operator}.

\section{Quantum Integrability via Dunkl Operators.}\label{sec:Dunkl}
Dunkl operators are originally conceived in the rational differential setting by Charles Dunkl in the 1980s \cite{Dunkl}. A whole framework is subsequently developed in the prominent works \cite{Opdam}, \cite{Cherednik1}, \cite{Heckman}. Here we restrict ourselves to the simplest case, the rational Dunkl operator.
\subsection{The Dunkl Operator, Commutativity.}
Let $\langle \cdot, \cdot \rangle$ denote the standard inner product in the $n$-dimensional Euclidean vector space $V$ with standard basis $e_{1}, \dots, e_{n}$. Let $R$ be a normalized Coxeter root system with roots $\alpha \in R$. Denote by $\sigma_{\alpha}$ the orthogonal reflection corresponding to the roots $\alpha \in R$ and by $H^{\alpha}=\langle \alpha, x \rangle = 0$ with $x = (x_{1}, \dots, x_{n})$ the hyperplane orthogonal to $\alpha$ as in Definition \ref{definitionOrthogonalReflection}. Next, define $V_{reg} = V - \cup_{\alpha} H^{\alpha}$. For $W$ a finite reflection group, let $c: R \to \mathbb{C}$ be a $W$-invariant function, known as a \emph{multiplicity function}. For $\alpha \in R$ we denote $c(\alpha) = c_{\alpha}$. Let $\alpha^{\vee} \in V$ be the coroot of $\alpha$. Let $R_{+}$ be the subset of $R$ containing all positive roots with respect to a given partial ordering. Let $a = (a_{1}, \dots, a_{n})$ and let $\partial_{a}$ be the directional derivative along the vector $a \in V$. Let $\mathbb{C}[V]$ denote the algebra of $\mathbb{C}$-valued polynomial functions on $V$, let $\mathbb{C}W$ be the group algebra, identified with the vector space of formal linear combination of elements of $W$, and finally, let $D(V)$ be the algebra of differential operators on $V$.
\begin{definition}\label{DunklOp}
The \emph{rational Dunkl operator} is defined by
\begin{align*}
D_{a} = \partial_{a} - \sum_{\alpha \in R_{+}} \dfrac{c_{\alpha} \langle \alpha, a \rangle}{\langle \alpha, \cdot \rangle} (1-\sigma_{\alpha}).
\end{align*}
The Dunkl operator is an element of the semi-direct product of the group algebra $\mathbb{C}W$ and the algebra of differential operators on $V_{reg}$, i.e.  $D_{a} \in \mathbb{C}W \ltimes D(V_{reg})$.
\end{definition}
The following property of the Dunkl operators is of great importance for our further results.
\begin{proposition}
Let $w \in W$. Then $w D_{a} = D_{w a } w$.
\end{proposition}

\begin{proof}
Consider the Dunkl operator as in Definition \ref{DunklOp}. Let $w \in W$. Then using the obvious relation $w \partial_{a} = \partial_{w(a)} w$ and the identity $ w \sigma_{\alpha} = \sigma_{w \alpha} w$ from Proposition \ref{propgroupaction} we have
\begin{align*}
w D_{a} &= w \partial_{a} - w \sum_{\alpha \in R_{+}} \dfrac{c_{\alpha} \langle \alpha, a \rangle}{\langle \alpha, \cdot \rangle} (1-\sigma_{\alpha}) \\
&= \partial_{w(a)} w - \sum_{\alpha \in R_{+}} \dfrac{c_{\alpha} \langle \alpha, a \rangle}{\langle \alpha, \cdot \rangle} (w-\sigma_{w \alpha} w ) \\
&= \left( \partial_{w(a)} - \sum_{\alpha \in R_{+}} \dfrac{c_{\alpha} \langle \alpha, a \rangle}{\langle \alpha, \cdot \rangle} (1-\sigma_{w \alpha}) \right) w \\
&= D_{w(a)} w.
\end{align*}
\end{proof}
We can now discuss the main property of the Dunkl operators, namely, their commutativity.
\begin{theorem}
The Dunkl operators commute for any $a, b \in \mathbb{C}^{n}$, that is, $[D_{a}, D_{b} ] =0$.
\end{theorem}
\begin{proof}
There are several ways to approach the proof of this statement. A rigorous one was used initially by Dunkl \cite{Dunkl}. However, the calculations are quite tedious. Here we shall employ a clever trick, used by Etingof \cite{Etingof}. Here is a sketch of the proof. Instead of showing the commutativity of the operators explicitly, we consider the commutator of the Dunkl operator $D_{a}$ with the function $\langle x, \xi \rangle : \mathbb{C} [x] \to \mathbb{C}[x]$. Then we consider the commutator of another Dunkl operator $D_{b}$, $b \in \mathbb{C}^{n}$, $a \neq b$, with $[D_{a}, \langle x, \xi \rangle ]$. Employing the Jacobi identity, we split the commutator in two symmetric parts, which in turn cancel. This leads to the conclusion that the Dunkl operators indeed commute.
\par  Let $\xi = (\xi_{1}, \dots, \xi_{n}) \in \mathbb{C}^{n}$. Consider first the following commutator.
\begin{align*}
[D_{a}, \langle x ,\xi \rangle ] = \left[ \partial_{a}  - \sum_{\alpha \in R_{+}} \dfrac{c_{\alpha} \langle \alpha, a \rangle}{\langle \alpha, \cdot \rangle} (1-\sigma_{\alpha}),  \langle x ,\xi \rangle  \right].
\end{align*}
Writing out the resulting expression explicitly, we get
\begin{align*}
\left[ D_{a},  \langle x ,\xi \rangle  \right] = [ \partial_{a},  \langle x ,\xi \rangle  ] - \left[  \sum_{\alpha \in R_{+}} \dfrac{c_{\alpha} \langle \alpha, a \rangle}{\langle \alpha, \cdot \rangle}   ,  \langle \cdot ,\xi \rangle  \right] + \left[  \sum_{\alpha \in R_{+}} \dfrac{c_{\alpha} \langle \alpha, a \rangle}{\langle \alpha, \cdot \rangle}\sigma_{\alpha}  , \langle x ,\xi \rangle  \right].
\end{align*}
We only need to worry about the first and last commutators in the above expression, since the second one yields
\begin{align*}
- \left[  \sum_{\alpha \in R_{+}} \dfrac{c_{\alpha} \langle \alpha, a \rangle}{\langle \alpha, \cdot \rangle}   ,  \langle x ,\xi \rangle  \right]  = -  \cancel{\sum_{\alpha \in R_{+}} \dfrac{c_{\alpha} \langle \alpha, a \rangle}{\langle \alpha, \cdot \rangle}    \langle x ,\xi \rangle}   +   \cancel{\langle x ,\xi \rangle  \sum_{\alpha \in R_{+}} \dfrac{c_{\alpha} \langle \alpha, a \rangle}{\langle \alpha,\cdot \rangle}}  = 0.
\end{align*}
The first commutator yields
\begin{align*}
[ \partial_{a},  \langle x ,\xi \rangle  ] = \left[ \sum_{i} a_{i} \partial_{x_{i}} , \sum_{j} x_{j} \xi_{j} \right] =  \sum_{i} a_{i} ( \partial_{x_{i}} x_{i} - x_{i} \partial_{x_{i}} ) \xi_{i} ,
\end{align*}
since $\partial_{x_{i}} x_{j} - x_{j} \partial_{x_{i}} = 0$ for $i \neq j $. Using the guiding principle behind our Key Lemma \ref{keyLemma}, which states that $\partial_{x_{i}} x_{i} - x_{i} \partial_{x_{i}} = \dfrac{d}{ d x_{i}} x_{i} = 1 $, we get
\begin{align*}
[ \partial_{a},  \langle x ,\xi \rangle  ] = \sum_{i} a_{i} \xi_{i}  = \langle a, \xi \rangle.
\end{align*}
Now, consider the third commutator. We have
\begin{align*}
\left[  \sum_{\alpha \in R_{+}} \dfrac{c_{\alpha} \langle \alpha, a \rangle}{\langle \alpha, \cdot \rangle}\sigma_{\alpha}  , \langle x ,\xi \rangle  \right] &= \sum_{\alpha \in R_{+}} \dfrac{c_{\alpha} \langle \alpha, a \rangle}{\langle \alpha, \cdot \rangle}\sigma_{\alpha}   \langle x ,\xi \rangle   -      \langle x ,\xi \rangle  \sum_{\alpha \in R_{+}} \dfrac{c_{\alpha} \langle \alpha, a \rangle}{\langle \alpha,\cdot \rangle}\sigma_{\alpha} \\
&= \sum_{\alpha \in R_{+}} \dfrac{c_{\alpha} \langle \alpha, a \rangle}{\langle \alpha, \cdot \rangle} \left( \sigma_{\alpha} \langle x ,\xi \rangle  - \langle x ,\xi \rangle  \sigma_{\alpha} \right).
\end{align*}
Since $ \sigma_{\alpha}$ is an orthogonal reflection, we know that for any $u, v \in V$, we have $\langle \sigma_{\alpha}^{-1}u , v \rangle = \langle u, \sigma_{\alpha} v \rangle$.  Hence, we get
\begin{align*}
\left[  \sum_{\alpha \in R_{+}} \dfrac{c_{\alpha} \langle \alpha, a \rangle}{\langle \alpha, \cdot \rangle}\sigma_{\alpha}  , \langle x ,\xi \rangle  \right]
&= \sum_{\alpha \in R_{+}} \dfrac{c_{\alpha} \langle \alpha, a \rangle}{\langle \alpha, \cdot \rangle} \left( \sigma_{\alpha} \langle x ,\xi \rangle   - \langle x ,\xi \rangle  \sigma_{\alpha} \right) \\
&= \sum_{\alpha \in R_{+}} \dfrac{c_{\alpha} \langle \alpha, a \rangle}{\langle \alpha, \cdot \rangle}  \left(  \langle \sigma_{\alpha}^{-1} x ,\xi \rangle  \sigma_{\alpha}- \langle x ,\xi \rangle  \sigma_{\alpha} \right) \\
&=  \sum_{\alpha \in R_{+}} \dfrac{c_{\alpha} \langle \alpha, a \rangle}{\langle \alpha, \cdot \rangle}  \left(  \langle  x , \sigma_{\alpha} \xi \rangle \sigma_{\alpha}  - \langle x ,\xi \rangle  \sigma_{\alpha} \right).
\end{align*}
Now using the fact that $\sigma_{\alpha} \xi = \xi - 2 \dfrac{\langle \alpha, \xi \rangle}{\langle \alpha, \alpha \rangle} \alpha = \xi - \langle \alpha, \xi \rangle \alpha^{\vee}$, we have
\begin{align*}
\left[  \sum_{\alpha \in R_{+}} \dfrac{c_{\alpha} \langle \alpha, a \rangle}{\langle \alpha, \cdot \rangle}\sigma_{\alpha}  , \langle x ,\xi \rangle  \right]
&=  \sum_{\alpha \in R_{+}} \dfrac{c_{\alpha} \langle \alpha, a \rangle}{\langle \alpha, \cdot \rangle}  \left(  \langle  x , ( \xi - \langle \alpha, \xi \rangle \alpha^{\vee}) \rangle  \sigma_{\alpha} - \langle x ,\xi \rangle  \sigma_{\alpha} \right) \\
&= \sum_{\alpha \in R_{+}} \dfrac{c_{\alpha} \langle \alpha, a \rangle}{\langle \alpha, \cdot \rangle}  \left(  \cancel{ \langle x ,\xi \rangle  \sigma_{\alpha}} - \langle x,  \langle \alpha, \xi \rangle \alpha^{\vee} \rangle \sigma_{\alpha} -  \cancel{ \langle x ,\xi \rangle  \sigma_{\alpha}} \right) \\
&=  \sum_{\alpha \in R_{+}} \dfrac{c_{\alpha} \langle \alpha, a \rangle}{\langle \alpha, \cdot \rangle}  \langle x,  \langle \alpha, \xi \rangle \alpha^{\vee} \rangle \sigma_{\alpha}.
\end{align*}
Combining the two results, we obtain
\begin{align}
\left[ D_{a},  \langle x ,\xi \rangle  \right]  &= \langle a, \xi \rangle - \sum_{\alpha \in R_{+}} \dfrac{c_{\alpha} \langle \alpha, a \rangle}{\langle \alpha, \cdot \rangle}  \langle x,  \langle \alpha, \xi \rangle \alpha^{\vee} \rangle \sigma_{\alpha} \nonumber \\
&= \langle \alpha, \xi \rangle - \sum_{\alpha \in R_{+}} \dfrac{ 2 c_{\alpha} \langle \alpha, a \rangle \langle \alpha, \xi \rangle \cancel{\langle x , \alpha \rangle } }{\cancel{\langle \alpha , \cdot \rangle }\langle \alpha, \alpha \rangle} \sigma_{\alpha}. \label{propEtingof}
\end{align}
Next consider 
\begin{align}\label{dunklCommutator}
\left[ [ D_{a} , D_{b} ] , \langle x ,\xi \rangle \right] = [[ D_{a}, \langle x ,\xi \rangle ], D_{b} ] - [[ D_{b}, \langle x ,\xi \rangle] , D_{a} ],
\end{align}
where we have used the Jacobi identity to split the commutator in two. Using the result in Equation \eqref{propEtingof} the first commutator of Equation \eqref{dunklCommutator} is
\begin{align*}
[[ D_{a}, \langle x ,\xi \rangle ], D_{b} ]  &= \left[  \langle \alpha, \xi \rangle - \sum_{\alpha \in R_{+}} \dfrac{ 2 c_{\alpha} \langle \alpha, a \rangle \langle \alpha, \xi \rangle  }{\langle \alpha , \cdot \rangle \langle \alpha, \alpha \rangle} \sigma_{\alpha}, D_{b} \right] \\
&= \cancelto{0}{ \left[  \langle \alpha, \xi \rangle, D_{b} \right] }+ \left[ - \sum_{\alpha \in R_{+}} \dfrac{ 2 c_{\alpha} \langle \alpha, a \rangle \langle \alpha, \xi \rangle   }{\langle \alpha , \cdot \rangle \langle \alpha, \alpha \rangle} \sigma_{\alpha}, D_{b} \right].
\end{align*}
We only need to worry about the action of the orthogonal reflection $\sigma_{\alpha}$ on the Dunkl operator $D_{b}$, given by $\sigma_{\alpha} D_{b} = D_{\sigma_{\alpha} (b)} \sigma_{\alpha}$, thus, yielding
\begin{align*}
[[ D_{a}, \langle x ,\xi \rangle ], D_{b} ] &= - \left( \sum_{\alpha \in R_{+}} \dfrac{ 2 c_{\alpha} \langle \alpha, a \rangle \langle \alpha, \xi \rangle   }{\langle \alpha , \cdot \rangle \langle \alpha, \alpha \rangle} \left( \sigma_{\alpha} D_{b} - D_{b} \sigma_{\alpha} \right) \right) \\
&= \sum_{\alpha \in R_{+}} \dfrac{ 2 c_{\alpha} \langle \alpha, a \rangle \langle \alpha, \xi \rangle   }{\langle \alpha , \cdot \rangle \langle \alpha, \alpha \rangle}  \left( D_{b} \sigma_{\alpha} - \sigma_{\alpha} D_{b} \right).
\end{align*}
Consider for a moment the expression in the brackets. We get
\begin{align*}
D_{b} \sigma_{\alpha} - \sigma_{\alpha} D_{b} =  D_{b} \sigma_{\alpha} - D_{\sigma_{\alpha} (b) } \sigma_{\alpha}   = (D_{b - \sigma_{\alpha} (b)} ) \sigma_{\alpha} .
\end{align*}
Using Definition \ref{definitionOrthogonalReflection}, we get
\begin{align*}
D_{b} \sigma_{\alpha} - \sigma_{\alpha} D_{b} =  \left( D_{b - b + 2 \frac{\langle \alpha, b \rangle}{\langle \alpha, \alpha \rangle}\alpha} \right) \sigma_{\alpha} = D_{\langle \alpha, b \rangle \alpha^{\vee}} \sigma_{\alpha} = \langle \alpha , b \rangle D_{\alpha^{\vee}} \sigma_{\alpha}.
\end{align*}
Hence, going back to the commutator we get
\begin{align}\label{comm1}
[[ D_{a}, \langle x ,\xi \rangle ], D_{b} ] &=  \sum_{\alpha \in R_{+}} \dfrac{ 2 c_{\alpha} \langle \alpha, a \rangle \langle \alpha, \xi \rangle   }{\langle \alpha , \cdot \rangle \langle \alpha, \alpha \rangle}   \langle \alpha , b \rangle D_{\alpha^{\vee}} \sigma_{\alpha}. 
\end{align}
Now consider the second commutator of Equation \eqref{dunklCommutator}. We have
\begin{align*}
[[ D_{b}, \langle x ,\xi \rangle ], D_{a} ] &=  \sum_{\beta \in R_{+}} \dfrac{ 2 c_{\beta} \langle \beta, b \rangle \langle \beta, \xi \rangle   }{\langle \beta , \cdot \rangle \langle \beta, \beta \rangle}  \left( D_{a} \sigma_{\beta} - \sigma_{\beta} D_{a} \right).
\end{align*}
Again using Definition \ref{definitionOrthogonalReflection}, we get
\begin{align*}
D_{a} \sigma_{\beta} - \sigma_{\beta} D_{a} =  \left( D_{a - a + 2 \frac{\langle \beta, a \rangle}{\langle \beta, \beta \rangle}\beta} \right) \sigma_{\beta} = D_{\langle \beta, a \rangle \beta^{\vee}} \sigma_{\beta} = \langle \beta , a \rangle D_{\beta^{\vee}} \sigma_{\beta}.
\end{align*}
Hence, the expression we obtain for the commutator is 
\begin{align}\label{comm2}
[[ D_{b}, \langle x ,\xi \rangle ], D_{a} ] &=  \sum_{\beta \in R_{+}} \dfrac{ 2 c_{\beta} \langle \beta, b \rangle \langle \beta, \xi \rangle   }{\langle \beta , \cdot \rangle \langle \beta, \beta \rangle} \langle \beta , a \rangle D_{\beta^{\vee}} \sigma_{\beta}.
\end{align}
Combining the expressions in \eqref{comm1} and \eqref{comm2}, Equation \eqref{dunklCommutator} becomes
\begin{align*}
\left[ [ D_{a} , D_{b} ] , \langle x ,\xi \rangle \right] &= \sum_{\alpha \in R_{+}} \dfrac{ 2 c_{\alpha} \langle \alpha, a \rangle \langle \alpha, \xi \rangle   }{\langle \alpha , \cdot\rangle \langle \alpha, \alpha \rangle}   \langle \alpha , b \rangle D_{\alpha^{\vee}} \sigma_{\alpha} - \sum_{\beta \in R_{+}} \dfrac{ 2 c_{\beta} \langle \beta, b \rangle \langle \beta, \xi \rangle   }{\langle \beta , \cdot \rangle \langle \beta, \beta \rangle} \langle \beta , a \rangle D_{\beta^{\vee}} \sigma_{\beta} \\
&= \sum_{\alpha \in R_{+}} \dfrac{2 c_{\alpha} \langle \alpha, \xi \rangle }{\langle \alpha , \cdot \rangle \langle \alpha, \alpha \rangle} ( \langle \alpha, a \rangle \langle \alpha, b \rangle - \langle \alpha , b \rangle \langle \alpha, a \rangle ) D_{\alpha^{\vee}} \sigma_{\alpha} = 0,
\end{align*}
which indeed shows that the Dunkl operators commute for arbitrary $a , b \in \mathbb{C}^{n}$.
\end{proof}
Before we proceed with the next big result, we need the following proposition.

\begin{proposition}\label{heckmanProp}
Let $B(\cdot, \cdot) $ be a bilinear form on $V$ such that for any $a, b \in V$, $B( \sigma_{\alpha} a, \sigma_{\alpha} b ) = B(a, b)$, where $\alpha \in R \cap \spn \langle a,  b \rangle$. Let $w \in W$ be a plane rotation, i.e. $\sigma_{\alpha} \sigma_{\beta} = w$ and $w \neq id$. Then
\begin{align*}
\sum_{\substack{ \alpha, \beta \in R_{+} \\  \sigma_{\alpha} \sigma_{\beta} = w}} c_{\alpha} c_{\beta} B(\alpha, \beta ) \dfrac{ (1-\sigma_{\alpha})}{\langle \alpha, \cdot \rangle} \dfrac{ (1-\sigma_{\beta})}{\langle \beta , \cdot \rangle}= 0.
\end{align*}
\end{proposition}
\begin{proof}
Writing out the expression explicitly we get
\begin{align*}
\sum_{\substack{ \alpha, \beta \in R_{+} \\  \sigma_{\alpha} \sigma_{\beta} = w}} c_{\alpha} c_{\beta} B(\alpha, \beta ) \dfrac{ (1-\sigma_{\alpha})}{\langle \alpha, \cdot \rangle} \dfrac{ (1-\sigma_{\beta})}{\langle \beta , \cdot \rangle}= E_{1} + E_{2} + E_{3},  
\end{align*}
where 
\begin{align*}
E_{1} &= \sum_{\substack{ \alpha, \beta \in R_{+} \\  \sigma_{\alpha} \sigma_{\beta} = w}} c_{\alpha} c_{\beta} B(\alpha, \beta ) \dfrac{1}{\langle \alpha, \cdot \rangle \langle \beta, \cdot \rangle } , \\
E_{2} &= \sum_{\substack{ \alpha, \beta \in R_{+} \\  \sigma_{\alpha} \sigma_{\beta} = w}} c_{\alpha} c_{\beta} B(\alpha, \beta ) \left( -  \dfrac{1}{\langle \alpha , \cdot \rangle} \sigma_{\alpha}\dfrac{1}{\langle \beta , \cdot \rangle}  - \dfrac{1}{\langle \alpha , \cdot \rangle} \dfrac{1}{\langle \beta , \cdot \rangle}  \sigma_{\beta} \right) \\
&=  \sum_{\substack{ \alpha, \beta \in R_{+} \\  \sigma_{\alpha} \sigma_{\beta} = w}} c_{\alpha} c_{\beta} B(\alpha, \beta ) \left(   -  \dfrac{1}{\langle \alpha , \cdot \rangle} \dfrac{1}{\langle \sigma_{\alpha}\beta , \cdot \rangle}\sigma_{\alpha}  - \dfrac{1}{\langle \alpha , \cdot \rangle} \dfrac{1}{\langle \beta , \cdot \rangle}  \sigma_{\beta} \right), \qquad \text{ and } \\
E_{3} &= \sum_{\substack{ \alpha, \beta \in R_{+} \\  \sigma_{\alpha} \sigma_{\beta} = w}} c_{\alpha} c_{\beta} B(\alpha, \beta )   \left(   \dfrac{1}{\langle \alpha, \cdot \rangle} \sigma_{\alpha} \dfrac{1}{\langle \beta, \cdot \rangle}\sigma_{\beta} \right) \\
&=  \sum_{\substack{ \alpha, \beta \in R_{+} \\  \sigma_{\alpha} \sigma_{\beta} = w}} c_{\alpha} c_{\beta} B(\alpha, \beta )   \left(   \dfrac{1}{\langle \alpha, \cdot \rangle}\dfrac{1}{\langle \sigma_{\alpha} \beta, \cdot \rangle} \sigma_{\alpha} \sigma_{\beta} \right) .
\end{align*}
Let $S$ be the normalized root system of the largest dihedral group $W(S)$ and let $w$ be an element of $S$. If $w = \sigma_{\alpha} \sigma_{\beta}$, then for any $\gamma \in S$ we have $\sigma_{\gamma} w \sigma_{\gamma} = w^{-1}$. Hence, $\sigma_{\gamma}  \sigma_{\alpha} \sigma_{\gamma} \sigma_{\beta} = \sigma_{\beta} \sigma_{\alpha}$. For any $\gamma \in S$, we apply $\sigma_{\gamma}$ to $E_{1}$ to get
\begin{align*}
\sigma_{\gamma} E_{1} &= \sum_{\substack{ \alpha, \beta \in R_{+} \\  \sigma_{\alpha} \sigma_{\beta} = w}} c_{\alpha} c_{\beta} B ( \sigma_{\gamma} \alpha,  \sigma_{\gamma} \beta ) \dfrac{1}{\langle  \sigma_{\gamma} \alpha, \cdot \rangle \langle  \sigma_{\gamma} \beta, \cdot \rangle } \\
&=  \sum_{\substack{ \alpha, \beta \in R_{+} \\  \sigma_{\alpha} \sigma_{\beta} = w}} c_{\alpha} c_{\beta} B(\alpha, \beta ) \dfrac{1}{\langle \beta, \cdot \rangle \langle \alpha, \cdot \rangle } = E_{1},
\end{align*}
due to the invariance of the bilinear form $B\langle \cdot, \cdot \rangle $ and the independence of the sum on the root system $R_{+}$. Now let $p_{S} = \prod_{\alpha \in S_{+}} \langle \alpha , \cdot \rangle $, where $S_{+} = R_{+} \cap S$. Then for $\sigma_{\alpha} \in W(S)$, we have 
\begin{align*}
\sigma_{\alpha} (p_{S} E_{1}) = \sigma_{\alpha} (p_{S}) \sigma_{\alpha} (E_{1}) = \sigma_{\alpha} (p_{S}) E_{1} = - p_{S} E_{1},
\end{align*}
for all $\alpha \in R_{+} \cap S$, thus, the polynomial $p_{S} E_{1}$ is anti-invariant. But we have $p_{S} E_{1}$ is homogeneous and $\deg( p_{S} E_{1})= \deg(p_{S})-\deg(E_{1}) = \deg (p_{S}) - 2$, thus, $p_{S} E_{1}=0$. 
\par Now consider the expression $E{2}$. Since $\sigma_{\alpha}\sigma_{\beta}=w= \sigma_{\sigma_{\alpha}\beta}\sigma_{\alpha}$, and $B\langle \alpha, \beta \rangle = B \langle \sigma_{\alpha} \beta, \sigma_{\alpha} \alpha \rangle = - B \langle \sigma_{\alpha} \beta , \alpha \rangle$, we see that $E_{2}$ is as well identically zero. 
\par Finally, consider $E_{3}$. We have
\begin{align*}
E_{3} &=  \sum_{\substack{ \alpha, \beta \in R_{+} \\  \sigma_{\alpha} \sigma_{\beta} = w}} c_{\alpha} c_{\beta} B(\alpha, \beta )   \left(   \dfrac{1}{\langle \alpha, \cdot \rangle}\dfrac{1}{\langle \sigma_{\alpha} \beta, \cdot \rangle} \sigma_{\alpha} \sigma_{\beta} \right)  \\
&= -  \left( \sum_{\substack{ \alpha, \beta \in R_{+} \\  \sigma_{\alpha} \sigma_{\beta} = w}} c_{\alpha} c_{\beta} B(\alpha, \beta ) \dfrac{1}{\langle \alpha, \cdot \rangle \langle \beta, \cdot \rangle } \right) \sigma_{\alpha} \sigma_{\beta} = - E_{1} w = 0.
\end{align*}
\end{proof}
\par Our next goal is to use the commutativity property of the Dunkl operators to somehow recover the Hamiltonian of the quantum rational Calogero-Moser \eqref{quantumCM}. The first thing to observe is that the Dunkl operator contains only first order differentials, while the momenta terms in the Hamiltonian \eqref{quantumCM} are of order two. This hints that we need to square the Dunkl operators.

\subsection{The Olshanetsky-Perelomov Operator.}
The ingenious idea is established by Olshanetsky and Perelomov in \cite{OlshanetskyPerelomov}. Here we follow the procedure as outlined in \cite{Etingof}. First we define a new operator, which is a polynomial function in terms of the squares the Dunkl operator. Then we dismiss the terms, which let the expression leave the subalgebra of $W$-invariant functions. In other words, we \emph{restrict} the operator to the subspace of $W$-invariant functions. Hence, we see that this operator defines a quantum integrable system, which is closely related to our Calogero-Moser Hamiltonian. We begin by presenting the so called \emph{Olshanetsky-Perelomov} operator.
\begin{definition}
The \emph{Olshanetsky-Perelomov} operator associated with the $W$-invariant function $c_{\alpha}:R_{+} \to \mathbb{C}$ is the second order differential operator given by
\begin{align}\label{opOperator}
L = \Delta - \sum_{\alpha \in R_{+}} \dfrac{c_{\alpha} (c_{\alpha} + 1 ) \langle \alpha, \alpha \rangle}{\langle \alpha , \cdot \rangle^{2}},
\end{align}
where $\Delta = \sum_{j=1}^{n} \partial_{e_{j}}^{2}$ is the Laplacian.
\end{definition}
\par Note that the Olshanetsky-Perelomov operator is $W$-invariant, which is clear due to the invariance of the Laplacian and the lack of terms involving reflections. Our aim is to obtain a relationship between the Dunkl and the Olshanetsky-Perelomov and in order to do so, we introduce the following restriction mechanism.
\begin{definition}
Define the \emph{restriction operation} $\res$ of the operator $D_{a}$ to the subspace of $W$-invariant differentials by 
\begin{align*}
\res (D_{a}): (\mathbb{C}W \ltimes D(V_{reg}))^{W} \to (D(V_{reg}))^{W}.
\end{align*}
\end{definition}
\par Having introduced the restriction operation, we can now present the operator, consisting of the squared Dunkl operators, which we then restrict to the space of $W$-invariants. The result is due to Heckman \cite{Heckman}, but we follow the approach as in Etingof \cite{Etingof}.
\begin{proposition}\label{opprop}
Let $e_{1}, \dots, e_{n}$ be an orthonormal basis for $V$. Then we have
\begin{align*}
\res\left( \sum_{j=1}^{n} D_{e_{j}}^{2} \right)= \bar{L},
\end{align*}
where 
\begin{align}\label{opop}
\bar{L} = \Delta - \sum_{\alpha \in R_{+}} \dfrac{c_{\alpha} \langle \alpha, \alpha \rangle}{\langle \alpha , \cdot \rangle } \partial_{\alpha^{\vee}}.
\end{align}
\end{proposition}

\begin{proof}
We begin by computing the product $D_{\xi} D_{\xi}$ for $\xi \in \mathbb{C}^{n}$ and leave the summation until the end. Hence, we get
\begin{align*}
D_{\xi} D_{\xi}  &= \left( \partial_{\xi} - \sum_{\alpha \in R_{+}} \dfrac{c_{\alpha} \langle \alpha, \xi \rangle}{\langle \alpha, \cdot \rangle} (1-\sigma_{\alpha})\right) \left( \partial_{\xi} - \sum_{\alpha \in R_{+}} \dfrac{c_{\alpha} \langle \alpha,\xi \rangle}{\langle \alpha, \cdot \rangle} (1-\sigma_{\alpha})\right) \nonumber \\
&= \partial_{\xi}^{2} - \partial_{\xi} \sum_{\alpha \in R_{+}} \dfrac{c_{\alpha} \langle \alpha, \xi \rangle}{\langle \alpha, \cdot \rangle} (1-\sigma_{\alpha}) - \sum_{\alpha \in R_{+}} \dfrac{c_{\alpha} \langle \alpha, \xi \rangle}{\langle \alpha, \cdot \rangle} (1-\sigma_{\alpha}) \partial_{\xi}  \nonumber \\
& \qquad \qquad \qquad + \sum_{\alpha \in R_{+}}\sum_{\beta \in R_{+}} \dfrac{c_{\alpha} c_{\beta} \langle \alpha, \xi \rangle \langle \beta,\xi \rangle }{\langle \alpha, \cdot \rangle \langle \beta, \cdot \rangle} (1-\sigma_{\alpha})(1-\sigma_{\beta}) \nonumber \\
&= \partial_{\xi}^{2} - \sum_{\alpha \in R_{+}} c_{\alpha} \langle \alpha, \xi \rangle \left( \partial_{\xi}  \dfrac{(1-\sigma_{\alpha})}{ \langle \alpha, \cdot \rangle} + \dfrac{(1-\sigma_{\alpha})}{ \langle \alpha, \cdot \rangle} \partial_{\xi} \right)   \\
& \qquad \qquad \qquad + \sum_{\alpha \in R_{+}}\sum_{\beta \in R_{+}} \dfrac{c_{\alpha} c_{\beta} \langle \alpha, \xi \rangle \langle \beta, \xi \rangle }{\langle \alpha, \cdot \rangle \langle \beta, \cdot \rangle} (1-\sigma_{\alpha})(1-\sigma_{\beta}) .\nonumber
\end{align*}
 Now using that $\sigma_{\alpha} \partial_{\xi} = \partial_{\sigma_{\alpha}(\xi)} \sigma_{\alpha}$, we get
\begin{align*}
D_{\xi} D_{\xi} &= \partial_{\xi}^{2} - \sum_{\alpha \in R_{+}} \dfrac{c_{\alpha} \langle \alpha , \xi \rangle}{\langle \alpha , \cdot \rangle} \left( \partial_{\xi} (1- \sigma_{\alpha}) + (1-\sigma_{\alpha}) \partial_{\xi}  \right) \\
& \qquad \qquad \qquad + \sum_{\alpha \in R_{+}}\sum_{\beta \in R_{+}} \dfrac{c_{\alpha} c_{\beta} \langle \alpha, \xi \rangle \langle \beta, \xi \rangle }{\langle \alpha, \cdot \rangle \langle \beta, \cdot \rangle} (1-\sigma_{\alpha})(1-\sigma_{\beta}) \\
&= \partial_{\xi}^{2} - \sum_{\alpha \in R_{+}} \dfrac{c_{\alpha} \langle \alpha , \xi \rangle}{\langle \alpha , \cdot \rangle} \left( \partial_{\xi} - \partial_{\xi}\sigma_{\alpha} + \partial_{\xi} - \sigma_{\alpha} \partial_{\xi}  \right)  \\
& \qquad \qquad \qquad + \sum_{\alpha \in R_{+}}\sum_{\beta \in R_{+}} \dfrac{c_{\alpha} c_{\beta} \langle \alpha, \xi \rangle \langle \beta, \xi \rangle }{\langle \alpha, \cdot \rangle \langle \beta, \cdot \rangle} (1-\sigma_{\alpha})(1-\sigma_{\beta}) \\
&= \partial_{\xi}^{2} - \sum_{\alpha \in R_{+}} \dfrac{c_{\alpha} \langle \alpha , \xi \rangle}{\langle \alpha , \cdot \rangle} \left( \partial_{\xi} - \partial_{\xi}\sigma_{\alpha} + \partial_{\xi} -  \partial_{\sigma_{\alpha}(\xi)} \sigma_{\alpha}   \right)  \\
& \qquad \qquad \qquad + \sum_{\alpha \in R_{+}}\sum_{\beta \in R_{+}} \dfrac{c_{\alpha} c_{\beta} \langle \alpha, \xi \rangle \langle \beta, \xi \rangle }{\langle \alpha, \cdot \rangle \langle \beta, \cdot \rangle} (1-\sigma_{\alpha})(1-\sigma_{\beta}) \\
&= \partial_{\xi}^{2} - \sum_{\alpha \in R_{+}} \dfrac{c_{\alpha} \langle \alpha , \xi \rangle}{\langle \alpha , \cdot \rangle} \left( \partial_{\xi} - \partial_{\xi}\sigma_{\alpha} + \partial_{\xi} -  \partial_{\xi} \sigma_{\alpha} + \langle \alpha, \xi \rangle \partial_{\alpha^{\vee}} 
\sigma_{\alpha}  \right)   \\
& \qquad \qquad \qquad + \sum_{\alpha \in R_{+}}\sum_{\beta \in R_{+}} \dfrac{c_{\alpha} c_{\beta} \langle \alpha, \xi \rangle \langle \beta, \xi \rangle }{\langle \alpha, \cdot \rangle \langle \beta, \cdot \rangle} (1-\sigma_{\alpha})(1-\sigma_{\beta}),
\end{align*}
where we have simply applied the orthogonal reflection to the partial derivative $\partial_{\xi}$. Next we group the terms as follows
\begin{align*}
D_{\xi} D_{\xi} &= \partial_{\xi}^{2} - \sum_{\alpha \in R_{+}} \dfrac{c_{\alpha} \langle \alpha , \xi \rangle}{\langle \alpha , \cdot \rangle} \left( 2\partial_{\xi} (1- \sigma_{\alpha})  \right) - \sum_{\alpha \in R_{+}} \dfrac{c_{\alpha} \langle \alpha , \xi \rangle^{2}}{\langle \alpha , \cdot \rangle} \partial_{\alpha^{\vee}} \sigma_{\alpha}  \\
& \qquad \qquad \qquad + \sum_{\alpha \in R_{+}}\sum_{\beta \in R_{+}} \dfrac{c_{\alpha} c_{\beta} \langle \alpha, \xi \rangle \langle \beta, \xi \rangle }{\langle \alpha, \cdot \rangle \langle \beta, \cdot \rangle} (1-\sigma_{\alpha})(1-\sigma_{\beta}).
\end{align*}
Summing over $j = 1, \dots, n$ for all possible $e_{j}$ we get
\begin{align*}
 \left( \sum_{j=1}^{n} D_{e_{j}}^{2} \right)&= \sum_{j=1}^{n} \partial_{e_{j}}^{2}  - \sum_{\alpha \in R_{+}} c_{\alpha} \dfrac{\sum_{i=1}^{n} \langle \alpha , e_{i} \rangle^{2}}{\langle \alpha, \cdot \rangle} \partial_{\alpha^{\vee}} +\sum_{\alpha \in R_{+}}\sum_{\beta \in R_{+}} \dfrac{c_{\alpha} c_{\beta} \sum_{j=1}^{n} \langle \alpha, e_{j} \rangle \langle \beta, e_{j} \rangle }{\langle \alpha, \cdot \rangle \langle \beta, \cdot \rangle} (1-\sigma_{\alpha})(1-\sigma_{\beta}), \\
&= \sum_{j=1}^{n} \partial_{e_{j}}^{2}  - \sum_{\alpha \in R_{+}} c_{\alpha} \dfrac{ \langle \alpha , \alpha \rangle}{\langle \alpha, \cdot \rangle} \partial_{\alpha^{\vee}} + \cancelto{0}{\sum_{\alpha \in R_{+}}\sum_{\beta \in R_{+}} \dfrac{c_{\alpha} c_{\beta}  \langle \alpha, \beta\rangle }{\langle \alpha, \cdot \rangle \langle \beta, \cdot \rangle} (1-\sigma_{\alpha})(1-\sigma_{\beta})}
\end{align*}
since $\displaystyle{\sum_{i=1}^{n}} \langle \alpha , e_{i} \rangle^{2} = \langle \alpha, \alpha \rangle$ and the last term is identically zero by Proposition \ref{heckmanProp}.
\end{proof}
\subsection{Quantum Integrability.}

We now use Chevalley's Theorem \ref{chevalley} to conjecture that the restriction of the operator in Equation \eqref{opop} defines a quantum integrable system. 
\begin{theorem}\label{heckmanTheorem}
Suppose by the Chevalley theorem that $\mathbb{C}[V]^{W} \cong \mathbb{C}[p_{1}, \dots, p_{n}]$ with $p_{1}, \dots, p_{n}$ being homogeneous polynomial functions of degrees $d_{1} \leq \dots \leq d_{n}$. Then the set of operators $\bar{L}_{i} = \{ \res (p_{i} (D_{e_{1}}, \dots D_{e_{n}} ) )\}$ forms a quantum integrable system.
\end{theorem}
\begin{proof}
Let $p_{i} \in \mathbb{C}[V]^{W}$, for $i=1,\dots, n$. To see why the operators indeed commute consider the commutator
\begin{align*}
[\bar{L}_{i}, \bar{L}_{j} ] = \left[ \res\left( p_{i} \left( D_{\xi} \right)\right), \res\left( p_{j} \left( D_{\xi} \right) \right) \right].
\end{align*}
Then we have
\begin{align*}
 \res\left( p_{i} \left( D_{\xi} \right)\right)  \res\left( p_{j} \left( D_{\xi} \right)\right) &=  \res\left( p_{i} \left( D_{\xi} \right) p_{j} \left( D_{\xi} \right) \right) \\
&= \res \left( p_{j} \left( D_{\xi} \right) p_{i} \left( D_{\xi} \right) \right) \\
&=  \res\left( p_{j} \left( D_{\xi} \right)\right)  \res\left( p_{i} \left( D_{\xi} \right)\right),
\end{align*}
where $p_{i} \left( D_{\xi} \right) p_{j} \left( D_{\xi} \right) = p_{j} \left( D_{\xi} \right)  p_{i} \left( D_{\xi} \right)$ holds due to the commutativity of the Dunkl operators.
\end{proof}

\par Finally, we are ready to establish the connection between Dunkl operators and the quantum rational Calogero-Moser Hamiltonian. In order to do so, we consider a gauge transformation of the operator $\bar{L}$, which yields the Olshanetsky-Perelomov operator \ref{opOperator}. The procedure is discussed in the proposition below, following the method outlined in \cite{Etingof}.
\begin{theorem}
 Let $\delta_{c} = \displaystyle{\prod_{\alpha \in R_{+}}} \langle \alpha , \cdot \rangle^{c_{\alpha}} $ and let $\bar{L}$ and $L$ be as in Equations \eqref{opOperator} and \eqref{opop} respectively. Then
\begin{align*}
\delta_{c}^{-1} \circ \bar{L} \circ \delta_{c} = L.
\end{align*}
\end{theorem}
\begin{proof}
We want to show that $\delta_{c}^{-1} \circ \bar{L} \circ \delta_{c} = L$. This is equivalent to 
\begin{align*}
\delta_{c} \circ L \circ \delta_{c}^{-1} = \bar{L} = \Delta - \sum_{\alpha \in R_{+}} \dfrac{c_{\alpha} \langle \alpha, \alpha \rangle}{\langle \alpha , \cdot \rangle } \partial_{\alpha^{\vee}}.
\end{align*}

Notice that
\begin{align*}
 \sum_{\alpha \in R_{+}} \dfrac{c_{\alpha} \langle \alpha, \alpha \rangle}{\langle \alpha , \cdot \rangle } \partial_{\alpha^{\vee}} = \sum_{\alpha \in R_{+}} \dfrac{2 c_{\alpha}}{\langle \alpha , \cdot \rangle} \partial_{\alpha},
\end{align*}
thus, our aim is to show that
\begin{align*}
\delta_{c} \circ L \circ \delta_{c}^{-1} =\Delta -  \sum_{\alpha \in R_{+}} \dfrac{2 c_{\alpha}}{\langle \alpha , \cdot \rangle} \partial_{\alpha}.
\end{align*}
Let us now perform the composition on $L$. We have
\begin{align}\label{lcomposite}
\delta_{c} \circ L \circ \delta_{c}^{-1} = \delta_{c} \circ \Delta \circ \delta_{c}^{-1} - \delta_{c} \circ  \sum_{\alpha \in R_{+}} \dfrac{c_{\alpha} (c_{\alpha} + 1 ) \langle \alpha, \alpha \rangle}{\langle \alpha , \cdot \rangle^{2}}  \circ \delta_{c}^{-1}.
\end{align}
We treat the two terms separately. Consider first the composition of the Laplacian $\Delta$. We get
\begin{align}\label{laplaceComposite}
 \delta_{c} \circ \Delta \circ \delta_{c}^{-1}  = \delta_{c} \left( \Delta \delta_{c}^{-1} + 2 \sum_{j} \partial_{e_{j}} \delta_{c}^{-1} \partial_{e_{j}} + \delta_{c}^{-1} \Delta \right).
\end{align}
The first term in Equation \eqref{laplaceComposite} yields
\begin{align*}
\delta_{c} \left( \Delta \delta_{c}^{-1} \right) &=  \delta_{c} \left( \sum_{j} \partial_{e_{j}}^{2} \prod_{\alpha \in R_{+}} \langle \alpha , \cdot \rangle^{c_{\alpha}} \right) \\
&= \delta_{c} \left( \sum_{\alpha \in R_{+}} -c_{\alpha} ( -c_{\alpha}-1 ) \langle \alpha, \alpha \rangle \prod_{\alpha \in R_{+}} \langle \alpha, \cdot \rangle^{-c_{\alpha}-2} + \sum_{\substack{\alpha,\beta \in R_{+} \\ \alpha \neq \beta}} c_{\alpha} c_{\beta} \dfrac{ \langle \alpha , \beta \rangle}{\langle \alpha, \cdot \rangle \langle \beta, \cdot \rangle} \prod_{\alpha \in R_{+}} \langle \alpha , \cdot \rangle^{c_{\alpha}} \right),
\end{align*}
where in the last term we have accounted for the different roots of $R_{+}$. Upon multiplying by $\delta_{c}$ we get
\begin{align*}
\delta_{c} \left( \Delta \delta_{c}^{-1} \right) &= \cancel{\delta_{c}} \sum_{\alpha \in R_{+}} c_{\alpha} ( c_{\alpha}+1) \dfrac{\langle \alpha , \alpha \rangle}{\langle \alpha , \cdot \rangle} \cancel{\delta_{c}^{-1}} + \cancel{\delta_{c}} \sum_{\substack{\alpha,\beta \in R_{+} \\ \alpha \neq \beta}} c_{\alpha} c_{\beta} \dfrac{ \langle \alpha , \beta \rangle}{\langle \alpha, \cdot \rangle \langle \beta, \cdot \rangle} \cancel{\delta_{c}^{-1}}.
\end{align*}
Consider the summation over the different roots $\alpha \neq \beta$. It is clearly $W$-invariant. Then the $W$-anti-invariant expression $\delta_{c} \sum_{\alpha \neq \beta \in R_{+}}$ is of degree $\deg (\delta_{c} \sum_{\alpha \neq \beta \in R_{+}}) = \deg(\delta_{c}) - 2$ and hence, identically zero. Now consider the second term of Equation \eqref{laplaceComposite}. We have
\begin{align*}
\delta_{c} \left( 2 \sum_{j} \partial_{e_{j}} \delta_{c}^{-1} \partial_{e_{j}} \right) &= \delta_{c} \left( \sum_{j} 2 c_{\alpha} \langle \alpha, \alpha \rangle \prod_{\alpha \in R_{+}} \langle \alpha, \cdot \rangle^{-c_{\alpha}-1} \partial_{e_{j}} \right) \\
&= \cancel{\delta_{c}} \sum_{j} \dfrac{2 c_{\alpha} \langle \alpha, \alpha \rangle}{\langle \alpha , \cdot \rangle} \cancel{\delta_{c}^{-1}} \partial_{e_{j}} \\
&= \sum_{j} \dfrac{2 c_{\alpha} \langle \alpha, \alpha \rangle}{\langle \alpha , \cdot \rangle}\partial_{e_{j}} = \sum_{\alpha \in R_{+}} \dfrac{2 c_{\alpha}}{\langle \alpha , \cdot \rangle} \partial_{\alpha}.
\end{align*}
The last term of Equation \eqref{laplaceComposite} is simply $\delta_{c} \delta_{c}^{-1} \Delta = \Delta$. Combining the three results, we get
\begin{align*}
\delta_{c} \circ \Delta \circ \delta_{c}^{-1} =  \Delta + \sum_{\alpha \in R_{+}} c_{\alpha} ( c_{\alpha}+1) \dfrac{\langle \alpha , \alpha \rangle}{\langle \alpha , \cdot \rangle} +\sum_{\alpha \in R_{+}} \dfrac{2 c_{\alpha}}{\langle \alpha , \cdot \rangle} \partial_{\alpha}.
\end{align*}
The only thing left is to consider the last term of Equation \eqref{lcomposite}. Since there are no operators involved, the composition by $\delta_{c}$ and $\delta_{c}^{-1}$ is simply equivalent to the multiplication on the left and on the right by the respective terms. Thus, we have
\begin{align*}
\delta_{c} \circ L \circ \delta_{c}^{-1} &=  \Delta + \cancel{\sum_{\alpha \in R_{+}} c_{\alpha} ( c_{\alpha}+1) \dfrac{\langle \alpha , \alpha \rangle}{\langle \alpha , \cdot \rangle}} + \sum_{\alpha \in R_{+}} \dfrac{2 c_{\alpha}}{\langle \alpha , \cdot \rangle} \partial_{\alpha} -  \cancel{\sum_{\alpha \in R_{+}} \dfrac{c_{\alpha} (c_{\alpha} + 1 ) \langle \alpha, \alpha \rangle}{\langle \alpha , \cdot \rangle^{2}}} \\
&=\Delta -  \sum_{\alpha \in R_{+}} \dfrac{2 c_{\alpha}}{\langle \alpha , \cdot \rangle} \partial_{\alpha} = \bar{L},
\end{align*}
where we have recovered the operator $\bar{L}$ as required.
\end{proof}
\par As a consequence of Theorem \ref{heckmanTheorem}, we have the following statement.
\begin{corollary}
The Olshanetsky-Perelomov operator $L$ defines a quantum integrable system with $L_{i} = \delta_{c}^{-1} \circ \bar{L}_{i} \circ \delta_{c}$.
\end{corollary}
We are now ready to present the final result of this section, namely, the integrability of the quantum Calogero-Moser system.
\begin{theorem}
Let $R = A_{n-1}$. Then $L = \widehat{H}$, where $\widehat{H}$ is the Hamiltonian of the quantum rational Calogero-Moser system \eqref{quantumCM}, and the system is in turn integrable.
\end{theorem}
\begin{proof}
The corresponding group $W$ of the root system of type $A_{n-1}$ is the symmetric group $\mathcal{S}_{n}$. Thus, $c_{\alpha} = k$. The integrability is clear from the commutativity of the Dunkl operators.
\end{proof}
\par We shall not discuss the uniqueness of the Olshanetsky-Perelomov operators, but details can be found in \cite{Etingof2} and \cite{EtingofStrickland}. In a similar manner, the Liouville integrability of the classical rational Calogero-Moser is obtained via the classical Dunkl operators, which we introduce in the next section.

\section{Classical Integrability via Dunkl Operators.}\label{sec:dunklClassical}
\subsection{Classical Dunkl Operators, Involution.}
\indent \par The classical analogue of the Dunkl operator $D_{a}$ is derived by mapping the partial operator $\partial_{\widehat{q}_j}$ to the classical momentum $p_{j}$, that is, $\partial_{\widehat{q}_j}\mapsto p_{j}$, under suitable conditions, see \cite{klein}. Hence, we have the following definition.
\begin{definition}
For $a \in \mathbb{C}^{n}$, the \emph{classical rational Dunkl operator} is defined by
\begin{align}\label{classicalDunkl}
D_{a}^{cl} = p_{a} - \sum_{\alpha \in R_{+}} \dfrac{c_{\alpha} \langle \alpha, a \rangle}{\langle \alpha, \cdot \rangle} (1-\sigma_{\alpha}).
\end{align}
The classical Dunkl operator is an element of the semi-direct product of the group algebra $\mathbb{C}W$ and the algebra of holomorphic functions on the cotangent space of $V_{reg}$, i.e.  $D^{cl}_{a} \in \mathbb{C}W \ltimes \mathcal{O} (T^{*}(V_{reg}))$.
\end{definition}
\begin{remark} In order to be more precise, we go back to the quantum setting for a moment. Let $\hbar$ denote Planck's constant. Then the classical Dunkl operators are obtained by renormalizing the Dunkl operators $D_{a}$ as
\begin{equation*}
D_{a} (c, \hbar) = \hbar D_{a} (c/\hbar)
\end{equation*}
and then reducing modulo $\hbar$. See Etingof \cite{Etingof} for more details.
\end{remark}
\par The classical Dunkl operators are also equivariant under $W$. 
\begin{proposition}
Let $w \in W$. Then $w D^{cl}_{a} = D^{cl}_{w a } w$.
\end{proposition}
\begin{proof}
Applying $w$ on the left, we have
\begin{align*}
w D_{a}^{cl}  &= w p_{a} - w \sum_{\alpha \in R_{+}} \dfrac{c_{\alpha} \langle \alpha, a \rangle}{\langle \alpha, \cdot\rangle} (1-\sigma_{\alpha}) = p_{w(a)}w - \sum_{\alpha \in R_{+}} \dfrac{c_{\alpha} \langle \alpha, a \rangle}{\langle \alpha, w^{-1}\cdot \rangle} (w-\sigma_{w\alpha}w) \\
&= p_{w(a)}w - \left(  \sum_{\alpha \in R_{+}} \dfrac{c_{\alpha} \langle \alpha, a \rangle}{\langle \alpha, w^{-1}\cdot \rangle} (1-\sigma_{w\alpha}) \right)w
= \left(p_{w(a)} -  \sum_{\alpha \in R_{+}} \dfrac{c_{\alpha} \langle \alpha, a \rangle}{\langle \alpha, w^{-1}\cdot \rangle} (1-\sigma_{w\alpha}) \right) w \\
&= D^{cl}_{w(a)}w.
\end{align*}
\end{proof}
Similarly to the quantum operators, the classical ones are in involution. Here we demonstrate this property explicitly, but the reader is invited to consult \cite{klein} for a systematized approach on considering classical theories as the $\hbar \to 0$ limit of quantum mechanics. Note that the involutiveness is in the sense of Equation \eqref{poissonMap} with respect to the canonical coordinates on the phase space $(x_{i}, p_{i})$.

\begin{theorem}\label{poisson_dunkl}
The classical Dunkl operators Poisson commute for any $a, b \in \mathbb{C}^{n}$, that is, $\{D^{cl}_{a}, D^{cl}_{b} \}=0$.
\end{theorem}
\begin{proof}
We have
\begin{align*}
D_{a}^{cl} = p_{a} - \sum_{\alpha \in R_{+}} \dfrac{c_{\alpha} \langle \alpha, a \rangle}{\langle \alpha, \cdot \rangle} (1-\sigma_{\alpha}) \text{ and }
D_{b}^{cl} = p_{b} - \sum_{\alpha \in R_{+}} \dfrac{c_{\alpha} \langle \alpha, b \rangle}{\langle \alpha, \cdot \rangle} (1-\sigma_{\alpha}).
\end{align*}
The Poisson bracket of the two Dunkl operators is
\begin{align}
\sum_{i=a,b} \left( \dfrac{\partial D_{a}^{cl} }{\partial p_{i}} \dfrac{\partial D_{b}^{cl} }{\partial x_{i}} - \dfrac{\partial D_{b}^{cl} }{\partial p_{i}}\dfrac{\partial D_{a}^{cl} }{\partial x_{i}} \right) &= \cancelto{1}{\dfrac{\partial D_{a}^{cl} }{\partial p_{a}}} \dfrac{\partial D_{b}^{cl} }{\partial x_{a}} - \cancelto{0}{\dfrac{\partial D_{b}^{cl} }{\partial p_{a}}\dfrac{\partial D_{a}^{cl} }{\partial x_{a}}} + \cancelto{0}{\dfrac{\partial D_{a}^{cl} }{\partial p_{b}} \dfrac{\partial D_{b}^{cl} }{\partial x_{b}}} - \cancelto{1}{\dfrac{\partial D_{b}^{cl} }{\partial p_{b}}}\dfrac{\partial D_{a}^{cl} }{\partial x_{b}} \nonumber \\
&= \dfrac{\partial D_{b}^{cl} }{\partial x_{a}}  - \dfrac{\partial D_{a}^{cl} }{\partial x_{b}} \label{poisson_dunkl_classical}
\end{align}
as $D_{a}^{cl} $ and $D_{b}^{cl} $ are independent of $p_{b}$ and $p_{a}$, respectively, and $\dfrac{\partial D_{a}^{cl}}{\partial p_{a}} = \dfrac{\partial D_{b}^{cl}}{\partial p_{b}} = 1$. So, we get
\begin{align*}
\dfrac{\partial D_{b}^{cl} }{\partial x_{a}}  &= \dfrac{\partial }{\partial x_{a}} \left( p_{b} - \sum_{\alpha \in R_{+}} \dfrac{c_{\alpha} \langle \alpha, b \rangle}{\langle \alpha, \cdot \rangle} (1-\sigma_{\alpha}) \right) = \cancelto{0}{\dfrac{\partial p_{b}}{\partial x_{a}}} - \dfrac{\partial }{\partial x_{a}} \left(\sum_{\alpha \in R_{+}} \dfrac{c_{\alpha} \langle \alpha, b \rangle}{\langle \alpha, \cdot \rangle} (1-\sigma_{\alpha}) \right)  \\
&=  \sum_{\alpha \in R_{+}} \dfrac{c_{\alpha} \langle \alpha, b \rangle \langle \alpha, a \rangle}{\langle \alpha, \cdot \rangle^{2}} (1-\sigma_{\alpha}),
\end{align*}
and similarly,
\begin{align*}
\dfrac{\partial D_{a}^{cl} }{\partial x_{b}}  =  \sum_{\alpha \in R_{+}} \dfrac{c_{\alpha} \langle \alpha, a \rangle \langle \alpha, b \rangle}{\langle \alpha, \cdot \rangle^{2}} (1-\sigma_{\alpha}),
\end{align*}
and substituting in Equation \eqref{poisson_dunkl_classical}, we have
\begin{align*}
 \dfrac{\partial D_{b}^{cl} }{\partial x_{a}}  - \dfrac{\partial D_{a}^{cl} }{\partial x_{b}} = 0.
\end{align*}
\end{proof}

\par Once again, our next step is to utilize the Poisson commutativity of the Classical Dunkl operators \eqref{classicalDunkl} in order to recover the Hamiltonian of the classical rational Calogero-Moser \eqref{CM}. The classical Dunkl operator \eqref{classicalDunkl} contains only first order momentum terms, in contrast with the squared momentum terms in the Hamiltonian \eqref{CM}.

\subsection{The Classical Olshanetsky-Perelomov Operator.}
\par Recall that in the quantum case, we recovered the Calogero-Moser Hamiltonian \eqref{quantumCM} via the Olshanetsky-Perelomov operator \eqref{opop}. Similarly, we define the \emph{classical version of the Olshanetsky-Perelomov operator} as in \cite{Etingof}, which will then enable us to obtain the Hamiltonian \ref{CM}.
\begin{definition}
The \emph{classical Olshanetsky-Perelomov} operator associated with the $W$-invariant function $c_{\alpha}: R_{+} \to \mathbb{C}$ is defined as
\begin{align}\label{opOperator}
L^{cl} = \sum_{j} p_{j}^{2} - \sum_{\alpha \in R_{+}} \dfrac{c_{\alpha}^{2} \langle \alpha, \alpha \rangle}{\langle \alpha , \cdot \rangle^{2}}.
\end{align}
\end{definition}

\par As before, we define the restriction operation, necessary for the further computations.
\begin{definition}
Given a classical Dunkl operator $D_{a}^{cl}$, define the \emph{restriction operation} $\rescl$ by the rule
\begin{align*}
\rescl (D_{a}^{cl}) : (\mathbb{C}W \ltimes \mathcal{O} (T^{*}(V_{reg})))^{W} & \to ( \mathcal{O} (T^{*}(V_{reg})))^{W}, \\
\rescl \left(\sum f_{g} \cdot g \right) & \mapsto \sum f_{g},
\end{align*}
for $g \in W$, $f \in \mathcal{O} (T^{*} V_{reg})$.
\end{definition}
 
 \par Now that we are equipped with the necessary tools, we introduce the operator, containing the square of the classical Dunkl operator, which we shall then restrict to the space of $W$-invariants. We follow the steps outlined by Etingof in \cite{Etingof}.

\begin{proposition}
Let $e_{1}, \dots, e_{n}$ be an orthonormal basis for $V$. Then we have
\begin{align*}
\rescl\left( \sum_{j=1}^{n} (D_{e_{j}}^{cl})^{2} \right)= \bar{L}^{cl},
\end{align*}
where 
\begin{align}\label{opopc}
\bar{L}^{cl} = \sum_{j} p_{j}^{2} - \sum_{\alpha \in R_{+}} \dfrac{c_{\alpha} \langle \alpha, \alpha \rangle}{\langle \alpha , \cdot \rangle } p_{\alpha^{\vee}}.
\end{align}
\end{proposition}
\begin{proof}
We begin by computing the square of the classical Dunkl operator, namely
\begin{align*}
(D_{e_{j}}^{cl})^{2} = p_{j}^{2}- p_{j} \sum_{\alpha \in R_{+}} \dfrac{c_{\alpha} \langle \alpha, e_{j} \rangle}{\langle \alpha, \cdot \rangle} (1-\sigma_{\alpha}) &+  \sum_{\alpha \in R_{+}} \dfrac{c_{\alpha} \langle \alpha, e_{j} \rangle}{\langle \alpha, \cdot \rangle} (1-\sigma_{\alpha})p_{j}  \\
 &+ \cancelto{0}{\sum_{\alpha \in R_{+}}\sum_{\beta \in R_{+}} c_{\alpha} c_{\beta} \langle \alpha, \beta \rangle \dfrac{(1-\sigma_{\alpha})}{\langle \alpha, \cdot \rangle} \dfrac{(1-\sigma_{\beta})}{ \langle \beta, \cdot \rangle}},
\end{align*}
where again we cancel the last term due to Proposition \ref{heckmanProp}. Next, we have
\begin{align*}
(D_{e_{j}}^{cl})^{2} &= p_{j}^{2}- p_{j} \sum_{\alpha \in R_{+}} \dfrac{c_{\alpha} \langle \alpha, e_{j} \rangle}{\langle \alpha, \cdot \rangle} (1-\sigma_{\alpha}) +  \sum_{\alpha \in R_{+}} \dfrac{c_{\alpha} \langle \alpha, e_{j} \rangle}{\langle \alpha, \cdot \rangle} (1-\sigma_{\alpha})p_{j}  \\
&= p_{j}^{2} -  \sum_{\alpha \in R_{+}} \dfrac{c_{\alpha} \langle \alpha, e_{j} \rangle}{\langle \alpha, \cdot \rangle} \left( p_{j} (1-\sigma_{\alpha}) + (1-\sigma_{\alpha}) p_{j} \right) \\
&= p_{j}^{2} -  \sum_{\alpha \in R_{+}} \dfrac{c_{\alpha} \langle \alpha, e_{j} \rangle}{\langle \alpha, \cdot \rangle}   \left( p_{j} - p_{j} \sigma_{\alpha} + p_{j} -\sigma_{\alpha} p_{j} \right) \\
&= p_{j}^{2} -  \sum_{\alpha \in R_{+}} \dfrac{c_{\alpha} \langle \alpha, e_{j} \rangle}{\langle \alpha, \cdot \rangle}   \left( p_{j} - p_{j} \sigma_{\alpha} + p_{j} - p_{j} \sigma_{\alpha} +\langle \alpha, e_{j} \rangle p_{\alpha^{\vee}} \sigma_{\alpha} \right),
\end{align*}
where we have applied the orthogonal reflection to the $p_{j}$. Next we group the terms as follows
\begin{align*}
(D_{e_{j}}^{cl})^{2} &= p_{j}^{2} -  \sum_{\alpha \in R_{+}} \dfrac{c_{\alpha} \langle \alpha, e_{j} \rangle}{\langle \alpha, \cdot \rangle}   ( 2p_{j} (1- \sigma_{\alpha}  ))- \sum_{\alpha \in R_{+}} \dfrac{c_{\alpha} \langle \alpha, e_{j} \rangle^{2}}{\langle \alpha, \cdot \rangle}   p_{\alpha^{\vee}} \sigma_{\alpha} 
\end{align*}
Now applying the restriction operation gives
\begin{align*}
\rescl\left( (D_{e_{j}}^{cl})^{2} \right) &= p_{j}^{2}-\sum_{\alpha \in R_{+}} \dfrac{c_{\alpha} \langle \alpha, e_{j} \rangle^{2}}{\langle \alpha, \cdot \rangle}   p_{\alpha^{\vee}} 
\end{align*}
and upon summing over all values of $p_{j}$ for $j=1,\dots, n$, we recover $\bar{L}^{cl}$ as required.
\end{proof}

\subsection{Classical Integrability.}
The natural question to ask is how do we obtain the Olshanetsky-Perelomov operator from the restriction of the squared classical Dunkl operators. The procedure is somewhat identical to the one in the quantum setting. We follow the approach outlined in Etingof \cite{Etingof}.
\begin{proposition}\label{propclassicalOP}
Define the automorphism $\theta_{c} : \mathbb{C}W \ltimes \mathcal{O} (T^{*}(V_{reg})) \to  \mathbb{C}W \ltimes \mathcal{O} (T^{*}(V_{reg})) $ by
\begin{align*}
\theta_{c}(x) = x, \quad \theta_{c}(\sigma_{\alpha})=\sigma_{\alpha}, \quad \text{ and } \quad \theta_{c}(p_{a})=p_{a}+\partial_{a}\log \delta_{c},
\end{align*}
where $\delta_{c} = \displaystyle{\prod_{\alpha \in R_{+}}} \langle \alpha , \cdot \rangle^{c_{\alpha}}$ for $x, a \in V$, $p \in \mathcal{O} (T^{*} V_{reg})$. Then the classical Olshanetsky-Perelomov $L^{cl}$ is obtained via
\begin{align*}
L^{cl} = \theta_{c} \left( \bar{L}^{cl} \right).
\end{align*}
\end{proposition}
\begin{proof}
First we rewrite $\bar{L}^{cl}$  as follows
\begin{align*}
\bar{L}^{cl} =  \sum_{j} p_{j}^{2} - \sum_{\alpha \in R_{+}} \dfrac{2 c_{\alpha} }{\langle \alpha , \cdot \rangle } p_{\alpha}
\end{align*}
Next we apply $\theta_{c}$ to $\bar{L}^{cl}$ which yields
\begin{align*}
\theta_{c}  \left( \bar{L}^{cl} \right) &= \theta_{c} \left(\sum_{j} p_{j}^{2} \right) -  \sum_{\alpha \in R_{+}} \dfrac{2 c_{\alpha} }{\langle \alpha , \cdot \rangle } \theta_{c}(p_{\alpha}) \\
&= \sum_{j} \left(\left( p_{j}+\partial_{j} \log \delta_{c}\right)\left(  p_{j}+\partial_{j} \log \delta_{c}\right) \right) -  \sum_{\alpha \in R_{+}} \dfrac{2 c_{\alpha} }{\langle \alpha , \cdot \rangle } \left( p_{\alpha} + \partial_{\alpha} \log \delta_{c} \right).
\end{align*}
Expanding the brackets we get
\begin{align}
\theta_{c}  \left( \bar{L}^{cl} \right) &= \sum_{j} \left( p_{j}^{2} +  p_{j} \partial_{j} \log \delta_{c} +  \partial_{j} \log \delta_{c} p_{j}+ \partial_{j} \log \delta_{c}\partial_{j} \log \delta_{c} \right)-   \sum_{\alpha \in R_{+}} \dfrac{2 c_{\alpha} }{\langle \alpha , \cdot \rangle }p_{\alpha}  -   \sum_{\alpha \in R_{+}} \dfrac{2 c_{\alpha} }{\langle \alpha , \cdot \rangle } \partial_{\alpha} \log \delta_{c} . \label{etingofclass}
\end{align}
Recall that the logarithm of a product is equal to the sum of the logarithms. Hence, we have
\begin{align*}
\sum_{j} \partial_{j} \log \delta_{c} = \sum_{\alpha \in R_{+}} \dfrac{c_{\alpha}\sum_{j} \alpha_{j}}{\langle \alpha , \cdot\rangle}.
\end{align*}
Thus, we get
\begin{align*}
\theta_{c}  \left( \bar{L}^{cl} \right) &= \sum_{j}   p_{j}^{2} +  \cancel{\sum_{j}   p_{j}  \sum_{\alpha \in R_{+}} \dfrac{c_{\alpha}\sum_{j} \alpha_{j}}{\langle \alpha , \cdot\rangle}} + \cancel{\sum_{\alpha \in R_{+}} \dfrac{c_{\alpha}\sum_{j} \alpha_{j}}{\langle \alpha , \cdot\rangle} \sum_{j}   p_{j} }\\
& \quad \quad + \sum_{\alpha \in R_{+}} \dfrac{c_{\alpha}\sum_{j} \alpha_{j}}{\langle \alpha , \cdot\rangle} \sum_{\beta \in R_{+}} \dfrac{c_{\beta}\sum_{j} \beta_{j}}{\langle \beta , \cdot\rangle}  - \cancel{\sum_{\alpha \in R_{+}} \dfrac{2 c_{\alpha} }{\langle \alpha , \cdot \rangle }p_{\alpha} } - \sum_{\alpha \in R_{+}} \dfrac{2 c_{\alpha} }{\langle \alpha , \cdot \rangle }  \sum_{\beta \in R_{+}} \dfrac{c_{\beta}\langle \beta , \beta \rangle}{\langle \beta , \cdot\rangle}  = \\
&= \sum_{j}   p_{j}^{2}  - \sum_{\alpha \in R_{+}} \dfrac{c_{\alpha}^{2} \langle \alpha, \alpha \rangle }{\langle \alpha , \cdot \rangle^{2} } = L^{cl}.
\end{align*}
\end{proof}

\par Analogously to the quantum setting, we use Chevalley's theorem in order to extract a classical integrable sense from the classical Olshanetsky-Perelomov operator.

\begin{theorem}
Suppose by the Chevalley theorem that $\mathbb{C}[V]^{W} \cong \mathbb{C}[p_{1}, \dots, p_{n}]$ with $p_{1}, \dots, p_{n}$ being homogeneous polynomial functions of degrees $d_{1} \leq \dots \leq d_{n}$. Then the set of operators $L^{cl}_{i} = \rescl (\theta (p_{i} (D^{cl}_{e_{1}}, \dots, D^{cl}_{e_{n}} )))$ forms a classical integrable system.
\end{theorem}
\begin{proof}
The proof follows from Chevalley's Theorem \ref{chevalley} and the Poisson commutativity of the classical Dunkl operators as in Theorem \ref{poisson_dunkl}.
\end{proof}

\par Finally, we can confirm the classical integrability of the rational Calogero-Moser system. We summarize this statement below.
\begin{theorem}
Let $R = A_{n-1}$. Then $L^{cl} = H$, where $H$ is the Hamiltonian of the classical rational Calogero-Moser system \eqref{CM}, and the system is integrable.
\end{theorem}

\begin{proof}
Once again, the corresponding group $W$ of the root system of type $A_{n-1}$ is the symmetric group $\mathcal{S}_{n}$. Thus, $c_{\alpha} = k$. The integrability is clear from the Poisson commutativity of the classical Dunkl operators.
\end{proof} 

\section*{Conclusion}
\indent \par We conclude this paper with several remarks. While the integrability results for both the classical and the quantum Calogero-Moser systems are familiar, there are no works available at present to the best knowledge of the author with the explicit computations for any cases where $n>2$, most likely due to the complexity and tediousness of the calculations involving a large number of indeces. This in turn is overcome with by the ability of the software \emph{Mathematica} to manipulate expressions symbolically. That being said, it is in the best hopes of the author that more members of the mathematical community worldwide would become interested and inspired to utilize either \emph{Mathematica} or other symbolic packages, for instance, \texttt{python}'s package \texttt{SymPy}. One can take advantage of many features, including, but not limited to non-commutative multiplication, which is essential for manipulations in the quantum setting.
\par Finally, a short remark on potential further studies. The most logical generalization and continuation would be the explicit study of the trigonometric Calogero-Moser system, which is done again via the Lax pair formalism \cite{Rujsenaars} or by means of the Dunkl-Cherednik operators \cite{Cherednik2}. Even further, one can investigate the integrability of the model with potential described by the Weierstrass $\wp$-function via Elliptic Dunkl operators \cite{Cherednik3}.

\appendix
\section{Appendices} \label{app}
\subsection{Computations for the Classical Rational Case.}\label{appRational}
\begin{figure}[H]
\includegraphics[scale=0.8]{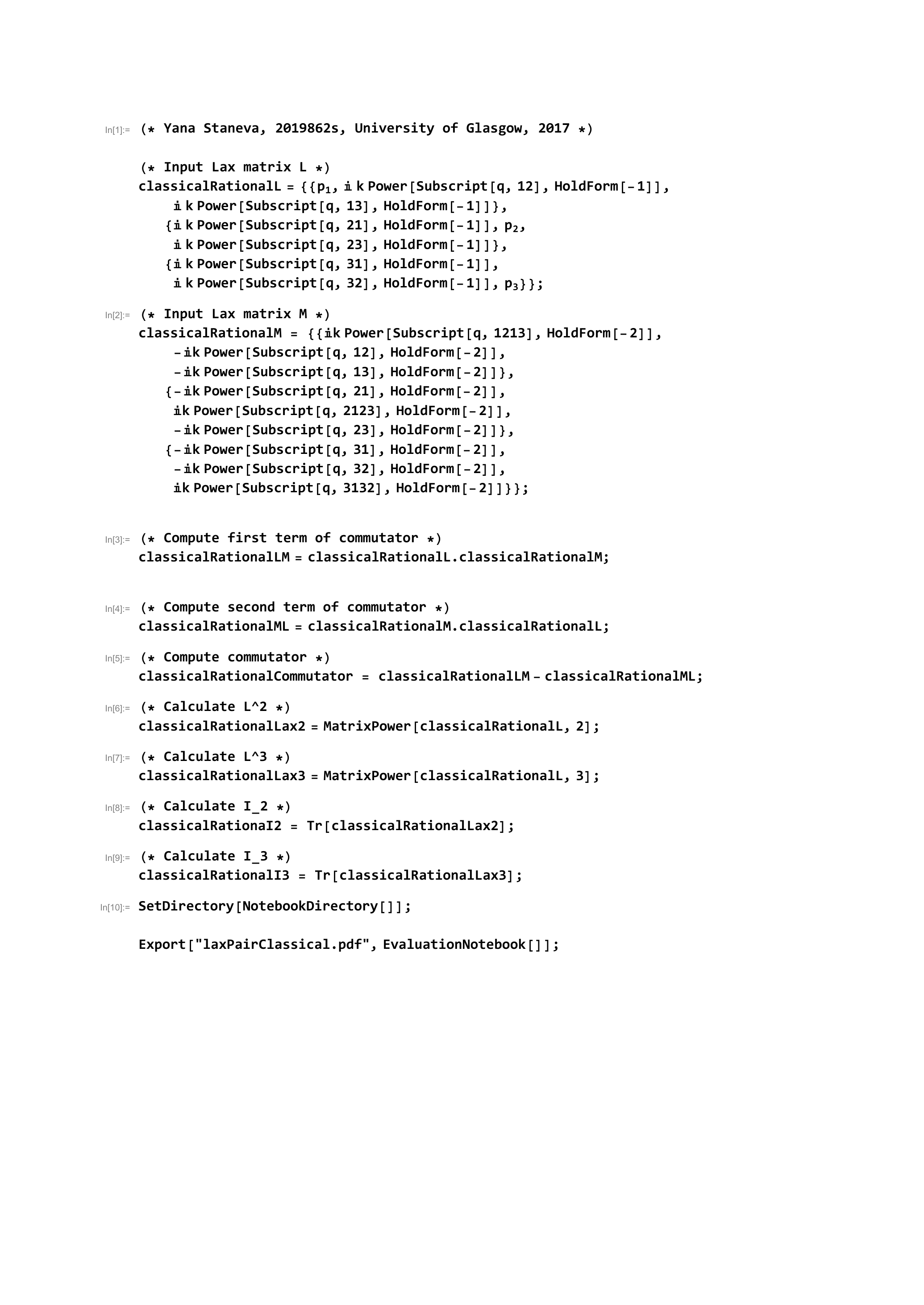}
\end{figure}
\subsection{Code for the Quantum Rational Case.}\label{appQuantumRational}
\begin{figure}[H]
\includegraphics[scale=0.8]{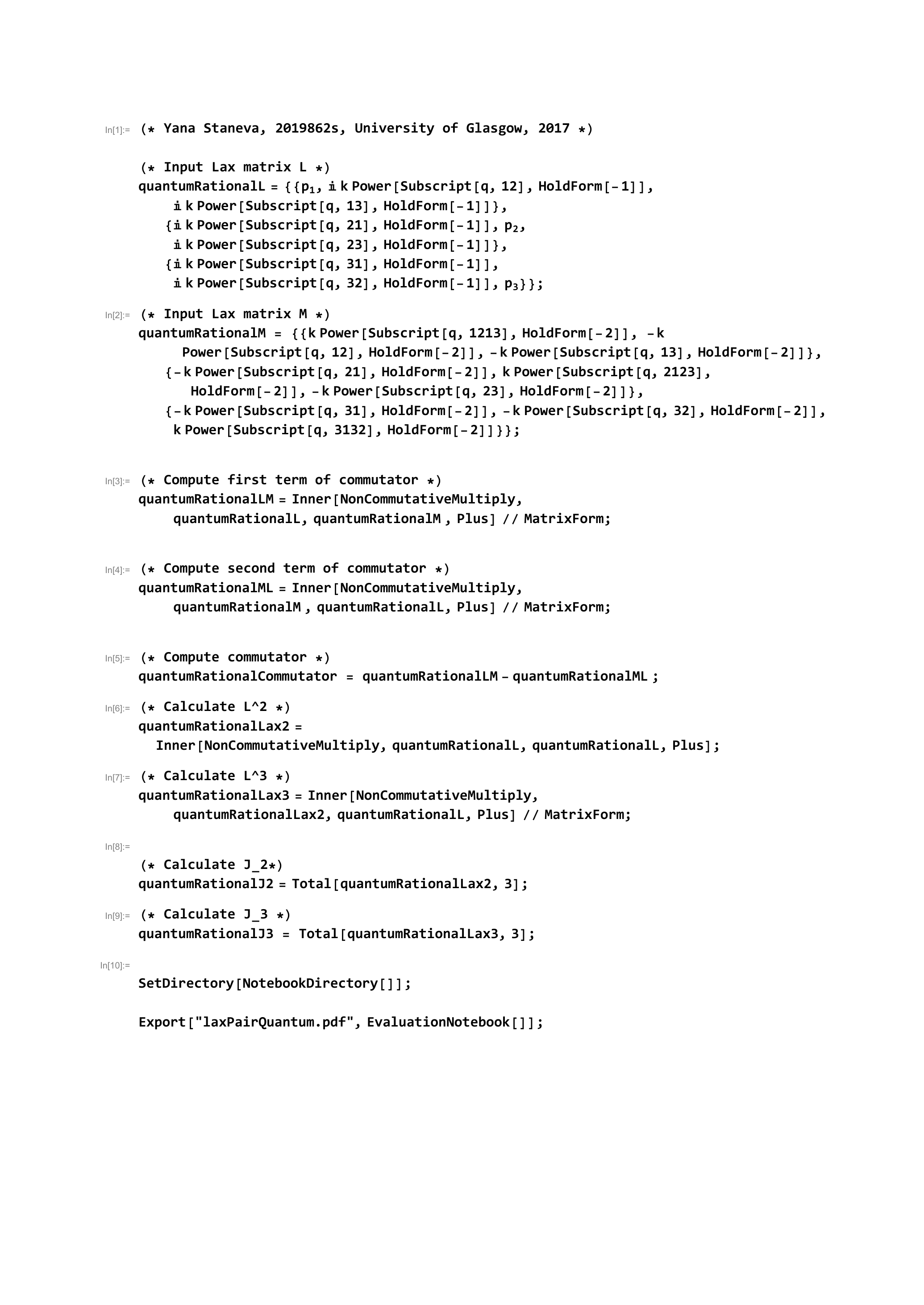}
\end{figure}
\newpage

\bibliography{Yana_Staneva_Explicit_Computations_for_the_Classical_and_Quantum_Integrability_of_the_3-Dimensional_Rational_Calogero-Moser_System} 
\bibliographystyle{amsplain}
\end{document}